\documentclass[10pt,twocolumn,twoside]{IEEEtran}

\IEEEoverridecommandlockouts                              

\usepackage{import}
\usepackage{graphics}
\usepackage{graphicx}
\graphicspath{{figures/}}
\usepackage[hidelinks]{hyperref}

\renewcommand{\footnoterule}{%
  \hspace{3pt} \hrule width 0.4\textwidth height 0.5pt
  \kern 2pt
}

\usepackage{flushend}

\usepackage{amsmath,mathtools,empheq}
\usepackage{amssymb,mathrsfs,siunitx}
\usepackage{accents} 

\usepackage{letltxmacro}
\LetLtxMacro\orgvdots\vdots
\LetLtxMacro\orgddots\ddots

\makeatletter
\DeclareRobustCommand\vdots{%
  \mathpalette\@vdots{}%
}
\newcommand*{\@vdots}[2]{%
  \sbox0{$#1\cdotp\cdotp\cdotp\m@th$}%
  \sbox2{$#1.\m@th$}%
  \vbox{%
    \dimen@=\wd0 %
    \advance\dimen@ -3\ht2 %
    \kern.5\dimen@
    \dimen@=\wd2 %
    \advance\dimen@ -\ht2 %
    \dimen2=\wd0 %
    \advance\dimen2 -\dimen@
    \vbox to \dimen2{%
      \offinterlineskip
      \copy2 \vfill\copy2 \vfill\copy2 %
    }%
  }%
}
\DeclareRobustCommand\ddots{%
  \mathinner{%
    \mathpalette\@ddots{}%
    \mkern\thinmuskip
  }%
}
\newcommand*{\@ddots}[2]{%
  \sbox0{$#1\cdotp\cdotp\cdotp\m@th$}%
  \sbox2{$#1.\m@th$}%
  \vbox{%
    \dimen@=\wd0 %
    \advance\dimen@ -3\ht2 %
    \kern.5\dimen@
    \dimen@=\wd2 %
    \advance\dimen@ -\ht2 %
    \dimen2=\wd0 %
    \advance\dimen2 -\dimen@
    \vbox to \dimen2{%
      \offinterlineskip
      \hbox{$#1\mathpunct{.}\m@th$}%
      \vfill
      \hbox{$#1\mathpunct{\kern\wd2}\mathpunct{.}\m@th$}%
      \vfill
      \hbox{$#1\mathpunct{\kern\wd2}\mathpunct{\kern\wd2}\mathpunct{.}\m@th$}%
    }%
  }%
}
\makeatother

\newcommand\ceil[1]{\left \lceil #1 \right \rceil}

\usepackage{amsthm}

\usepackage[usenames,dvipsnames]{xcolor}
\usepackage{color}
\definecolor{p2color}{HTML}{D95319}
\definecolor{pscolor}{HTML}{77AC30}
\definecolor{hpscolor}{HTML}{7E2F8E}
\definecolor{xblue}{HTML}{193fd9}

\usepackage{tikz-cd}

\usepackage[customcolors]{hf-tikz}

\tikzset{
  style dandelion/.style={
    opacity=0.8, set fill color= Dandelion  , 
    set border color=Dandelion,
  },
  style limeGreen/.style={
    opacity=0.8, set fill color= LimeGreen  , 
    set border color=LimeGreen,
  },
  style green/.style={
    opacity=0.7, set fill color= green!50!lime!60  , 
    set border color=white,
  },
  style cyan/.style={
    opacity=0.6, set fill color=cyan!90!blue!60,
    set border color=cyan!90!blue!60,
  },
  style orange/.style={
   opacity=0.6, set fill color=orange!80!red!60,
    set border color=orange!80!red!60,
  },
  hor/.style={
    above left offset={-0.15,0.31},
    below right offset={0.15,-0.125},
    #1
  },
  ver/.style={
    above left offset={-0.1,0.3},
    below right offset={0.15,-0.15},
    #1
  },
  blk/.style={
    above left offset={-0.12,0.4},
    below right offset={0.12,-0.1},
    #1
  }
}
\usepackage[export]{adjustbox} 
\usepackage[caption=true,format=hang,margin=1em]{subfig}
\usepackage{algorithm}
\usepackage{algpseudocode} 
\usepackage{nicematrix}
\usepackage{arydshln}
\usepackage{dashrule}

\usepackage[normalem]{ulem}

\usepackage{multirow}
\usepackage{booktabs}

\let\labelindent\relax
\usepackage{enumitem}
\renewcommand{\theenumi}{\rm{(\roman{enumi})}}

\newcommand{\btitle}[1]{\mbox{}{\bf{(#1).}}}
\newtheorem{theorem}{Theorem}[section]
\newtheorem{assumption}{Assumption}
\newtheorem{proposition}[theorem]{Proposition}
\newtheorem{lemma}[theorem]{Lemma}
\newtheorem{definition}[theorem]{Definition}

\newtheorem{remark}{Remark}[theorem]

\newcommand{\real}{\ensuremath{\mathbb{R}}}
\newcommand{\realpos}{\ensuremath{\mathbb{R}_{>0}}}
\newcommand{\realnonneg}{\ensuremath{\mathbb{R}_{\ge 0}}}

\newcommand{\intgnonneg}{\ensuremath{\mathbb{Z}_{\ge 0}}}
\newcommand{\cplx}{\ensuremath{\mathbb{C}}}
\newcommand{\ind}{\ensuremath{\mathbb{I}}}

\newcommand{\mb}[1]{\mathbf{ #1 }}
\newcommand{\norm}[1]{\left\Vert #1 \right\Vert}

\DeclareMathOperator{\col}{\mathrm{col}}
\DeclareMathOperator{\rank}{\mathrm{rank}}
\DeclareMathOperator{\spann}{\mathrm{span}}
\DeclareMathOperator{\diag}{\mathrm{diag}}

\DeclareMathOperator{\spec}{\mathrm{spec}}

\newcommand{\setdef}[2]{\left\{#1 \mid #2\right\}}
\newcommand{\paren}[1]{\left(#1\right)}

\newcommand{\braces}[1]{\left\{#1\right\}}

\newcommand{\s}{\scriptscriptstyle}

\newcommand{\sigt}{\sigma(t)}

\DeclareMathOperator*{\dprime}{\prime \prime}
\newcommand{\dtilde}[1]{\widetilde{\raisebox{0pt}[0.85\height]{$\widetilde{#1}$}}}

\newcommand{\Ac}{\mathcal{A}}

\newcommand{\Ec}{\mathcal{E}}

\newcommand{\Gc}{\mathcal{G}}
\newcommand{\Hc}{\mathcal{H}}
\newcommand{\Ic}{\mathcal{I}}

\newcommand{\Kc}{\mathcal{K}}
\newcommand{\Lc}{\mathcal{L}}

\newcommand{\Nc}{\mathcal{N}}
\newcommand{\Oc}{\mathcal{O}}

\newcommand{\Qc}{\mathcal{Q}}
\newcommand{\Rc}{\mathcal{R}}
\newcommand{\Sc}{\mathcal{S}}
\newcommand{\Tc}{\mathcal{T}}

\newcommand{\Vc}{\mathcal{V}}

\newcommand{\Xc}{\mathcal{X}}
\newcommand{\Yc}{\mathcal{Y}}

\newcommand{\xdot}{\dot{\mb{x}}}
\newcommand{\x}{\mb{x}}
\newcommand{\y}{\mb{y}}

\newcommand{\A}{\mb{A}_{\sigma}}
\newcommand{\Ba}{\mb{B}_{\Ac}}

\newcommand{\Bapp}{\mb{B}_{\Ac^{\dprime}}}

\newcommand{\Barest}{\mb{B}_{\Ac^{\rm r}}}

\newcommand{\C}{\mb{C}_{\sigma}}

\newcommand{\xhatdot}{\dot{\hat{\mb{x}}}}
\newcommand{\xhat}{\hat{\mb{x}}}
\newcommand{\yhat}{\hat{\mb{y}}}
\newcommand{\Hobs}{\mb{H}_{\sigma}}

\newcommand{\edot}{\dot{\mb{e}}}
\newcommand{\e}{\mb{e}}
\newcommand{\res}{\mb{r}}

\newcommand\nullH{\Hc^{\raisebox{0pt}{\scriptsize $\mathfrak{0}$}}}
\newcommand\alterH{\Hc^{\raisebox{0pt}{\scriptsize $\mathfrak{1}$}}}

\newcommand{\pstar}{\boldsymbol{p}^{\boldsymbol{\star}}}

\newcommand{\ua}{\mb{u}_{\! \Ac}}

\newcommand{\uapp}{\mb{u}_{\! \Ac^{\dprime}}}
\newcommand{\uarest}{\mb{u}_{\! \Ac^{\rm r}}}

\newcommand{\pb}{\boldsymbol{p}}
\newcommand{\vb}{\boldsymbol{v}}
\newcommand{\ub}{\boldsymbol{u}}

\newcommand{\efrak}{\mathfrak{e}}
\newcommand{\bfrak}{\mathfrak{b}}

\newcommand{\adj}{\mathsf{A}}
\newcommand{\lap}{\mathsf{L}}
\newcommand{\inc}{\mathsf{B}}

\newcommand{\onehop}[1]{\Nc^{\s #1(1)}_{\sigma}}
\newcommand{\twohop}[1]{\Nc^{\s #1(2)}_{\sigma}}
\newcommand{\khop}[1]{\Nc^{\s #1(k)}_{\sigma}}

\newcommand{\graphx}[1]{\Gc^{#1^{\prime}}_{\sigt}=(\Vc^{\,#1^{\prime}}_{\sigma} , \Ec^{\,#1^{\prime}}_{\sigma})}
\newcommand{\graphxx}[1]{\Gc^{#1^{\dprime}}_{\sigt}=(\Vc^{\,#1^{\dprime}}_{\sigma} , \Ec^{\,#1^{\dprime}}_{\sigma})}

\newcommand{\graphind}[1]{\bar{\Gc}^{#1^{\prime}}_{\sigt}=(\Vc^{\,#1^{\prime}}_{\sigma} , \bar{\Ec}^{\,#1^{\prime}}_{\sigma})}

\newcommand{\ones}{\mathbf{1}}
\newcommand{\zeros}{\mathbf{0}}

\newcommand\rev[1]{{\color{xblue} #1}}
\renewcommand\rev[1]{{#1}}
\newcommand\revv[1]{{\color{xblue} #1}} 
\renewcommand\revv[1]{{#1}}

\newcommand\submittedtext{%
  \footnotesize 
  This work has been submitted to the IEEE for possible publication. Copyright may be transferred without notice, after which this version may no longer be accessible.}

\newcommand\submittednotice{%
\begin{tikzpicture}[remember picture,overlay]
\node[anchor=south,yshift=3pt] at (current page.south) {\fbox{\parbox{\dimexpr0.8\textwidth-\fboxsep-\fboxrule\relax}{\submittedtext}}};
\end{tikzpicture}%
}

\newif\ifshort
\shortfalse

\title{\LARGE \bf
Distributed Detection of Adversarial Attacks for Resilient Cooperation of Multi-Robot Systems 
with Intermittent Communication
}

\author{Rayan Bahrami and Hamidreza Jafarnejadsani
\vspace{-1em}
\thanks{This work was supported by the National Science Foundation under Award No. 2137753.} 
\thanks{The authors are with the Department of Mechanical Engineering, Stevens Institute of Technology, Hoboken, NJ 07030, USA,
        {\tt\small \{mbahrami,hjafarne\}@stevens.edu}.
        }
}

\begin{document}

\maketitle
\thispagestyle{empty}
\pagestyle{empty}

\begin{abstract}
This paper concerns the consensus and formation of a network of mobile autonomous agents in adversarial settings where a group of malicious (compromised) agents are subject to deception attacks. In addition, the communication network is arbitrarily time-varying and subject to intermittent connections, possibly imposed by denial-of-service (DoS) attacks. 
We provide explicit bounds for network connectivity in an integral sense, enabling the characterization of the system's resilience to specific classes of adversarial attacks. We also show that under the condition of connectivity in an integral sense uniformly in time, the system is finite-gain $\Lc$ stable and uniformly exponentially fast consensus and formation are achievable, provided malicious agents are detected and isolated from the network.
We present a distributed and reconfigurable framework with theoretical guarantees for detecting malicious agents, allowing for the resilient cooperation of the remaining cooperative agents. Simulation studies are provided to illustrate the theoretical findings.
\end{abstract}
%
\submittednotice
\vspace{-1.4em}
\section{Introduction}
\label{sec:introduction}
\IEEEPARstart{R}esilient network-enabled cooperation of mobile autonomous agents (e.g., Unmanned Aerial/Ground Vehicles) is central to many safety-critical and time-critical applications such as surveillance, monitoring, and motion and time coordination \cite{ren2007information,martinez2007motion,cichella2015cooperative,kantaros2018distributed,yu2021resilient}. 
A common premise in the prevalent attack detection and attack-resilient cooperation approaches for multi-agent systems is that specific measures of network connectivity are almost constantly maintained throughout time \cite{pasqualetti2011consensus,leblanc2013resilient,dibaji2015consensus,dibaji2017resilient}. This, however, is a challenging assumption in the case of mobile autonomous agents whose mobility and limited communication capabilities give rise to an \emph{ad hoc} and \emph{intermittent} network connectivity \cite{kantaros2018distributed,yu2021resilient}.
Aligned with recent studies \cite{yu2021resilient,saulnier2017resilient,santilli2021dynamic}, this paper extends the previous results by considering distributed attack detection and resilient cooperation of a class of multi-agent systems in the adversarial settings where intermittent (time-varying) communication and a group of malicious agents \rev{\emph{concurrently}} render unreliable information exchange.

\textbf{Related work}.
The vulnerabilities of networked systems to adversarial attacks have been reported in
\cite{amin2009safe,teixeira2015secure}. 
These attacks can be classified as \emph{deception}
attacks, targeting the integrity/trustworthiness of transmitted data, and \emph{denial-of-service} (DoS) attacks, targeting the data availability upon demand \cite{amin2009safe}.
Considerable effort has been devoted to the characterization of resilience to adversarial attacks on networked systems as well as to the attack detection mechanisms. 

In terms of resilience,
it has been shown that the resilience to specific types of adversarial attacks and/or a group of malicious agents in a network can be characterized through certain connectivity-related properties of the underlying communication network.
Particularly for coordination/consensus of multi-agent systems with first-order dynamics, early studies have characterized the worst-case bounds for the total number of malicious (non-cooperative) agents that can be detected and identified in a given static communication network with a certain degree of vertex-connectivity \cite{pasqualetti2011consensus}.  
In \cite{leblanc2013resilient,zhang2015notion}, a connectivity-related property known as graph $r$-robustness was proposed to quantify resilience to a certain number of malicious agents in consensus dynamics over static networks. 
The results have also been extended to the cases of agents with higher-order dynamics \cite{dibaji2015consensus,dibaji2017resilient,rezaee2020almost}.
See also the survey \cite{pirani2023graph}. 
However, only a few studies have recently considered resilient consensus beyond static communication networks, \rev{particularly when the network is subject to both deception and DoS attacks.}
Few examples are \cite{saulnier2017resilient,rezaee2020almost,yu2021resilient,usevitch2019resilient}.

\rev{In terms of detection of and defense mechanisms against deception attacks, the problem is inherently challenging.}
First, a priori knowledge of the system dynamics can be exploited to design sophisticated attacks that are \emph{stealthy} to the conventional anomaly detectors \cite{teixeira2015secure}. Examples of such attacks 
are zero-dynamics attack (ZDA) \cite{pasqualetti2011consensus,mao2020novel,bahrami2022detection,bahrami2021privacy}, covert attack \cite{gallo2020distributed,bahrami2021privacy}, and replay attack \cite{rezaee2020almost}. 
Second, the body of effective observer-based frameworks developed for distributed detection of stealthy attacks in spatially invariant systems such as power networks
\cite{pasqualetti2013attack,barboni2020detection,gallo2020distributed} are premised on having \emph{a priori} known and more often static communication topology which is not the case in spatially distributed mobile agents \cite{martinez2007motion,kia2019tutorial}.
Alternatively, \cite{bonczek2022detection} proposed a framework to detect communication and sensor attacks that are stealthy to other residual-based methods. In \cite{mustafa2022adversary}, a control barrier function (CBF) approach was proposed for safety and objective specifications serving as metrics for the identification of adversarial agents and resilient control of multi-agent systems. \rev{This paper extends observer-based approaches \cite{pasqualetti2011consensus,mao2020novel} by relaxing their dependency on point-wise-in-time network connectivity/robustness and quantifying resilience to concurrent adversarial attacks. 
}

\textbf{Statement of contributions}. 
We consider coordination (consensus) and cooperation (consensus-based formation) of multi-agent systems with second-order dynamics in adversarial settings.
We consider the adversarial settings where a group of malicious (adversarial) agents introduces a class of \emph{deception} attacks to disrupt a normal operation of the system. Moreover, the communication network is subject to intermittent connections possibly caused by \emph{DoS} attacks, rendering the information exchange unreliable.
We will show if the arbitrarily time-varying (switching) communication network maintains its connectivity only in an integral sense, uniformly in time, the following results are guaranteed:

In Section \ref{sec:net_res_and_stability}, we  \rev{characterize and provide explicit bounds for the network resilience to both intermittent and permanent disconnections. The former is relevant to \emph{DoS} attacks, and the latter is relevant to \emph{deception} attacks.} 
We also provide explicit bounds for uniformly exponentially fast convergence of the multi-agent systems \rev{in the presence of a class of \emph{DoS} attacks} as well as for their bounded-input-bounded-output (BIBO) stability in the presence of a class of \emph{deception} attacks. 
\rev{Compared to the previous results \cite{yu2021resilient,dibaji2017resilient,usevitch2019resilient}, the network resilience is quantified explicitly based on algebraic connectivity in an integral sense, and the connectivity and stability analyses for both types of attacks are in the continuous-time domain.}

In Section \ref{sec:observer_design}, we \rev{characterize the system vulnerability to a class of stealthy deception attacks, based on zero dynamics of the switched systems, and} provide explicit worst-case bounds on the number of malicious agents subject to deception attacks that can be detected in a given network. 
\rev{Compared to the previous results \cite{pasqualetti2011consensus,dibaji2017resilient}, we show some of these well-known bounds can be improved, provided some extra information on the local dynamics is available only in an integral sense.} We then present a distributed and reconfigurable framework with theoretical guarantees for the distributed detection of malicious agents introducing deception attacks. \rev{Compared to centralized frameworks \cite{mao2020novel}, our framework relies solely on locally available information in an integral sense, making it well-suited for mobile agent applications subject to intermittent connectivity.}

Additionally, \rev{in Section \ref{Sec:resilient_cooperation}}, we present an algorithmic framework for detaching from the detected set of malicious agents, and for achieving resilient coordination and cooperation.

\section{Problem Formulation}\label{sec:problem_formulation}
\subsection{Preliminaries}
\textbf{{Notation}.} 
We use  $ \real $, $ \realpos $, $ \realnonneg $, $\intgnonneg$, and $\cplx$ to denote the set of reals, positive reals, nonnegative reals,  nonnegative integers, and complex numbers, respectively. 
$\ones_n$, $ \zeros_n$, $ {I}_n $ and $ \zeros_{n \times m} $ stand\footnote{We may omit the subscripts when it is clear from the context.} for the $n$-vector of all ones, the $n$-vector of all zeros, the identity $ n $-by-$ n $ matrix, and the $ n $-by-$ m $ zero matrix, respectively. 
We use $\efrak^{\,i}_{n}$ to denote the $i$-th canonical vector in $\real^{n}$, and 
$\norm{\cdot}$ (resp. $\norm{\cdot}_{\infty}$) to denote the Euclidean (resp. infinity) norm of vectors and induced norm of matrices.
In addition, for any piecewise continuous, real-valued Lebesgue measurable signal $x(t) \in \real^{n}$, we use $\norm{{(x)}_{T_d}}_{{\Lc_p}}$, were $ 1 \leq p  $ $\leq \infty$ and $T_d \in [0, \infty)$, to denote the ${\Lc_p}$ norm of its truncation signal defined as ${(x)}_{T_d} = x(t)$ for $0\leq t \leq T_d,$ and ${(x)}_{T_d} = \zeros$ if $ T_d < t$.
The extended space ${\Lc_{pe}}$ 
is defined as
${\Lc_{pe}} = \braces{x(t) \mid (x)_{T_d} \in \Lc_{p}, \, \forall \, T_d \in [0,\, \infty) } $. 
The notations $\spec(\cdot)$ and $\lambda_{i}(\cdot)$ denote the spectrum of a matrix in the ascending order by magnitude and its $i$-th eigenvalue, respectively. 
We use $\col(\cdot)$ and $\diag(\cdot)$ to denote the column and diagonal concatenation of vectors or matrices.
Finally, for any set $\Sc$, $ |\Sc| $ denotes its cardinality.

\textbf{{Communication topology}.} The communication network (topology) of $N\geq3$ mobile agents (robots), indexed by the set $\Vc=\{1,\dots, N\}$, is described by a finite\footnote{The set of possible communication graphs, $\Qc$, is finite by $2^{\s \binom{N}{2}}$ possible cases because an undirected graph with $N$ nodes at most is complete with $\binom{N}{2} = N(N-1)/2$ edges 
\cite[Ch. 1, P. 11]{west2001introduction}
.} collection of undirected graphs $ \Gc_{\sigt}=(\Vc,\Ec_{\sigt}) $, where the edge set $\Ec_{\sigt} \subset \Vc \times \Vc$ denotes the (potentially ad hoc or intermittent) communication links. An edge $(i,j)\in \Ec_{\sigt}$ if and only if the $i$-th and $j$-th agents are adjacent neighbors exchanging information in an active communication mode $\sigt \in \{1,2,\dots,\mb{q} \} =: \Qc, \, \mb{q} \in \intgnonneg \setminus \{0\}$. The switching (interchangeably time-varying) network, $\Gc_{\sigt}, \, \sigt \in \Qc$, with a piecewise constant and right-continuous switching signal $\sigt:  \realnonneg \rightarrow  \Qc $, allows for modeling the inter-agent time-varying communication topology with intermittent and lossy datalinks.
The agents' communication is further encoded into a symmetric adjacency matrix $\adj_{\sigt} \coloneqq  [a^{\sigt}_{ij}]\in \mathbb{R}^{ N \times N}_{\geq0}$ such that $a^{\sigt}_{ij} = a^{\sigt}_{ji} = 1 $ if an edge $ (i,j) \in \Ec_{\sigt}$, and $a^{\sigt}_{ij}= a^{\sigt}_{ji} = 0$, otherwise.
The Laplacian matrix $\lap_{\sigt} \coloneqq  [l^{\sigt}_{ij}] \in \real^{N \times N}$ is defined as  $ l^{\sigt}_{ii} $ $= \sum_{j=1}^{N}a^{\sigt}_{ij}$ and $ l^{\sigt}_{ij} = -a^{\sigt}_{ij}$ if $i \neq j$. 
We let $\inc_{\sigt} \in \real^{N \times E} $, where $ E = \max_{\sigt \in \Qc}|\Ec_{\sigt}| \leq N(N-1)/2$, denote the (oriented) incidence matrix that satisfies $\lap_{\sigt} = \inc^{}_{\sigt}\inc^{\top}_{\sigt}$.
Also, an undirected graph $\Gc_{\sigma(t)}$ is connected at \emph{a given time instant} $t=t' \in \realnonneg$ if and only if the second-smallest eigenvalue of the associated Laplacian matrix, referred to as \emph{algebraic connectivity}, holds $ \lambda_{2}(\lap_{\sigma(t')})>0 $. $\norm{\lap_{\sigt}}\leq |\Vc|=N$.
A graph is called $\boldsymbol{\kappa}$-(vertex)-connected at \emph{a given time instant} $t=t' \in \realnonneg$ if $ \boldsymbol{\kappa} \leq \boldsymbol{\kappa}(\Gc_{\sigma(t')}) \in \realpos$, where $ \boldsymbol{\kappa}(\Gc_{\sigma(t')}) $ is the minimum number of vertices that make a vertex cutset \rev{\cite[Ch. 3 and 13]{godsil2001algebraic}}.


We make the convention that, in any active mode $\sigma(t) \in \Qc$ ($\sigma$ in short), $\khop{i} \subseteq \Vc $ denotes the set of $k$-hop neighbors of the agent $ i \in \Vc$, that is for $ j \in \khop{i}$ there exists a path of length $k $, where $ k \in \intgnonneg \setminus \{0\}$, in mode $\sigma$, between the agents $i$ and $j$. 
Accordingly, we let the subgraphs $\graphx{i}$, $\graphind{i}$, and $\graphxx{i}$, denote, resp., the $1$-hop proximity, the $1$-hop induced, and the $2$-hop proximity communication network of the $i$-th mobile agent with its $k$-hop neighbors, where $k\in \{1,2\}$, and
\noindent
\begin{subequations}\label{eq:local_graphs}
\begin{align}
   \hspace{-1ex} 
   \Vc^{\, i^{\prime}}_{\sigma} \!&= \{i\} \cup  \onehop{i},  \qquad \quad \label{eq:1prox_edge}
    {\Ec}^{\, i^{\prime}}_{\sigma} =  \{i\} \times  \onehop{i}  \subseteq \Ec_{\sigma},
    \\  \label{eq:inducedgraph_edge}
   \bar{\Ec}^{\, i^{\prime}}_{\sigma} \!&=   \big( \Vc^{\, i^{\prime}}_{\sigma} \times \Vc^{\, i^{\prime}}_{\sigma} \big) \cap    \Ec_{\sigma},
    \\ \label{eq:2prox_edge}
  \hspace{-1ex}   
  \Vc^{\, i^{\dprime}}_{\sigma} \!&= \Vc^{\, i^{\prime}}_{\sigma}  \! \cup  \twohop{i},  
  \,\,
    \Ec^{\, i^{\dprime}}_{\sigma} \!\! = \bar{\Ec}^{\, i^{\prime}}_{\sigma} \! \cup 
    \paren{
    \big( \onehop{i} \! \times \! \twohop{i} \big) \cap \Ec_{\sigma}}\!. 
\end{align}
\end{subequations}
\begin{figure*}[ht]
  \centering 
  \subfloat[10-agent network\label{fig:network_eg_}]{\small \centering
    \includegraphics[ width=.28\linewidth]{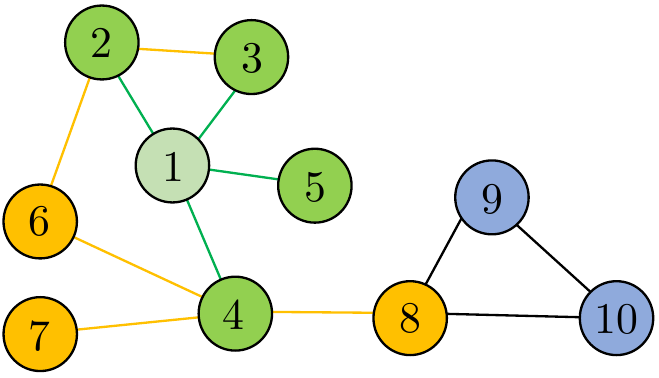} 
  }
  \quad
 \subfloat[Partitioned Laplacian matrix. \label{fig:network_Lap}]{\small \centering
    \includegraphics[width=.2\linewidth]{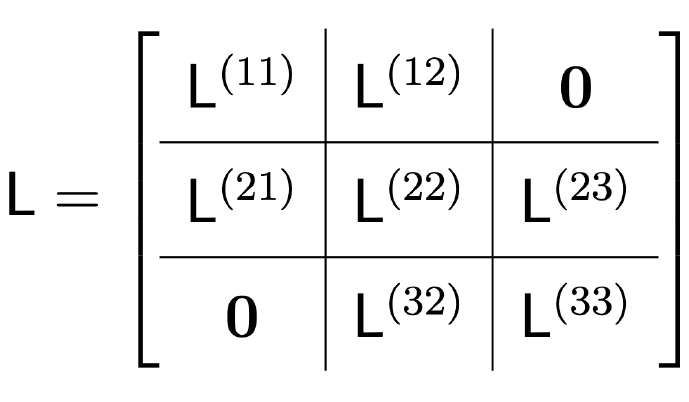}
  }
  \quad
  \subfloat[Partitioned Laplacian matrix from node $1$'s perspective. \label{fig:lap_1}]{\small \centering
    \includegraphics[width=.2\linewidth]{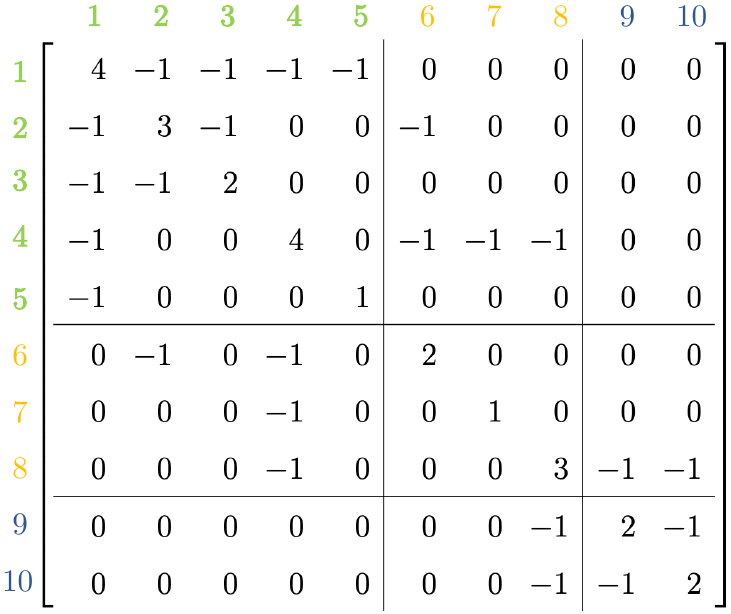}
  }
  \quad
  \subfloat[Partitioned Laplacian matrix from node $7$'s perspective. \label{fig:lap_7}]{\small \centering
    \includegraphics[width=.2\linewidth]{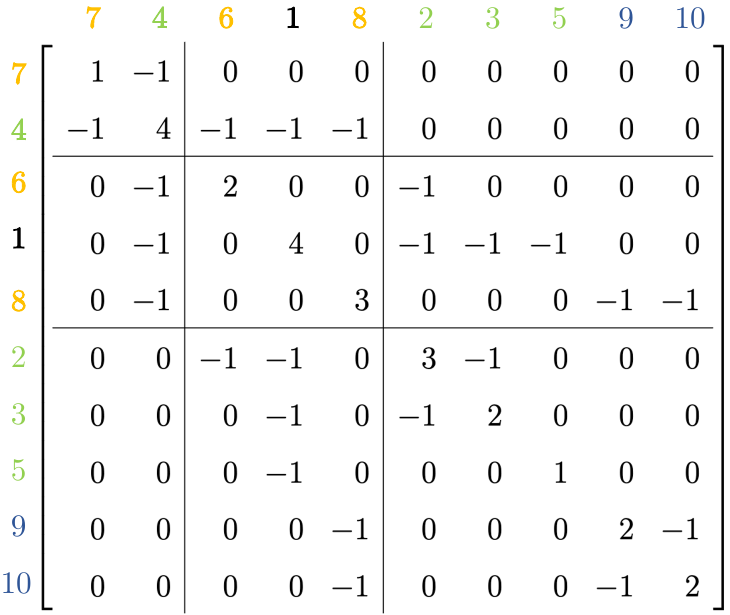}
  }
  \caption{\small Illustrative example of a 10-node communication network for the sub-graphs defined in \eqref{eq:local_graphs}-\eqref{eq:laplacian_decomposition}. (a) The undirected communication network $\Gc_{\sigt}=(\Vc, \Ec_{\sigt})$ illustrated in a mode $\sigma$.
  The set of {\color{Green}$1$-hop} (resp. {\color{YellowOrange}$2$-hop}) neighbors of node $1 \in \Vc$ is given by {\color{Green}$\onehop{1}=\{2,3,4,5\}$} $\big(\text{resp.} \; {\color{YellowOrange}\twohop{1} = \{5,7,8\}} \big) $. (b) Accordingly, the $1$-hop dynamics of agent 1 encompasses the $1$-hop Laplacian $\lap^{(11)}_{\sigma}= {\color{Green}\lap^{\prime}_{\sigma}} + {\color{YellowOrange}\widetilde{\lap}_{\sigma}} $, where ${\color{Green}\lap^{\prime}_{\sigma}}$ always has a {\color{Green}star-like structure} for any node since it encodes the connections with {\color{Green} immediate neighbors}, which can always be inferred by agent 1,
  and ${\color{YellowOrange} \widetilde{\lap}_{\sigma}}$ encodes the connection between the immediate (1-hop) neighbors such as {\color{YellowOrange}edge (2,3)} in this example. (b)-(d)  Proximity-based partitioned Laplacian matrix as defined in \eqref{eq:local_graphs}-\eqref{eq:laplacian_decomposition}.
  }
  \label{fig:network_eg}
\end{figure*}
Having defined the $k$-hop neighbors, the Laplacian matrix $\lap_\sigma$ can be partitioned according to the incoming flow of information to the agent $i\in \Vc$. Let $i\in \Vc$ be the first agent and accordingly $\Vc^{\, i^{\prime}}_{\sigma}$ come first, $\twohop{i}$ come second, and $\Vc \setminus \{ \Vc^{\, i^{\prime}}_{\sigma} \cup \twohop{i} \}$ come last, then $\lap_\sigma$ can be rewritten as 
\noindent
\begin{subequations}\label{eq:laplacian_decomposition}
{\small
\begin{align}
\hspace{-1ex}
\lap_{\sigma} \!=\! 
\left[
\begin{array}{c|c|c}
 \tikzmarkin[blk=style dandelion ]{col 2-L}
 \tikzmarkin[blk=style limeGreen ]{col 1-L}
 \lap^{(11)}_{\sigma} 
 \tikzmarkend[blk=style limeGreen ]{col 1-L}
  & \lap^{(12)}_{\sigma}  &   \zeros \!\!
 \\ \hline
  \lap^{(21)}_{\sigma} & 
  \tikzmarkin[blk=style cyan ]{col 3-L}
  \lap^{(22)}_{\sigma} 
  \tikzmarkend[blk=style dandelion]{col 2-L} & \lap^{(23)}_{\sigma} \!\!
  \\ \hline 
  \zeros & \lap^{(32)}_{\sigma} & \lap^{(33)}_{\sigma} \!\! \tikzmarkend[blk=style cyan ]{col 3-L}
 \end{array}
\right],  
\lap^{\dprime}_{\sigma} 
=
\left[
\begin{array}{l|l}
\tikzmarkin[blk=style dandelion ]{col 5-L}
\lap^{(11)}_{\sigma}  & \lap^{(12)}_{\sigma} \!\!
 \\ \hline
  \lap^{(21)}_{\sigma} & \lap^{\s (\dprime \setminus \prime)}_{\sigma} \! \! 
\tikzmarkend[blk=style dandelion]{col 5-L}  
 \end{array}
\right], 
\\ 
\lap^{(11)}_{\sigma} =
 \tikzmarkin[blk=style limeGreen ]{col 4-L}
 \lap^{\prime}_{\sigma}
 \tikzmarkend[blk=style limeGreen ]{col 4-L} 
+ \widetilde{\lap}^{}_{\sigma}, 
\hspace{6em} 
\begin{bmatrix}
\widetilde{\lap}^{}_{\sigma} & \lap^{(12)}_{\sigma}
\end{bmatrix} \ones = \zeros,
\\
\lap^{(22)}_{\sigma} = \lap^{\s (\dprime \setminus \prime)}_{\sigma} + {\dtilde{\lap}}^{}_{\sigma},
\hspace{5em} \!
\begin{bmatrix}
\dtilde{\lap}^{}_{\sigma} & \lap^{(23)}_{\sigma}
\end{bmatrix} \ones = \zeros, 
\end{align}
\normalsize}
\end{subequations}
in which $\lap^{\prime}_{\sigma}$ is the Laplacian matrix of the $1$-hop proximity graph $\graphx{i}$. $\widetilde{\lap}^{}_{\sigma}$ encodes the edge set $\bar{\Ec}^{i}_{\sigma} \setminus {\Ec}^{\, i^{\prime}}_{\sigma}$, that is the set of existing edges between the $1$-hop neighbors with one another, and the set of existing edges between the $1$-hop neighbors and the $2$-hop neighbors. $\lap^{\s (\dprime \setminus \prime)}_{\sigma}$ is the Laplacian matrix encoding the edge set ${\Ec}^{\, i^{\dprime}}_{\sigma} \setminus \bar{\Ec}^{\, i^{\prime}}_{\sigma}$ that is the connections of the $2$-hop neighbors with the $1$-hop neighbors, and $\dtilde{\lap}^{}_{\sigma}$ encodes the existing edges between the $2$-hop neighbors with one another, and those existing between the $2$-hop neighbors and the rest in higher hops, i.e., $\Vc \setminus \Vc^{\, i^{\dprime}}_{\sigma}$. Finally, $\lap^{\dprime}_{\sigma}$ is the Laplacian matrix associated with the $2$-hop proximity graph $\graphxx{i}$.
%
An illustrative example of the foregoing communication graphs is given in Fig. \ref{fig:network_eg}. 

\subsection{System dynamics}\label{sec:consenus_formation}
Consider a multi-agent system consisting of $ N \geq 3$ mobile agents (robots) with double-integrator dynamics as follows: 
\noindent
\begin{align}\label{eq:ol_sys}
\Sigma_{i}:
\left\{
\begin{array}{l}
\dot{\pb}_i(t) = \vb_i(t)\\ 
\dot{\vb}_i(t) = \ub_{i}(t) 
\end{array},
\right. 
\ \  i \in \mathcal{V} = \{1,\dots, N\} , 
\end{align}
in which $ \boldsymbol{p}_i(t) \in \real^{},$ and $ \vb_i(t)  \in \real^{} $ are the position and velocity states. $\ub_{i}(t)  \in \real^{} $ denotes\footnote{For brevity, we may omit the time argument, $t$, from expressions whenever possible for the remainder of this paper.} the control input of each mobile agent to be computed given local information exchange with its neighbors, $\onehop{i}$, over the unreliable communication network $\Gc_{\sigt}$.
We let an unknown subset of agents, denoted by $\Ac \subset \Vc$ and referred to as \emph{malicious} agents, update their control inputs $\ub_{i}, \, i \in \Ac, $ such that $\ub_{i} = \ub^{\rm{n}}_{i} + {\ub}^{\rm a}_{i} $ in \eqref{eq:ol_sys} where $\ub^{\rm{n}}_{i}$ is the \emph{normal} control input and $\ub^{\rm{a}}_{i}$ is an injected attack signal. We also refer to the rest of the agents, $ \Vc \setminus {\Ac} $, as \emph{cooperative} (or \emph{normal}) agents.
In this adversarial setting, the \emph{cooperative} (resp. \emph{malicious}) agents seek to achieve (resp. prevent) the cooperation objective that is defined as
\noindent
\begin{subequations}\label{eq:formation_consensus}
\begin{align} \label{eq:formation_consensus_pos}
\lim \limits_{t\rightarrow{\infty}} \left|\pb_i(t)-\pb_j(t) - \pstar_{ij} \right| &= \zeros,
  &\forall \; i,j \in \Vc \setminus \Ac ,
\\ \label{eq:formation_consensus_vel}
\lim \limits_{t\rightarrow{\infty}} \left|\vb_i(t) \right| &= \zeros, 
  &\forall \; i \in \Vc \setminus \Ac ,
\end{align}
\end{subequations}
where the predefined constants $\pstar_{ij} = \pstar_{i}-\pstar_{j}$ are the desired relative positions for any pair of mobile agents in the formation settings and $\pstar_{ij} = 0$ in the consensus settings. 

\rev{For agents in \eqref{eq:ol_sys} with the cooperation and adversary objectives \eqref{eq:formation_consensus}}, we consider the distributed control protocol:
%
\begin{align}\label{eq:ctrl_proto}
\hspace{-1ex}
\ub_i \!=\!  \ub^{\rm{n}}_{i} + \ub^{\rm{a}}_{i}, \;\; 
\ub^{\rm{n}}_{i} \!=\!
    -\alpha \!
    \hspace{-1ex}
    \sum_{ j \in \onehop{i}}
    \hspace{-1ex} 
    a^{\sigma}_{ij}(\pb_i-\pb_j -\pstar_{ij} ) 
    - \gamma \vb_i, 
\end{align}
which relies only on communication with $1$-hop neighbors $\onehop{i}$. 
Also, the constants $\alpha, \gamma \in \realpos$ are the control gains.

Given \eqref{eq:ol_sys} and \eqref{eq:ctrl_proto}, the network-level dynamics of the multi-agent system can be represented by a family of linear switched systems as follows:
\noindent
\begin{align}\label{eq:cl_sys}
\Sigma_{\sigt}:
\begin{bmatrix}
 \dot{\widetilde{\pb}} \\ \dot{\vb}  
\end{bmatrix}
&=
\begin{bmatrix}
 \zeros &  \phantom{-}I \\
 -\alpha \lap_{\sigma} & - \gamma I
\end{bmatrix}
\begin{bmatrix}
 \widetilde{\pb} \\ {\vb}  
\end{bmatrix}
+
\begin{bmatrix}
 \zeros \\ I_{\Ac}
\end{bmatrix}
\ua
\nonumber \\
 &=: \A \x + \Ba \ua, \qquad \x_0 = \x(t_0),
  \\ \label{eq:locally_measurements}
 \y^{\, i}_{\sigma} &= \col (\widetilde{\pb}_{j \in \Ic_i},\vb_i) =: \C^{\, i} \x,
\end{align}
where $\widetilde{\pb}=\pb-\pstar$ in $\x = \col(\widetilde{\pb},\vb)$ with $\pb \in \real^{N}$, $\pstar \in \real^{N}$, and $\vb \in  \real^{N}$ being, respectively, the stacked the position states, formation references ($\pstar = \zeros$ in consensus settings), and velocity states of all $N$ agents. Also, $ \lap_{\sigma} $ is the Laplacian matrix of the network $\Gc_{\sigt}$, encoding the communication links. 
$\ua = \col \paren{ {\ub}^{\rm a}_{{i}} }_{i \in \Ac} \in \real^{|\Ac|} $ and
$I_{\! \Ac} = \big[\efrak^{\s i_{1}}_{\s N}\; \efrak^{\s i_2}_{\s N}\; \dots\; \efrak^{\s i_{\s |\Ac|}}_{\s N}\big] \in \real^{N \times |\Ac|}$, where $\efrak^{\, i}_{N}$ specifies the input direction in $\real^{N}$, corresponding to the $i$-th malicious agent among $N$ agents. Finally, $\y^{\, i}_{\sigma}$ denotes the state measurements available for the $i$-th mobile agent consisting of the (relative) position states of a set of neighboring agents $\{i\} \cup \onehop{i} \subseteq \Ic_i \subseteq \Vc$ (where $\Ic_i $ will be determined later) and the velocity state $\vb_i$.

It is necessary to note that the nature of arbitrary switching modes $\sigt \in \Qc$, induced by the unreliability of network $\Gc_{\sigt}$, renders \emph{a priori} unknown system matrix $\A$ in \eqref{eq:cl_sys}. 
This imposes stability and observability challenges in distributed settings which will be addressed in Sections \ref{sec:net_res_and_stability} and \ref{sec:observer_design}. 

\subsection{Adversary model}\label{Sec:adversary_model}
We consider two classes of adversarial attacks, namely \emph{deception} attacks and \emph{denial-of-service} (DoS) attacks.

\textbf{Deception attacks.} In this model, a set of malicious agents $\Ac \subset \Vc$, as described in Section 
\ref{sec:consenus_formation}, inject some undesirable data $0\neq {\ub}^{\rm a}_{i}(t) \in \Lc_{p e},\; \forall\, i \in \Ac,\; \forall \, t \in [t^{\rm a}_i,\, \infty)$, where $t^{\rm a}_i \in \realnonneg $ is the activation time instant in \eqref{eq:ctrl_proto}. 
Among the well-studied deception attacks including data injection attack \cite{dibaji2015consensus,leblanc2013resilient},
zero-dynamics attacks (ZDA) \cite{pasqualetti2011consensus,mao2020novel,bahrami2022detection}, \revv{covert attack \cite{gallo2020distributed,bahrami2021privacy}}, replay attack \cite{rezaee2020almost}, \revv{and Byzantine attacks \cite{leblanc2013resilient}}, our analysis covers the first two models.
Similar to \cite{leblanc2013resilient,dibaji2017resilient}, the worst-case upper bounds on the number of malicious agents in the network are parameterized as follows:

\begin{definition}\btitle{$F$-local and $F$-total adversary sets}\label{assum:adversary_bound}
    The unknown adversary set $\Ac \subset  \Vc$ is termed $F$-total if $|\Ac |\leq F$, where $F \in \intgnonneg$, that is there exist at most $F$ malicious agents in the network with $0\neq {\ub}^{\rm a}_{i}(t) \in \Lc_{p e}$ in \eqref{eq:ctrl_proto}. 
    The set $\Ac \subset  \Vc$ is termed $F$-local if $\forall \, i \in \Vc \setminus \Ac, \; |\Ac \cap \boldsymbol{\Nc}^{\, i(1)}| \leq F$, where $F \in \intgnonneg$ and the aggregated set of $1$-hop neighbors $\boldsymbol{\Nc}^{\, i(1)} \!=\! \setdef{j\in \Vc}{(i,j)\in \boldsymbol{\Ec}^{\mu}_{\s T}}$ with the edge set $\boldsymbol{\Ec}^{\mu}_{\s T}$ (to be specified) defined uniformly over the time interval $[t,\, t+T)$, $\forall \, t \in \realnonneg$, $\exists \, T \in \realpos$. i.e., each cooperative agent has no more than $F$ malicious agents with $0 
    \neq \ub^{\rm{a}}_{i} \in \Lc_{p e} $ in \eqref{eq:ctrl_proto} among its aggregated set of $1$-hop neighbors defined uniformly in time. 
\end{definition}
An explicit upper bound on $F$, and the explicit definition of the edge set $\boldsymbol{\Ec}^{\mu}_{\s T}$ will be given in Section \ref{sec:observer_design}.

\begin{remark}
The $F$-local model presented herein is a relaxation of the model in \cite{leblanc2013resilient,zhang2015notion} that required the upper bound inequality
holds point-wise in time, i.e., $\forall \, t \in \realnonneg,\, \boldsymbol{\Nc}^{\, i(1)} $ $ = \Nc^{\s i(1)}_{\sigt} = \setdef{j\in \Vc}{(i,j)\in {\Ec}_{\sigt}}$ 
(cf. the discrete-time version in \cite[Sec. 4.4]{dibaji2017resilient} and \cite{saldana2017resilient}).
\end{remark}

\textbf{Denial-of-Service (DoS) attack}.
We consider a time-constrained (distributed) DoS attack on the communication network $\Gc_{\sigt}=(\Vc, \Ec_{\sigt})$ that causes the intermittent unavailability of (state) information exchange, either partially or fully \cite{amin2009safe,de2015input,lu2017input}. 
We take into account such DoS attacks by the inclusion of some modes $\sigma(t) \in \Qc^{\rm a} \subset \Qc $ for the network $\Gc_{\sigt}$ where an unknown subset of edges $\Ec_\sigt$ are nullified. Accordingly, $\Gc_{\sigt \in \Qc^{\rm a}}$ is at most disconnected as a consequence of nullified (blocked) edge links, that is
\noindent
\begin{align}\label{eq:DoS_state}
    \exists \, \Ec^{\rm a}_{\sigma(t)}=\braces{(i,j) \in \Ec_{\sigma(t)} \mid (i,j)=\emptyset, \; i,j \in \Vc}
    &\quad \text{s.t.} 
    \nonumber \\
    \lambda_{2}(\lap_{\sigt}) = 0, 
    \ \ \forall \, t \in \boldsymbol{\Tc}_{\! \rm a} 
    = 
    \{ \Tc^{\rm a}_{\;k}\}^{\mb{n}}_{k \in \intgnonneg}&,
\end{align}
where $ \lambda_{2}(\lap_{\sigt})$ is the algebraic connectivity at $t$, $ \boldsymbol{\Tc}_{\! \rm a} = \{ \Tc^{\rm a}_{\;k}\}^{\mb{n}}_{k \in \intgnonneg}$ with $k \leq \mb{n} \in \intgnonneg$, denotes a finite sequence of $\mb{n} $ DoS attacks having bounded but not necessarily contiguous time intervals $\Tc^{\rm a}_{\;k}$'s, where $\Tc^{\rm a}_{\;k} = [t_k, \, t_k+T^{{\rm a}}_k)$ with \, $T^{{\rm a}}_k \in \realpos$ and $ \sigma(t_k) \in \Qc^{\rm a}$. 

We will show that if the set $\boldsymbol{\Tc}_{\! \rm a}$ in 
\eqref{eq:DoS_state} is sufficiently small the cooperation objective \eqref{eq:formation_consensus} is achievable. 

\subsection{Problem statement}\label{sec:problem_statement}
Consider $\Sigma_{\sigt}$ in \eqref{eq:cl_sys} with an unreliable communication network, $\Gc_{\sigt}, \, \sigt \in \Qc$, subject to the DoS attack in \eqref{eq:DoS_state} as well as deception attacks that are injected by a set of \emph{malicious} agents $\Ac \subset \Vc$.
The problems of interest are distributed detection of the set of \emph{malicious} agents $\Ac$, and the resilient cooperation of the remaining cooperative agents $\Vc\setminus\Ac$.

\textbf{Distributed attack detection.} 
We cast the attack detection problem as a form of distributed hypothesis testing problem where each mobile agent $\Sigma^{}_{i}$ in \eqref{eq:ol_sys} locally verifies either the null hypothesis $\nullH \!: {\rm attack}\text{-}{\rm free}, \; \text{if} \; \onehop{i} \cap \Ac \!=\! \emptyset$  or the alternative hypothesis $\alterH \!: {\rm attacked}, \; \text{if} \; \onehop{i} \cap \Ac \!\neq\! \emptyset.$
For this purpose, we equip each $\Sigma^{}_{i}$ in \eqref{eq:ol_sys} with a reconfigurable local attack detector module
of the form $\Sigma^{\s \Oc}_{\s \Vc^{i^{\s \dprime}}_{\sigma}} \!:\! 
    \res^{\, i}_{ \sigma}(t) \!=\! O^{i}_{\sigma}(\y^{\, i}_{\sigma})$,
%
in which $O^{i}_{\sigma}(\cdot)$ is a stable linear filter (e.g., a Luenberger-type observer) in each mode $\sigma \in \Qc$, whose explicit expression will be given in Section \ref{sec:local_obs_design}. Also, its inputs and outputs are, resp., $\y^{\, i}_{\sigma}$ in \eqref{eq:locally_measurements} and the residual $\res^{\, i}_{\sigma}(t) = \y^{\, i}_{\sigma} - \hat\y^{\, i}_{\sigma} $, where $\hat\y^{\, i}_{\sigma}$ is the estimation of $\y^{\, i}_{\sigma}$. Given the switching nature of $\Sigma_{\sigt}$ in \eqref{eq:cl_sys} with possibly unknown modes of $\Gc_{\sigt}$ \rev{subject to the DoS attack in \eqref{eq:DoS_state}}, it is necessary to note that the realization and reconfiguration of $\Sigma^{\s \Oc}_{\s \Vc^{i^{\s \dprime}}_{\sigma}}$, rely on (a minimum amount of) \rev{local information that is available intermittently not point-wise in time}. This contains the set of $2$-hop information available for each agent $i \in \Vc$, defined as

\noindent
\begin{align}\label{eq:2hop_info}
\hspace{-1.2ex}
    \Phi_{\sigt}^{\, i} = \braces{\graphxx{i}, \; \pb_j(t), \, \forall \, j  \in \Vc^{\,i^{\dprime}}_{\sigma}, \vb^{}_i(t) }\!,
\end{align}
where the \revv{topological knowledge $\Gc^{i^{\dprime}}_{\sigt}$, defined by \eqref{eq:2prox_edge}, can be either be pre-programmed as in autonomous monitoring scenarios \cite{yu2021resilient,mao2020novel} or be obtained via information exchange with only the $1$-hop neighbors $\onehop{i}, \, i \in \Vc$ upon network availability}. 
We remark that we do not explicitly address the case of $F$-total \emph{Byzantine} agents transmitting inconsistent information to their neighbors, \revv{nor covert attacks where a \emph{malicious} agent hides its true state by sending an altered state. We refer to \cite{leblanc2013resilient} and \cite{bahrami2021privacy} for these cases, respectively.
}

The attack detector module $\Sigma^{\s \Oc}_{\s \Vc^{i^{\s \dprime}}_{\sigma}}$ allows for quantification and verification of the simple null and alternative hypotheses using the local residuals, $\res^{\, i}_{ \sigma}(t)$'s, as follows: 
\begin{subequations}\label{eq:hypotheses_local}
\begin{align}\label{eq:hypotheses_local_null}
    \nullH &: {\rm attack}\text{-}{\rm free}, 
    && \text{if} \;\;  \forall \, j \in \onehop{i}, \;\, \forall\, i \in \Vc \setminus \Ac,
    \nonumber \\ & &&
    | \res^{i,j}_{ \sigma}(t) | \leq  \epsilon^{i,j}_{\sigma}, \;\,  \forall \, t \in \realnonneg,
    \\\label{eq:hypotheses_local_alternative}
    \alterH &: {\rm attacked}, 
    && \text{if} \;\;  \exists \, j \in \onehop{i}, \;\, \exists \, i \in \Vc \setminus \Ac,
    \nonumber \\ & &&
    | \res^{i,j}_{ \sigma}(t) | >  
    \epsilon^{i,j}_{\sigma} , \;\,  
    \exists \, t \in \realnonneg,
\end{align}
\end{subequations}
where $\res^{i,j}_{ \sigma}(t)$ is the $j$-th component of the residual signal of the local attack detector $\Sigma^{\s \Oc}_{\s \Vc^{\dprime}_{i}}$, and $\epsilon^{i,j}_{\sigma}$'s are the corresponding (dynamic) thresholds that will be defined later in Section \ref{sec:local_obs_design}.
%

\textbf{Resilient cooperative control.} Resilient cooperation refers to detaching from the detected set of malicious agents $\Ac \subset \Vc$ and convergence of the remaining cooperative mobile agents, $\Vc \setminus \Ac$, 
to the equilibrium defined in \eqref{eq:formation_consensus}. The isolation process induces topology switching that can be ultimately modeled by the induced network $\bar\Gc_{\sigt}=(\bar\Vc, \bar\Ec_{\sigt})$ with
%
\noindent
\begin{align}\label{eq:comm_net_res}
    \bar\Vc := \Vc \setminus \Ac, &&  \bar\Ec_{\sigt} := \Ec_{\sigt} \cap (\bar\Vc \times \bar\Vc).
\end{align}
Then, the problem of interest is to investigate under what conditions the resilient cooperation \eqref{eq:formation_consensus} over $\bar\Gc^{}_{\sigt}$ is achievable.

\section{Network Resilience and Stability analysis}\label{sec:net_res_and_stability}
In this section, we investigate the network resilience to intermittent and permanent disconnections, as well as the stability and convergence of the multi-agent system in \eqref{eq:cl_sys} with the unreliable communication network $\Gc_{\sigt}$. 
%
\begin{assumption}\label{assum:switching}
It is assumed that there exists a finite number of switches in any finite time interval. 
This allows us to rule out the Zeno phenomenon \cite{tanwani2012observability}. Formally, there exists a finite sequence $\braces{{t}_k}_{k=0}^{\mb{m}} = t_0, \dots, t_{\mb{m}}$, where $ \mb{m} \in \intgnonneg $ and $\mb{m} > \mb{n}$ in \eqref{eq:DoS_state}, that forms the set of $\mb{m}$ time instants in the ascending order of occurrence during any given time interval $[t_0, \, t_0+T)$, where $t_0 \in \realnonneg$, and $T\in \realpos$ are defined such that $T>(t_{\mb{m}}-t_0) \geq 0 $. Accordingly, the $\mb{m}+1$ (possibly unknown) modes $ \sigma(t_0),  \sigma(t_1), \dots, \sigma(t_{\mb{m}})$ ($\braces{\sigma(t_{k}) \in \Qc'\subseteq \Qc, \; k \in \{0, \dots, \mb{m} \} }$) denote the respective active modes of $\Sigma_{\sigt}$ in \eqref{eq:cl_sys} during the interval $ [t_0, \, t_0+T)$.
\end{assumption}
%
%
\begin{definition}\btitle{$(\mu,T)$-PE connected communication network}\label{def:comm_capability}
A communication network $\Gc_{\sigt}=(\Vc,\Ec_{\sigt})$ is called $(\mu,T)$-PE connected with some $ T \in \realpos $ and $\mu \in (0,\, N]$ if its associated Laplacian matrix 
$\lap_{\sigt}$
satisfies a $(\mu,T)$-Persistence of Excitation (PE) condition of the form 
\begin{align}\label{eq:PE_cond}
    \frac{1}{T}\int_{t}^{t+T} Q\lap_{\sigma(\tau)} Q^{\top} \, {\rm d}\tau \geq \mu I_{N-1},   \ \ \forall \, t\in \realnonneg,  
\end{align}
where the matrix $ Q \in \real^{\s (N-1) \times N}$ is defined such that
\begin{align}\label{eq:PE_Q}
 Q  \ones_{N} = \zeros, \;\, QQ^{\top} = I_{N-1}, \;\, Q^{\top}Q = I_{N} - \frac{1}{N} \ones^{}_{N} \ones^{\top}_{N}. 
\end{align}
\end{definition}
\begin{assumption}\label{assum:comm_capability}
    The communication network $\Gc_{\sigt}$ of the system $\Sigma_{\sigma(t)}$ in \eqref{eq:cl_sys} is assumed to be  $(\mu,T)$-PE connected. 
\end{assumption}
\begin{remark}\label{rmk:PE_connectivity}
The $(\mu,T)$-PE connectivity has appeared in the literature in different forms \cite{cichella2015cooperative,anderson2016convergence}. 
It relaxes the strict point-wise in-time connectivity to the case of connectivity only in the integral sense of \eqref{eq:PE_cond}, whereby positive algebraic connectivity in the integral sense that $\lambda_2 ( \frac{1}{T} \int_{t-T}^{t} \lap_{\sigma(\tau)}  \, {\rm d}\tau ) > \mu$ holds $\forall \,  t \geq T $ and $\exists \, \mu,T \in \realpos$ as in \eqref{eq:PE_cond}. This relaxation allows for modeling a class of networks including periodically switching networks, \cite{mao2020novel}, periodic with intermittent communications \cite{saldana2017resilient,yu2021resilient}, and jointly connected networks \cite{ren2007information},
\rev{as well as for quantifying resilience to the deception and DoS attacks defined in Section \ref{Sec:adversary_model}. See Proposition  \ref{prop:convergence_PE_like}.}
Finally, the matrix $Q$ in \eqref{eq:PE_Q} can be recursively obtained \cite[rmk. 2]{cichella2015cooperative}.
\end{remark}

We next show the equivalence of the $(\mu,T)$-PE condition presented herein and that in \cite{anderson2016convergence} for ensuring a positive \emph{algebraic connectivity} in the integral sense. 
\rev{The equivalent conditions will later be used in the stability and robustness analyses, particularly in Theorem \ref{thm:robustness-to-disconection} and Lemma \ref{lemma:net_level_obs}}. 

\begin{lemma}\btitle{Equivalence of $(\mu,T)$-PE conditions for Network Connectivity}
\label{lemma:PE_equvalence}
Consider a $(\mu,T)$-connected network $\Gc_{\sigt}=(\Vc,\Ec_{\sigt})$ with the associated Laplacian matrix $\lap_{\sigt}$. The following statements are equivalent:
\begin{enumerate}
    \item \label{itm:PE_cond} The condition \eqref{eq:PE_cond} holds.
    \item \label{itm:PE_cond_alt} There exist $\mu_{m}, \, \mu_{M}, \, T \in \realpos$ such that $\forall \, t\in \realnonneg$,
   \begin{align}\label{eq:PE_cond_alt}
    \hspace{-2em}
        \mu_{m} I_N &\leq 
        \frac{1}{T}\int_{t}^{t+T}
        \!
        \paren{ \inc_{\sigma(\tau)} + \frac{\ones^{}_{N}\bfrak^{\top}_{\tau}}{\sqrt{N}}} 
        \! \!
        \paren{ \inc_{\sigma(\tau)} + \frac{\ones^{}_{N}\bfrak^{\top}_{\tau}}{\sqrt{N}}}^{\!\!\!\top} {\!\rm d}\tau 
        \nonumber \\ 
        &= 
        \frac{1}{T}\int_{t}^{t+T} 
        \paren{ \lap_{\sigma(\tau)} + \frac{\ones^{}_{N}\ones^{\top}_{N}}{{N}} } \, {\rm d}\tau 
        \leq \mu_{M} I_{N}.
    \end{align}
    where $ \bfrak_{t} $ is the unit vector in the kernel of $ \inc_{\sigma(t)} $ \cite{anderson2016convergence}.
    \item \label{itm:PE_cond_alt_edge} There exist $ \delta, T  \in \realpos $ such that the set of edges 
    \begin{align} \label{eq:PE_cond_alt_edge}
        \boldsymbol{\Ec}_{\s T}^{\mu} =
        \big\{ (i,j) \in  
        \; &
        \Ec_{\sigma(t)} \mid  
        \frac{1}{T}   \int_{t}^{t+T}  a^{\sigma(\tau)}_{ij}  \, {\rm d}\tau \geq \delta ,
        \nonumber \\ &
        \forall \, t \in \realnonneg , \ \ i,j \in \Vc ,\,  i\neq j \big\},
    \end{align}
    forms a connected simple\footnote{This spans the class of graphs with a finite number of nodes ($|\Vc| < \infty$) containing no self-loop or multiple edges between any pair of nodes.} graph in the integral sense, denoted by $\boldsymbol{\Gc}_{\s T}^{\mu} = (\Vc, \boldsymbol{\Ec}_{\s T}^{\mu})$, where $\frac{1}{T}   \int_{t}^{t+T}  a^{\sigma(\tau)}_{ij}  \, {\rm d}\tau$'s in \eqref{eq:PE_cond_alt_edge} form the entries of the corresponding weighted adjacency matrix $\boldsymbol{\adj}$ and Laplacian matrix ${\boldsymbol{\lap}}$ that are defined as follows:
    \begin{align}\label{eq:PE_lap_adj}
    \hspace{-1em}
    \boldsymbol{\adj} = 
    \frac{1}{T}\int_{t}^{t+T} \adj_{\sigma(\tau)}  \, {\rm d}\tau,  
    \ \
    {\boldsymbol{\lap}   = \frac{1}{T}\int_{t}^{t+T} \lap_{\sigma(\tau)}  \, {\rm d}\tau }.
    \end{align}
\end{enumerate}
\end{lemma}

\begin{proof}
See Appendix \ref{app_PE_equvalence}.
%
\end{proof}
         
We next show the relation between the $(\mu,T)$-PE connectivity and the bounds on the vertex connectivity and robustness of graphs.
The connectivity-related bounds provide a measure of the network resilience to intermittent and permanent disconnections. The former is associated with resilience to DoS attack in \eqref{eq:DoS_state} and the latter is required for resilience to malicious agents though disconnecting from them (see \eqref{eq:comm_net_res}).  
\begin{definition}\btitle{$r$-robust network \cite{leblanc2013resilient}}\label{def:r_robust}
    A static graph $\Gc=(\Vc,\Ec)$ is $r$-robust, where $r = r(\Gc) \in \intgnonneg$ with $ 0 \leq r(\Gc) \leq \ceil{|\Vc|/2} $, if for any pair of nonempty, disjoint subsets of $\Vc$, at least one of the subsets, denoted as $\Sc$, holds $|\Nc^{\s i(1)}_{}\setminus \Sc| \geq r,\, \; \exists \, i \in \Sc$.
\end{definition}
\begin{definition}\btitle{$(r,T)$-robust network}\label{def:r_T_robust}
A time-varying network $\Gc_{\sigt}=(\Vc,\Ec_{\sigt})$ is called $(r,T)$-robust\footnote{The $(r,T)$-robust network herein is robustness in an integral sense as a relaxation of the $r$-robust static network (cf. the discrete-time version in \cite[Def. 2.2]{yu2021resilient}). The $(r,T)$-robust in Definition \ref{def:r_T_robust} should not be confused by the notation of $(r,s)$-robustness, for some $r \in \intgnonneg $ and $1\leq s \leq |\Vc|$, that is a strict generalization of $r$-robustness defined for a static graph \cite{leblanc2013resilient}.} 
with some $T \in \realpos $ and $r \in \intgnonneg \setminus \{0\}$ if the resultant static network $\boldsymbol{\Gc}_{\s T}^{\mu} = (\Vc, \boldsymbol{\Ec}_{\s T}^{\mu})$ with $\boldsymbol{\Ec}_{\s T}^{\mu}$ in \eqref{eq:PE_cond_alt_edge}, obtained under the $(\mu,T)$-PE condition \eqref{eq:PE_cond}, is $r$-robust, where $r \leq r(\boldsymbol{\Gc}_{\s T}^{\mu})$.
\end{definition}
\begin{definition}\btitle{$(\boldsymbol{\kappa},T)$-vertex-connected network}\label{def:k_T_conect}
A time-varying network $\Gc_{\sigt}=(\Vc,\Ec_{\sigt})$ is called $(\boldsymbol{\kappa},T)$-vertex-connected with some $T \in \realpos $ and $\boldsymbol{\kappa} \in \intgnonneg \setminus \{0\}$ if the resultant static network $\boldsymbol{\Gc}_{\s T}^{\mu} = (\Vc, \boldsymbol{\Ec}_{\s T}^{\mu})$ with $\boldsymbol{\Ec}_{\s T}^{\mu}$ in \eqref{eq:PE_cond_alt_edge}, obtained under the $(\mu,T)$-PE condition \eqref{eq:PE_cond}, is $\boldsymbol{\kappa}$-vertex-connected,  where $\boldsymbol{\kappa} \leq \boldsymbol{\kappa}(\boldsymbol{\Gc}_{\s T}^{\mu})$.
\end{definition}
\begin{proposition}\label{prop:graph_bounds}
    Let $\Gc_{\sigt}=(\Vc, \Ec_{\sigt})$ be a $(\mu,T)$-PE connected network under Assumptions \ref{assum:switching} and \ref{assum:comm_capability}. Then, $\Gc_{\sigt}$ is at least $(\ceil{ \frac{\mu}{2} } ,T)$-vertex-connected and $(\ceil{ \frac{\mu}{2}} ,T)$-robust, 
    and
    the following inequalities hold for the resultant network $\boldsymbol{\Gc}_{\s T}^{\mu} = (\Vc, \boldsymbol{\Ec}_{\s T}^{\mu})$ with $\boldsymbol{\Ec}_{\s T}^{\mu}$ in \eqref{eq:PE_cond_alt_edge}.
\noindent
\begin{align}\label{eq:graph_bounds_all}
\hspace{-1ex}
  \ceil{ \frac{\hat\mu}{2} }  
  \! \leq 
  r(\boldsymbol{\Gc}_{\s T}^{\mu})
  \leq
  \boldsymbol{\kappa}(\boldsymbol{\Gc}_{\s T}^{\mu}) \leq |\Vc|-1, \ \ \hat\mu :=   \lambda_2 \paren{ \boldsymbol{\lap} }  \geq \mu ,  
\end{align}
where the robustness $r(\boldsymbol{\Gc}_{\s T}^{\mu})$,  vertex connectivity $\boldsymbol{\kappa}(\boldsymbol{\Gc}_{\s T}^{\mu})$ are defined 
based on the adjacency and Laplacian matrices in \eqref{eq:PE_lap_adj}.
Additionally, if $\boldsymbol{\Gc}^{\mu}_{\s T}$ is a noncomplete, we have $ \lambda_2(\boldsymbol{\lap})  \leq  \boldsymbol{\kappa}(\boldsymbol{\Gc}_{\s T}^{\mu})$ ensuring that $\Gc_{\sigt}$ is at least $(\ceil{ \mu } ,T)$-vertex-connected.
\end{proposition}
\begin{proof}
See Appendix \ref{app_graph_bounds}.
\end{proof}
\begin{figure}
    \centering
    \includegraphics[width=.6\linewidth]{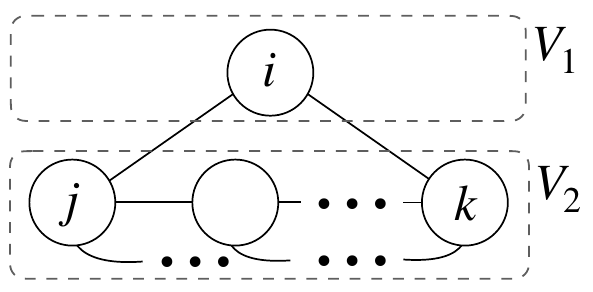}
    \caption{\small An example that illustrates how intermittent communication can drastically change the graph/network's algebraic connectivity $\lambda_2(\cdot)$ and consequently its robustness. See \eqref{eq:PE_cond}, \eqref{eq:graph_bounds_all}, and Remark \ref{rmk:PE_connectivity}.
    Let graph $\Gc_{\sigma(t)} = (\Vc, \Ec_{\sigma(t)})$ such that $|\Vc|= N+1$, with $\Vc = V_1 \cup V_2$ and $|V_2| = N$, where $N \geq 3$, and that the subgraph 
    $\bar{\Gc}_{\sigma(t)} = (\Vc\setminus V_1, \bar{\Ec}_{\sigma(t)})$ induced by removing the set $V_1$ and its incident edges is a complete graph 
    $\Kc_{|V_2|}=\bar{\Gc}_{\sigt}$. Note that the singleton $i \in V_1$ can be connected to any pair of disjoint nodes $j\neq k \in V_2$, and thus $\Sc =\{j,k\} \subset \Vc$ and the bidirectional edge set $\Ec_{\mathrm{cut}}=\{(i,j), (i,k) \}$ make, respectively, the minimum vertex cutset and edge cutset of $\Gc_{\sigt}$. Accordingly, one can verify that $\lambda_2(\Gc_{\sigt}) \leq \boldsymbol{\kappa}(\Gc_{\sigt})=\boldsymbol{e}(\Gc_{\sigt})=\boldsymbol{\delta}_{\min}(\Gc_{\sigt})= 2$, where $\boldsymbol{e}(\cdot)$ and $\boldsymbol{\delta}_{\min}(\cdot)$ are, resp., the edge connectivity and minimum node-degree. Also, if $\exists \, t \in \realnonneg \; \textrm{s.t} \; \Gc_{\sigt} = (\Vc, \Ec \setminus \Ec_{\mathrm{cut}}) $ because of an intermittent connection of the edges $\Ec_{\mathrm{cut}}$, we have graph disconnection with $\lambda_2(\Gc_{\sigt}=(\Vc, \Ec \setminus \Ec_{\mathrm{cut}})) = 0$. Yet, the induced subgraph $\Kc_{|V_2|}$ holds even a higher algebraic connectivity since $\lambda_2(\Kc_{|V_2|}) = |V_2|=N $, and $\boldsymbol{\kappa}(\Kc_{|V_2|})=\boldsymbol{e}(\Kc_{|V_2|})=\boldsymbol{\delta}_{\min}(\Kc_{|V_2|})= N-1$. This example has been constructed based on the discussions in \cite[Ch. 13.5]{godsil2001algebraic}.}
    \label{fig:graph_bounds_exmp}
\end{figure}
\begin{theorem}\btitle{Network resilience to node and edge disconnections}\label{thm:robustness-to-disconection}
    Let a $(\mu,T)$-PE connected network $\Gc_{\sigt}=(\Vc, \Ec_{\sigt})$ be $(r,T)$-robust (resp. $(\boldsymbol{\kappa},T)$-vertex-connected) under Assumptions \ref{assum:switching} and \ref{assum:comm_capability}. Let $\Ac \subset \Vc$ be a $(r-1)$-local (resp. $(\boldsymbol{\kappa} -1)$-total) adversary set. Then, the induced subgraph $ \bar\Gc_{\sigt}=(\bar\Vc=\Vc\setminus\Ac, \bar\Ec_{\sigt}) $
    in \eqref{eq:comm_net_res} admits the $(\bar\mu,\bar{T})$-PE connectivity condition in \eqref{eq:PE_cond}, for some  $\bar\mu, \bar{T} \in \realpos$, where $\bar{T} \leq T$ and $ \mu \leq \bar{\mu} +  |\Ac|$.
\end{theorem}
\begin{proof}
See Appendix \ref{app_robustness-to-disconection}.
\end{proof} 
\rev{
In other words, Theorem \ref{thm:robustness-to-disconection} together with \eqref{eq:graph_bounds_all} implies that if the network $\Gc_{\sigt}=(\Vc, \Ec_{\sigt})$ holds algebraic connectivity $\lambda_2(\boldsymbol{\lap})\geq \mu \geq 2F+\epsilon$ with $F \in \intgnonneg$ and $ \epsilon \in \realpos $ in the integral sense of \eqref{eq:PE_cond}, it will be at least $(F+1,T)$-vertex-connected (resp. $(F+1,T)$-robust) which in turn ensures the resulting network from the removal of an $F$-total (resp. $F$-local) adversary set, i.e. \eqref{eq:comm_net_res}, still maintains its connectivity in the sense of \eqref{eq:PE_cond} to a lower degree, allowing to achieve \eqref{eq:formation_consensus}.
}
\begin{remark}\label{rmk:graph_bounds_comparison} In this paper, we use the parameter $\mu$ in \eqref{eq:PE_cond} and \eqref{eq:graph_bounds_all} as a rather conservative proxy metric of the resilience of time-varying networks to disconnections of sorts. 
It is also noteworthy that the lower bound $ \ceil{ \frac{1}{2} \lambda_2 \paren{\boldsymbol{\lap}} }  \leq r(\boldsymbol{\Gc}_{\s T}^{\mu})  $ in \eqref{eq:graph_bounds_all} is tight as shown in
\cite[lemma 1]{shahrivar2017spectral}, \cite[Thm. 2]{saulnier2017resilient} for fixed graphs. One the other hand, the gap between ${r}(\boldsymbol{\Gc}_{\s T}^{\mu})$ and $\boldsymbol{\kappa}(\boldsymbol{\Gc}_{\s T}^{\mu})$ can be arbitrarily large (\cite{leblanc2013resilient,shahrivar2017spectral}) depending on a priori unknown intermediary typologies $\Gc_{\sigma}$'s, $\sigma \in \Qc'$ that form a $\boldsymbol{\Gc}_{\s T}^{\mu}$.
The illustrative example in Fig. \ref{fig:graph_bounds_exmp} demonstrates how an intermittent connection of edges in a class of graphs can affect the bounds in \eqref{eq:graph_bounds_all}. Moreover, in the special case of complete graphs over $N$ nodes, denoted by $\Kc_{N}$, we have $\ceil{\lambda_2(\Kc_3)/2}= \boldsymbol{\delta}_{\min}(\Kc_3) = 2$  for $N=3$ nodes, that shows the bound in \eqref{eq:graph_bounds_all} is tight. 
If the exclusion of complete graphs can be guaranteed the lower bound to $\hat\mu \leq \boldsymbol{\kappa}(\boldsymbol{\Gc}^{\mu}_{\s T})$ can be used for node connectivity. (cf. \cite[Cor. 13.5.2]{godsil2001algebraic}). 
\end{remark}

We now provide a convergence bound for the consensus/formation equilibrium in \eqref{eq:formation_consensus}. 
Associated with \eqref{eq:formation_consensus}, we define an output (coordinate) vector ${\Yc} \in \real^{2N-1}$ as 
\noindent
\begin{align}\label{eq:output_coord_vec}
    {\Yc} = 
    \begin{bmatrix}
       \zeta \\  \vb
    \end{bmatrix}
    :=
    \begin{bmatrix*}[l]
       Q & \zeros_{\s (N-1) \times N} \\ 
       \zeros_{\s N} & I_{\s N} 
    \end{bmatrix*}
    \begin{bmatrix}
       \widetilde{\pb} \\  \vb
    \end{bmatrix} = \mb{C}_{\rm eq} \x,
\end{align}
where $\widetilde{\pb}$, $\vb$ and $Q$ are given in \eqref{eq:cl_sys} and \eqref{eq:PE_Q}, respectively.
It then follows that $\zeta = Q \widetilde{\pb} = \zeros_{N-1}$ and $ \vb = \zeros_{N} $ imply $ \widetilde{\pb}_i - \widetilde{\pb}_j = 0, \; \forall \, i,j \in \Vc $ and $ {\vb}_i = 0, \, \forall \, i \in \Vc$ (that is $\Yc(t) \rightarrow \zeros \equiv $ \eqref{eq:formation_consensus}). 

\begin{proposition}\btitle{Stability under $(\mu,T)$-PE connectivity}\label{prop:convergence_PE_like}
Consider the system in \eqref{eq:cl_sys} and let Assumptions \ref{assum:switching} and \ref{assum:comm_capability}
hold. Also, let $ \sum_{k \in \intgnonneg}^{\mb{n}} T^{\rm a}_{k}< T$, where $T$ is the period in $\eqref{eq:PE_cond}$,
hold uniformly in time for the DoS attack in \eqref{eq:DoS_state}.  
Then, there exists a sufficiently large control gain $\gamma $ for \eqref{eq:ctrl_proto} that ensures, for each $\norm{{\x}(t_0)} \leq \infty$ and for every $\ua \in \Lc_{p e}$ with $\sup_{t_0 \leq t \leq T_d}\norm{\ua(t)} < \infty $, the system \eqref{eq:cl_sys} with the output ${\Yc}(t)$ in \eqref{eq:output_coord_vec}
is finite-gain $\Lc_{p}$ stable with the following upper bound:
\noindent
\begin{subequations}\label{eq:state_bounds}
\begin{align}
    &\norm{{\Yc}(t)}_{\Lc_p}
    \leq 
    \kappa_{\x}
    e^{-\lambda_{\x}(t-t_0)}
    \norm{{\x}(t_0)}
    +
    \nonumber \\
    & \hspace{8em}
    \kappa_{\mb{u}}
    \norm{(\ua)_{T_d}}_{\Lc_{p}},
    \ \   \forall \, t \geq t_0 \in \realnonneg, 
    \\
    &\kappa_{\x} = \norm{C} \sqrt{
    {\scriptstyle
    \frac{ \max \braces{\lambda^{-1}_{ \chi},\beta}}{ \min \braces{ \frac{\gamma}{\alpha N},\beta}}}
    } 
    \norm{C^{\s -1}}, 
    \,
    \kappa_{\mb{u}} = \norm{C} 
    {\scriptstyle
    \frac{ \max \braces{\lambda^{-1}_{ \chi},\beta}}{ \lambda^{}_{ \x} \min \braces{ \frac{\gamma}{2\alpha N},\frac{\beta}{2}}}
    },
    \\&
     0 < \lambda_{\x} < \lambda_{ \chi} = \eta e^{-2 \eta T}, 
     \ \
     C =
    \begin{bsmallmatrix*}[l]
       \frac{1}{\gamma} I_{\s N-1} & -\frac{1}{\gamma}Q \\ 
       \zeros_{\s N \times (N-1)} & \hspace{1em} I_{\s N} 
    \end{bsmallmatrix*},
\end{align}
\end{subequations}
where $ \eta = -\frac{1}{2T} \ln (1-\frac{(\alpha / \gamma) \mu T}{1+(\alpha / \gamma)^{2} N^{2} T^{2}})$ and $\beta \in \realpos$.
Additionally, if $\ua=\zeros$, the system's state trajectories uniformly exponentially converge to the equilibrium \eqref{eq:formation_consensus}, (provided the formation configuration is feasible) with the bound
\begin{subequations}\label{eq:formation_consensus_expo}
\begin{align} \label{eq:formation_consensus_pos_expo}
 \left|\pb_i(t)-\pb_j(t) - \pstar_{ij} \right|  \leq 
 \sqrt{2}\kappa_{\boldsymbol{\x}}
    e^{-\lambda_{\x}(t-t_0)} \norm{\x(t_0)},
\\ \label{eq:formation_consensus_vel_expo}
 \left|\vb_i(t) \right| \leq  \kappa_{\boldsymbol{\x}}
    e^{-\lambda_{\x}(t-t_0)} \norm{\x(t_0)}, 
  \ \ \forall \, i,j \in \Vc 
\end{align}
\end{subequations}
for all $t \geq t_0 \in \realnonneg $, with $\kappa_{\boldsymbol{\mb{x}}}, \, \lambda_{\mb{x}} \in \realpos$ as given in \eqref{eq:state_bounds}. 
\end{proposition}
\begin{proof}
See Appendix \ref{app_convergence_PE_like}.
\end{proof}
\rev{We remark that the choice of $\gamma=\alpha N$, for $\alpha\geq1$ yields a convergence rate $\lambda_{\chi}$ in \eqref{eq:state_bounds}-c that depends only on $\mu,\, T$, with the maximum occurring at $\mu=N, T=1$, associated with complete network connectivity, see \eqref{eq:PE_cond} and Remark \ref{rmk:PE_connectivity}.
This, however, is not the only valid choice. 
}

\section{Observer design and Attack Detection}\label{sec:observer_design}
Here, we consider the design of observers serving the reconfigurable local attack 
detector module $\Sigma^{\s \Oc}_{\s \Vc^{i^{\s \dprime}}_{\sigma}}$ in Section \ref{sec:problem_statement}. The observer design for $\Sigma_{\sigt}$ in \eqref{eq:cl_sys} is subject to two constraints. First, \emph{a priori} full knowledge of $\A$ may not be available for each mobile agent due to random communication link dropouts or switching links.
Second, local state information $\y^{\, i}_{\sigma}$ in \eqref{eq:locally_measurements},  which is available for each mobile agent, is subject to change since the respective $k$-hop neighbors change in an \emph{a priori} unknown time-varying network. Consequently, ensuring the uniform observability of $(\A^{}, \C^{\, i}), \; \forall \, t \in \realnonneg, \, \forall \, i \in \Vc$ may not be tractable or feasible. 

In what follows, we, first, characterize network-level conditions under which almost any set of adversarial inputs $\ua$ is observable at the measurements of the cooperative agents that are $\y^{\, i}_{\sigma}$'s, for $i \in \Vc \setminus \Ac$.  
Second, we propose a class of local observers for $\Sigma^{\s \Oc}_{\s \Vc^{i^{\s \dprime}}_{\sigma}}$ that are realizable using \eqref{eq:2hop_info}, enabling distributed attack detection through  \eqref{eq:hypotheses_local}.

\subsection{Detectability of adversarial inputs}
We note that the switched system $ \Sigma_{\sigt} $ in \eqref{eq:cl_sys}-\eqref{eq:locally_measurements} represents a family of linear time-invariant (LTI) systems, each of which is associated with one mode $\sigma(t) \in \Qc$. Therefore, similar to \cite{mao2020novel,pasqualetti2011consensus}, the results herein are derived based on the concepts of output-zeroing and \emph{state and input observability} of the (switched) LTI systems.

Consider a generic solution $\x(t; \x(t_{0}), \ua(t))$ to $\Sigma_{\sigt}$ in \eqref{eq:cl_sys} under Assumptions \ref{assum:switching} and \ref{assum:comm_capability}. Then, the concatenation of the measurements $\y^{\,i}_{\sigma}$'s, given in \eqref{eq:locally_measurements}, of the set of cooperative agents, $ \Vc\setminus\Ac = \braces{i_1, \dots , i_{\s |\Vc\setminus\Ac|}}  $, is defined as

\noindent
\begin{align}\label{eq:collective_measurements}
    \y^{\s \Vc \setminus \Ac}_{\sigma}(t; \x(t_{0}) , \ua(t))
    = 
    \col(\y^{\, i_1}_{\sigma}, \dots, \y^{\, i_{\s |\Vc\setminus\Ac|}}_{\sigma}) &=
    \nonumber \\
    \col(\C^{\, i_1}, \dots, \C^{\, i_{\s |\Vc\setminus\Ac|}})
    \x(t; \x(t_{0}), \ua(t)) &=
    \nonumber \\
    \C^{\s \Vc \setminus \Ac} \,
    \x(t; \x(t_{0}), \ua(t))&.
\end{align}

It is necessary to note that the entirety of measurement $\y^{\s \Vc \setminus \Ac}_{\sigma}(t; \x(t_{0}) , \ua(t))$ is not available for any agent $i \in \Vc$. We use this collective set of the measurements of cooperative agents $\Vc \setminus \Ac$ and a generic set of adversarial inputs $\ua$ introduced by the set of malicious agents $\Ac $, in an input observability context for attack detection analyses.

\begin{definition}\btitle{Stealthy and indistinguishable attacks}\label{def:undetectable _indistinguishable}
For $\Sigma_{\sigt}$ in \eqref{eq:cl_sys} under Assumptions \ref{assum:switching} and \ref{assum:comm_capability}, any generic set of inputs $\ua(t) \in \Lc_{p e}$ injected by a set of malicious agents $\Ac$ is stealthy for the remaining cooperative agents $\Vc \setminus \Ac$ if

\noindent
\begin{align}\label{eq:undetectable}
    \exists \; \x(t_{0}), \x'(t_{0}) \in \real^{2N} \ \ &\text{\rm s.t.} \ \ 
      \forall \, t  \in [t_0 ,\; t_0+T), 
    \nonumber \\ 
    \y^{\s \Vc \setminus \Ac}_{\sigma}(t; \x(t_{0}), \ua(t)) &= \y^{\s \Vc \setminus \Ac}_{\sigma}(t; \x'(t_{0}), \zeros ),
\end{align}
with $\y^{\s \Vc \setminus \Ac}_{\sigma}(t;\cdot,\cdot)$ as in \eqref{eq:collective_measurements}, $t_0 \in \realnonneg$, and $T\in \realpos$ as in 
\eqref{eq:PE_cond}.
Likewise, for any given two generic sets $\mb{u}_{\! \Ac_1}(t) \in \Lc_{p e}$ and $\mb{u}_{\! \Ac_2}(t) \in \Lc_{p e}$ injected, resp., by a nonempty set of malicious agents $ \Ac_1 \in \Ac$ and some $ \Ac_2 \in \Ac$, where $\Ac_1 \neq \Ac_2$, $\mb{u}_{\! \Ac_1}$ is indistinguishable from $\mb{u}_{\! \Ac_2}$ for the cooperative agents $\Vc \setminus\! \Ac$ if   
%

\noindent
\begin{align}\label{eq:indist}
    \exists \; \x(t_{0}), \x'(t_{0}) \in \real^{2N} \ \ &\text{\rm s.t.} \ \ 
    \; \forall \, t  \in [t_0 ,\; t_0+T), 
     \nonumber \\
    \y^{\s \Vc \setminus \Ac}_{\sigma}(t; \x(t_{0}), \mb{u}_{\! \Ac_1}(t)) &= \y^{\s \Vc \setminus \Ac}_{\sigma}(t; \x'(t_{0}), \mb{u}_{\! \Ac_2}(t) ).
\end{align}
\end{definition}
\begin{lemma}\btitle{Characterization of stealthy and indistinguishable attacks}\label{lemma:undetectable_chrz}
Consider $\Sigma_{\sigt}$ in \eqref{eq:cl_sys} and let Assumptions \ref{assum:switching} and \ref{assum:comm_capability} hold. Also, let $\mb{B}_{\Ac_1} \mb{u}_{\! \Ac_1}(t)$ and $\mb{B}_{\Ac_2} \mb{u}_{\! \Ac_2}(t)$ be two generic sets of adversarial $\Lc_{p e}$-norm bounded inputs injected by a nonempty set of malicious agents $ \Ac_1 \in \Ac$ and some $ \Ac_2 \in \Ac$, where $\Ac_1 \neq \Ac_2$. Then, $\mb{u}_{\! \Ac_1}(t)$ and $\mb{u}_{\! \Ac_2}(t)$ are indistinguishable for the cooperative agents $\Vc \setminus \Ac$ during $ t \in [t_0 ,\, t_0+T)$,  if and only if $\exists \, {\rm x}(t_0) \in \real^{2N}$ such that 

\noindent
\begin{align}\label{eq:undetect_chrz}
& \mb{C}^{\s \Vc\setminus\Ac}_{\sigma(t_{k})} e^{\mb{A}_{\sigma(t_{k})}(t-t_{k})} {{\rm x}(t_{k})}
= 
    \nonumber \\ & 
\mb{C}^{\s \Vc\setminus\Ac}_{\sigma(t_{k})} \int_{t_{k}}^{t} e^{\mb{A}_{\sigma(t_{k})}(t-\tau)} 
\big( \mb{B}_{\Ac_1} 
\mb{u}_{\! \Ac_1}(\tau)
- 
\mb{B}_{\Ac_2} \mb{u}_{\! \Ac_2}(\tau) \big) \,{\rm d}\tau, \ \
\nonumber \\  
 & \hspace{15em}
 t \in \;  [t_{k}, \, t_{k+1}),
\end{align}
where $0 \leq k \leq \mb{m}$ and  $ t_{\mb{m}+1}=: t_0+T $, with $\mb{m}$ and $T$ as in Assumptions \ref{assum:switching} and \ref{assum:comm_capability}, and ${\rm x}(t_0)= (\mb{x}'(t_0)-\x(t_{0}))$ when $k =0$, and for $k\neq 0$, we have ${\rm x}(t_{k}) =  \x'(t_{k}) - \x(t_{k})$, where
\noindent
\begin{align}\label{eq:undetect_chrz_IC}
\hspace{-1ex}
& {\rm x}(t_{k}) = 
\prod_{i=k}^{1} e^{\mb{A}_{\sigma(t_{i-1})}
(t_i - t_{i-1})
}
    {\rm x}(t_0)
+ \sum\limits_{i=1}^{k}
    \prod_{j=k}^{i+1} 
	e^{\mb{A}_{\sigma(t_{j-1})} 
        (t_j - t_{j-1})
        }
\nonumber \\
\hspace{-1ex}
& {\int^{t_i}_{t_{i-1}}e^{\mb{A}_{\sigma(t_{i-1})}(t_i-\tau)} }
\big( \mb{B}_{\Ac_1} \mb{u}_{\! \Ac_1}(\tau) - \mb{B}_{\Ac_2} \mb{u}_{\! \Ac_2}(\tau)) \, {\rm d}{\tau} \big).
\end{align}
Additionally, if \eqref{eq:undetect_chrz}-\eqref{eq:undetect_chrz_IC} hold for $\mb{u}_{\! \Ac_2}(t) = \zeros, \, \forall \, t \in [t_0 ,\, t_0+T)$, then $\mb{u}_{\! \Ac_1}(t)$ is stealthy.    
\end{lemma}
\begin{proof}
    It follows directly from the definitions in \eqref{eq:undetectable}-\eqref{eq:indist} and the system $\eqref{eq:cl_sys}$'s solution for the state trajectories, see \cite{RB_PHDthesis2024}. 
\end{proof}

It is necessary to note that the realization of \eqref{eq:undetect_chrz} requires a priori knowledge of the system which is not available for any agent. Moreover, based on the concepts of \emph{state and input observability} \cite[Thm. 2]{boukhobza2007state},
\cite[Ch. 3.11]{zhou1996robust}, and the invariant zeros of the switched LTI systems \cite{mao2020novel}, the realizations of \eqref{eq:undetect_chrz} in each mode coincide with the existence of the set of vector-valued adversarial input $\ua$ unobservable at the vector-valued output $\y^{\s \Vc\setminus\Ac}_{\sigma}$. For LTI systems, it is well-known that such a set of inputs, \rev{referred to as zero-dynamics attacks (see \cite{mao2020novel,pasqualetti2011consensus} and Section \ref{Sec:adversary_model})}, is not generic, and is characterized using the output-zeroing directions of the system. Particularly, for $\Sigma_{\sigt}$ \emph{in each mode} $\sigma \in \Qc$, it follows from \cite[Ch. 3.11]{zhou1996robust} that the output-zeroing directions are induced by the rank deficiencies of the matrix pencil $\boldsymbol{P}(\lambda_{\rm o}, \sigma)$ for some $\lambda_{\rm o} \in \cplx$, where  
\begin{align}\label{eq:pencil_collective}
\boldsymbol{P}(\lambda_{\rm o}, \sigma) =
    \begin{bmatrix}
    \lambda_{\rm o} I - \mb{A}^{}_{\sigma}  & -\Ba \\
    \mb{C}^{\s \Vc \setminus \Ac}_{\sigma} & \zeros
    \end{bmatrix}.
\end{align}

We next present conditions under which the intersection of the output-zeroing subspaces of $\Sigma_{\sigt}$ in \eqref{eq:cl_sys} make an empty set, ensuring \rev{almost no deception attacks, defined in Section \ref{Sec:adversary_model}, can be stealthy in the sense of \eqref{eq:undetectable}.} 
\begin{lemma}\label{lemma:net_level_obs}
    Consider $\Sigma_{\sigt}$ in \eqref{eq:cl_sys}-\eqref{eq:locally_measurements} with \eqref{eq:2hop_info} during an interval $[t_0, \, t_0+T)$ defined under Assumptions \ref{assum:switching} and \ref{assum:comm_capability}. Let $\Sigma_{\sigt}$ be subject to any generic set of adversarial inputs $\ua(t) \in \Lc_{p e}$ injected by an $F$-total (resp. $F$-local) set $\Ac$ of malicious agents such that $ 0 \leq F \leq \boldsymbol{\kappa}(\boldsymbol{\Gc}_{\s T}^{\mu})-1$ (resp. $ 0 \leq F $ $ \leq r(\boldsymbol{\Gc}_{\s T}^{\mu})-1$). 
    Then, the following statements are equivalent.
\begin{enumerate}
\item There exists no generic set of inputs $\ua(t) \in \Lc_{\infty e}$ stealthy in the sense of \eqref{eq:undetectable}.
\item For almost all $\lambda_{\rm o} \in \cplx$, $\cap_{\sigma \in  \Qc'} \ker \paren { \boldsymbol{P}(\lambda_{\rm o}, \sigma) } = \emptyset$, where $\boldsymbol{P}(\lambda_{\rm o}, \sigma)$ is given by \eqref{eq:pencil_collective}. 
\end{enumerate} 
\end{lemma}
\begin{proof} 
See Appendix \ref{app_net_level_obs}.
\end{proof}
\rev{
In other words, Lemma \ref{lemma:net_level_obs} states that $(F+1, T)$-vertex connectivity (resp. $(F+1, T)$-robustness), where $F \in \intgnonneg$, ensures almost no $F$-total (resp. $F$-local) set of malicious agents with the deception attacks defined in Section \ref{Sec:adversary_model} is stealthy in the sense of \eqref{eq:undetectable} for the cooperative agents in \eqref{eq:cl_sys} with \eqref{eq:2hop_info}.
(cf. \cite{pasqualetti2011consensus,dibaji2017resilient} where $(F+1)$-vertex connectivity and $(2F+1)$-robustness are required point-wise in time.)
}

We next investigate the level of local observability for each agent given the locally available information $\Phi_{\sigt}^{\, i}$ in \eqref{eq:2hop_info} and measurements $\y^{i}_{\sigma}$ in \eqref{eq:locally_measurements}, as opposed to ensuring the global observability of the pair $(\A^{},\C^{\, i})$ associated with \eqref{eq:cl_sys}-\eqref{eq:locally_measurements} that might not be tractable.

\subsection{Local dynamics and observability analysis}\label{}
Consider the set of $2$-hop information available for each agent $i \in \Vc$ as defined by $\Phi_{\sigt}^{\, i}$ in  \eqref{eq:2hop_info}, and local measurements $\y^{i}_{\sigma}$ in \eqref{eq:locally_measurements}.
Let $ \Ic_i = \Vc^{\, i^{\dprime}}_{\sigma} $ ($\Ic$ in short) in \eqref{eq:locally_measurements} and $ \Rc_i = \Vc \setminus \Vc^{\, i^{\dprime}}_{\sigma} $ and assume $\Ic_i$ and $\Rc_i$ ($\Rc$ in short) are sorted in the ascending order of agents' indices. 
Then, 
\eqref{eq:cl_sys}-\eqref{eq:locally_measurements}
can be partitioned as
\noindent
\begin{align} \label{eq:2hop_sys}
\Sigma_{\s \Vc^{\, i^{\s \dprime}}_{\sigma}}\!&:\!
\left\{
\begin{array}{l}
\!\!
    \xdot_{\s \Ic} 
    = 
    \A^{\s \Ic}
    \x_{\s \Ic} 
 +
 \rho(\x_{\s \Ic} , \x_{\s \Rc} )
 + \Bapp \uapp, 
 \\
 \! \!
 \y^{\, i}_{\sigma} = \C^{\s \Ic} \x_{\s \Ic}
\end{array}
\right.
    \\ \label{eq:rest_sys}
\Sigma_{\s \Rc_{i}}\!&: \phantom{\{}  
\xdot_{\s \Rc}
=
 \A^{\s \Rc}  \x_{\s \Rc} + \A^{\s \Rc,\Ic}  \x_{\s \Ic} + \Barest \uarest,
\end{align}
where $\x_{\bullet} \coloneqq \col \paren{ \widetilde{\pb}_{\bullet}, \vb^{}_{ \bullet } }, \; \bullet \in \{ \Ic, \Rc \}$ denotes the position and velocity states of the agents in each set, and the system matrices are defined as
\begin{subequations}\label{eq:sys_matrices_twohop}
\begin{align}
    \A^{\s \Ic} &=
    \begin{bmatrix}
    \zeros_{\s |\Ic|\times |\Ic|} & \phantom{-\gamma}I_{\s |\Ic|}
    \\
    -\alpha \lap^{\dprime}_{\sigma} &  -\gamma I_{\s |\Ic|}
    \end{bmatrix}, 
    \; \,
    \Bapp \!= \!
    \begin{bsmallmatrix} 
    \zeros_{\s |\Ic|\times |\Ac^{\dprime}| } \\
    \rule[.5ex]{3.2em}{0.4pt} \\
    I_{\s \Ac^{\dprime}}  
    \end{bsmallmatrix}\!,
    \\ \label{eq:sys_matrices_twohop_C}
\C^{\s \Ic} &= \diag (I^{}_{\s |\Ic|},
    {\efrak^{\, \s 1 }_{\s |\Ic|}}^{\!\! \s \top}),  
  \\
\hspace{-1ex}
\rho(\x_{\s \Ic}, \x_{\Rc}) &= 
  \dtilde{\A^{\s \Ic}} \x_{\s \Ic}+\A^{\s \Ic, \Rc} \x_{\s \Rc} 
  =
   \begin{bsmallmatrix}
    \zeros_{\s |\Ic|\times 1 }  \\ 
     \rule[.5ex]{2.5em}{0.4pt} \\
     \zeros_{\s |\Vc^{\,i^{\prime}}_{\sigma}|\times 1} \\  
     \underline{\rho}
 \end{bsmallmatrix}, 
    \\
    \dtilde{\A^{\s \Ic}} &= \!
    \left[\begin{array}{c|c}
    \zeros_{\s |\Ic| \times |\Ic|}  & \!\!  \zeros_{\s |\Ic|}
    \\ \hline 
    \!\!
    \begin{array}{ll}
     \zeros_{ }   &  \phantom{-} \zeros_{ }  \\
     \zeros_{}  & -\alpha {\dtilde{\lap}}_{\sigma}
    \end{array} \!\!
    &  \!\! \zeros_{\s |\Ic| }
    \end{array}\right]\!,\!
    \\
\A^{\s \Ic, \Rc} &= \!
    \left[\begin{array}{c|c}
    \zeros_{\s |\Ic| \times |\Rc|} & \zeros_{\s |\Ic| \times |\Rc|}
    \\ \hline
    \begin{array}{ll}
    \zeros_{\s  |\Vc^{\,i^{\prime}}_{\sigma}| \times  |\Rc|}  \\
      -\alpha \lap^{ (23)}_{\sigma}
    \end{array}
    &  \zeros_{ }
    \end{array}\right],
    \\
    \A^{\s \Rc} &= \!
    \begin{bmatrix}
    \zeros_{} & \phantom{-\gamma}I_{\s |\Rc|}
    \\
    -\alpha \lap^{(33)}_{\sigma} &  -\gamma I_{\s |\Rc|},
    \end{bmatrix}, 
    \\
\A^{\s \Rc,\Ic} &= \!
    \left[\begin{array}{cc|c}
    \zeros_{}  & \zeros & \zeros_{}
    \\ \hline
     \zeros_{}    &  -\alpha \lap^{ (32)}_{\sigma}
    &  \zeros_{}  \end{array} \right], \!
    \\
    \Ac^{\dprime} &= \Ac \cap \Vc^{\, i^{\dprime}}_{\sigma}, \qquad  \Ac^{\rm r} = \Ac \setminus \Ac^{\dprime}.
\end{align}
\end{subequations}
where $\underline{\rho}=-\alpha \big( {\dtilde{\lap}}_{\sigma} \widetilde{\pb}_{\s \twohop{i}} + \lap^{(23)}_{\sigma} \widetilde{\pb}_{\s \Rc} \big)$, $\uapp = \col \paren{ {\ub}^{\rm a}_{{i}} }_{i \in \Ac^{\dprime}} \in \real^{|\Ac^{\dprime}|}$, $I_{\! \Ac^{\dprime}} = \big[\efrak^{\s \, i_{1}}_{\s |\Ic|}\; \efrak^{\s \, i_2}_{\s |\Ic|}\; \dots\; \efrak^{\s i_{\s |\Ac^{\dprime}|}}_{\s |\Ic|}\big] \in \real^{|\Ic| \times |\Ac^{\dprime}|} $. 

We note that $\Sigma_{\s \Vc^{\dprime}_{i}}$ with known $\A^{\s \Ic}$ and $\C^{\s \Ic},$ 
is the dynamics available for each agent $i \in \Vc$ given $\Phi_{\sigt}^{\, i}, \, \forall \, t \in \realnonneg$, and the dynamics $\Sigma_{\s \Rc_{i}}$ and the possibly existing coupling term $\rho(\x_{\s \Ic} , \x_{\s \Rc} )$ are unknown to the agent $i \in \Vc$. 
Moreover, for every agent $i\in \Vc$,
$\y^{\, i}_{\sigma}$ in \eqref{eq:2hop_sys} and \eqref{eq:locally_measurements} are the same set of measurements obtained by reordering $\C^{\s \Ic}$ and $\C^{\, i}$. 

The following results address the effect of $\Sigma_{\s \Rc_{i}}$ on  $\Sigma_{\s \Vc^{\, i^{\s \dprime}}_{\sigma}}$ as well as the observability of  $\Sigma_{\s \Vc^{\, i^{\s \dprime}}_{\sigma}}$, \rev{which will be used later in local observer design.}

\begin{proposition}\label{prop:couplingTerm_convergence}
Consider \eqref{eq:2hop_sys} and \eqref{eq:rest_sys} under Assumptions \ref{assum:switching} and \ref{assum:comm_capability}. The coupling term $\rho(\x_{\s \Ic}, \x_{\Rc})$ in $\Sigma_{\s \Vc^{\, i^{\s \dprime}}_{\sigma}}$ of agent $i \in \Vc$, holds the bound $ \norm{ \rho(\x_{\s \Ic}, \x_{\Rc}) } \leq  \alpha
     \kappa_{\x}  e^{-\lambda_{\x}(t-t_0)} 
     \norm{\x(t_0)}
    +
    \alpha \kappa_{\mb{u}}
    \sup_{ t_0 \leq t \leq T_d} \norm{\ua(t)}, \;
     \forall \, t \geq t_0 \in \realnonneg$, where $\alpha$ is given in \eqref{eq:ctrl_proto}, and ${\kappa}_{\x} $, $\kappa_{\mb{u}}$, and $\lambda_{\x}$ are given in \eqref{eq:state_bounds}. 
     Additionally, if $\Ac = \emptyset $, the coupling term exponentially converges to $\zeros$ with  $ \norm{ \rho(\x_{\s \Ic}, \x_{\Rc}) } \leq  \alpha
     \kappa_{\x}  e^{-\lambda_{\x}(t-t_0)} 
     \norm{\x(t_0)}, \,
     \forall \, t \geq t_0 \in 
     \realnonneg $.
\end{proposition}
\begin{proof}
See Appendix \ref{app_couplingTerm_convergence}.
\end{proof}
\begin{proposition}\label{prop:observability_local}
Consider the $2$-hop dynamics $\Sigma_{\s \Vc^{\, i^{\s \dprime}}_{\sigma}}$ in \eqref{eq:2hop_sys} for each agent $i \in \Vc \setminus \Ac$  communicating over $\Gc_{\sigt}$ under Assumptions \ref{assum:switching} and \ref{assum:comm_capability}. Then, 
the following statements hold.
\begin{enumerate} 
    \item \label{prop:observability_local_st1}
    The pair $({\mb{A}}^{\s \Ic}_{\sigma}, {\C^{\s \Ic}})$ in 
    $\Sigma_{\s \Vc^{i^{\s \dprime}}_{\sigma}},\, \; \forall \, i \in \Vc$, is observable in each mode $\sigma \in \Qc$.
    \item \label{prop:observability_local_st2}
    There exists no generic set of inputs $\ua(t) \in \Lc_{p e} $ stealthy in the sense of \eqref{eq:undetectable}, where $\y^{\, i}_{\sigma}$'s are given in \eqref{eq:2hop_sys}, provided the set $\Ac$ of malicious agents is $F$-total (resp. $F$-local), with $ 0 \leq F \leq \boldsymbol{\kappa}(\boldsymbol{\Gc}_{\s T}^{\mu})-1$ (resp. $ 0 \leq F \leq r(\boldsymbol{\Gc}_{\s T}^{\mu})-1$).
\end{enumerate}
\end{proposition}
\begin{proof}
See Appendix \ref{app_observablity_local}.
\end{proof}
Having quantified the conditions on the attack stealthiness and local observability, we next propose the reconfigurable local attack detector module that relies only on the time-varying local information $\Phi_{\sigt}^{\, i}$ in \eqref{eq:2hop_info} and that performs the distributed hypothesis testing in \eqref{eq:hypotheses_local}. 

\subsection{Reconfigurable attack detector (local observer)}\label{sec:local_obs_design}
For each agent $i \in \Vc$ with the $2$-hop dynamics $\Sigma_{\s \Vc^{\dprime}_{i}}$ in \eqref{eq:2hop_sys} and local information $\Phi_{\sigt}^{\, i}$ in \eqref{eq:2hop_info}, the local attack detector $\Sigma^{\s \Oc}_{\s \Vc^{i^{\s \dprime}}_{\sigma}}$  is proposed as follows
\noindent
\begin{subequations}\label{eq:2hop_obs}
\begin{align} 
\hspace{-1em}  
\Sigma^{\s \Oc}_{\s \Vc^{i^{\s \dprime}}_{\sigma}} \!:\!
&\left\{
\begin{array}{l}
\!\!
    \xhatdot_{\s \Ic} 
    = 
    \A^{\s \Ic}
    \xhat_{\s \Ic} 
 +
 \Hobs^{\s \Ic} (\y^{\, i}_{\sigma}- \yhat^{\, i}_{ \sigma})
 \\
 \! \!
 \yhat^{\, i}_{ \sigma} = \C^{\s  \Ic} \xhat_{\s \Ic} 
 \\
 \! \!
 \res^{\, i}_{ \sigma} =  \y^{\, i}_{\sigma}- \yhat^{\, i}_{ \sigma}
\end{array},
\right.
\\ \label{eq:2hop_obs_IC}
\hspace{-1em} 
\xhat_{\s \Ic}(t_k) 
 \!=\!
    & \begin{cases}
     \ind_{\Ic_i}
     \x_{\s \Ic}(t_k), & \!\!
     \text{if} \ \
     \Vc^{\,i^{\dprime}}_{\sigma(t_k)} \! \neq \! \Vc^{\,i^{\dprime}}_{\sigma(t_{k-1})} \text{ OR} 
     \; k = 0, \\
     \xhat_{\s \Ic}(t_{k}) ,& \!\!
     \text{if} \ \
     \Vc^{\,i^{\dprime}}_{\sigma(t_k)} \!=\! \Vc^{\,i^{\dprime}}_{\sigma(t_{k-1})},
     \end{cases} 
\end{align}
\end{subequations}
where $\xhat_{\s \Ic} $ is the estimation of $\x_{\s \Ic} $ in \eqref{eq:2hop_sys}, and the initial conditions $\xhat_{\s \Ic}(t_k) $ are updated at
$ \braces{{t}_k}_{k=0}^{\mb{m}}, \; \mb{m} \in \intgnonneg$ corresponding to the modes $\sigma(t_k)\text{'s} \in \Qc$ (see Assumption \ref{assum:switching}),
and $\ind_{\Ic_i}= \diag ( I^{}_{\s |\Ic_i|}, \efrak^{\, \s 1}_{\s |\Ic_i|} {\efrak^{\, \s 1}_{\s |\Ic_i|}}^{\s \!\!\! \top} )$. 
$\Hobs^{\s \Ic} =  
\begin{bsmallmatrix}
\zeros_{\s |\Ic|\times|\Ic|} &  \phantom{h}\zeros_{\s |\Ic|} \\ H^{\s \Ic}_{\sigma} & h_{\sigma} {\efrak^{\, \s 1}_{\s |\Ic|}}
\end{bsmallmatrix} $ is the observer's gain matrix with 
a scalar $h_{\sigma} \in \realpos $ and a symmetric P.D. matrix $ H^{\s \Ic}_{\sigma} \in \realpos^{\s |\Ic|\times |\Ic|} $ 
such that $\bar{\mb{A}}^{\s \Ic}_{\sigma}=(\A^{\s \Ic}-\Hobs^{\s \Ic} \C^{\s \Ic})$ is Hurwitz stable in every mode $\sigma \in \Qc$. 
Note that the availability of $\Phi_{\sigt}^{\, i} $ in \eqref{eq:2hop_info} allows each agent to readily update $\Sigma^{\s \Oc}_{\s \Vc^{i^{\s \dprime}}_{\sigma}}$ upon a switch occurs between the communication modes. 

Let estimation error $\e_{\s \Ic} = \x_{\s \Ic} - \xhat_{\s \Ic}$, its dynamics are obtained from \eqref{eq:2hop_sys} and \eqref{eq:2hop_obs} as follows
\noindent
\begin{align} \label{eq:2hop_obs_error}
\Sigma^{\s \widetilde{\Oc}}_{\s \Vc^{i^{\s \dprime}}_{\sigma}}\!:\!
\left\{
\begin{array}{l}
\!\!
    \edot_{\s \Ic} 
    = 
    \bar{\mb{A}}^{\s \Ic}_{\sigma}
    \e_{\s \Ic} 
 + 
 \rho(\x_{\s \Ic } , \x_{\s \Rc})
 + \Bapp \uapp
 \\
 \! \!
 \res^{\, i}_{ \sigma} =  \C^{\s \Ic}  \e_{\s \Ic} 
\end{array} ,
\right.
\end{align}
in which $\e_{\s \Ic}(t_k) = \diag ( \zeros^{}_{\s |\Ic_i|+1\times|\Ic_i|+1}, I^{}_{\s |\Ic_i|-1}) \x_{\s \Ic}(t_k) $
if $ \Vc^{\,i^{\dprime}}_{\sigma(t_k)} $ $ \neq \Vc^{\,i^{\dprime}}_{\sigma(t_{k-1})} $ or $k=0$, and $\e_{\s \Ic}(t_k) = \x_{\s \Ic}(t_k) - \xhat_{\s \Ic}(t_k) $ otherwise, with $ \braces{{t}_k}_{k=0}^{\mb{m}}, \; \mb{m} \in \intgnonneg $. 
%
\begin{theorem}\label{thm:detectability_local}
Consider $\Sigma_{\sigt}$ in \eqref{eq:cl_sys} with a $\boldsymbol{\kappa}(\boldsymbol{\Gc}^{\mu}_{\s T})$-vertex-connected (resp. $r(\boldsymbol{\Gc}^{\mu}_{\s T})$-robust) communication network as defined in \eqref{eq:graph_bounds_all} under Assumptions \ref{assum:switching} and \ref{assum:comm_capability}. Let \eqref{eq:cl_sys} be subject to an $F$-total, where $F\leq \boldsymbol{\kappa}(\boldsymbol{\Gc}^{\mu}_{\s T})-1$, (resp. $F$-local, where $F\leq r(\boldsymbol{\Gc}^{\mu}_{\s T})-1$) adversary set with input $\ua \in \Lc_{pe}$. Let each mobile agent $i \in \Vc$ be equipped with a reconfigurable local attack detector $\Sigma^{\s \Oc}_{\s \Vc^{i^{\s \dprime}}_{\sigma}}$ given by \eqref{eq:2hop_obs} and local information \eqref{eq:2hop_info}. Then, for each $\norm{\e_{\s \Ic}(t_k)} < w_{\s \Ic}$, with $w_{\s \Ic} \in \realpos$, \eqref{eq:2hop_obs_error} is finite-gain $\Lc_{p}$ stable and the residuals $\res^{\, i}_{ \sigma}(t)$'s hold the bound
   \noindent \begin{align}\label{eq:detectability_local_st2}
    \left| \res^{i,j}_{ \sigma}(t)  \right| 
        &\leq  
{\kappa}^{\s \Ic}_{\e}
    w_{\s \Ic}
    e^{- {\lambda}^{\s \Ic}_{\e}  (t-t_k)}
       + 
    \big( {\scriptstyle \frac{{\kappa}^{\s \Ic}_{\res}}{\lambda^{\s \Ic}_{\e}}} \norm{\x(t_0)} e^{-\lambda_{\x}(t_k-t_0)}\big)
    \nonumber \\
    & \ \
    \big(1 - e^{ - \lambda^{\s \Ic}_{\e} ( t -t_k ) }\big)+
\big({\scriptstyle \frac{1+ {\kappa}^{\s \Ic}_{\res}}{\lambda^{\s \Ic}_{\e}}}\big)
\sup_{ t_0 \leq t \leq T_d} \norm{\ua(t)}
\nonumber\\
& \ \
   \big(1 - e^{ - \lambda^{\s \Ic}_{\e} ( t -t_k ) }),
   \ \ \forall \, t \in [t_k, \; t_{k+1}\big),
    \end{align}
    where $\res^{i,j}_{ \sigma}(t)$ is the $j$-th component of $\res^{\,i}_{ \sigma}(t)$ and denotes the position estimation of the two-hop neighbors, corresponding to the $j$-th row of $\C^{\, i}$, in each mode $\sigma(t_k) \in \Qc,\; \forall \, t \in [t_k \; t_{k+1}), \; k \in \intgnonneg $. Also,
     $ {\kappa}^{\s \Ic}_{\res} = \alpha  {\kappa}_{\x} {\kappa}^{\s \Ic}_{\e} $, with the known constants\footnote{Recall that $\Ic$ is a shorthand for the set $ \Ic_i = \Vc^{\, i^{\dprime}}_{\sigma} $ and thus the constants are mode-dependent for each cooperative agent $i\in \Vc\setminus\Ac$.} $ {\kappa}^{\s \Ic}_{\e}, \, {\lambda}^{\s \Ic}_{\e} \in \realpos $, such that  $ \| e^{\bar{\mb{A}}^{\s \Ic}_{\sigma(t_k)}(t-t_k)} \| \leq {\kappa}^{\s \Ic}_{\e} e^{- {\lambda}^{\s \Ic}_{\e}  (t-t_k)} $, and $ {\kappa}_{\x} $ and ${\lambda}_{\x} $ are given in Proposition \ref{prop:convergence_PE_like}.
Additionally, if $\Ac = \emptyset$, each $\Sigma^{\s \Oc}_{\s \Vc^{\, i^{\s \dprime}}_{\sigma}}$ is exponentially stable with $\e_{\s \Ic}(t) \rightarrow \zeros$, and the residuals in \eqref{eq:detectability_local_st2} hold the following bound
\begin{align}\label{eq:res_threshold}
       & \left| \res^{i,j}_{ \sigma}(t)  \right| 
        \leq   
    {\kappa}^{\s \Ic}_{\e}
    w_{\s \Ic} e^{- {\lambda}^{\s \Ic}_{\e}  (t-t_k)}
       + 
    \big({\scriptstyle \frac{{\kappa}^{\s \Ic}_{\res}}{\lambda^{\s \Ic}_{\e}}} \norm{\x(t_0)} e^{-\lambda_{\x}(t_k-t_0)}\big)
    \nonumber \\
    & \hspace{11em}
    \big(1 - e^{ - \lambda^{\s \Ic}_{\e} ( t -t_k ) }\big)
:= \epsilon^{i,j}_{\sigma}, 
    \end{align}
    where $\epsilon^{i,j}_{\sigma}$ is a threshold that can be used in \eqref{eq:hypotheses_local}.
\end{theorem}
\begin{proof}
See Appendix \ref{app_detectability_local}.
\end{proof}
\rev{Theorem \ref{thm:detectability_local} shows that the local observer \eqref{eq:2hop_obs} with residual $\res^{\, i}_{ \sigma}$ has BIBO stability for the worst-case number of malicious agents 
with deception attacks that are defined in Section \ref{Sec:adversary_model}, and that whose detectability is ensured by a certain degree of network connectivity (see Lemma \ref{lemma:net_level_obs}).}

\begin{algorithm}[t]
\small
\renewcommand{\algorithmicrequire}{\textbf{Input:}}
\renewcommand{\algorithmicensure}{\textbf{Require:}}
\def\NoNumber#1{{\def\alglinenumber##1{}\State#1}\addtocounter{ALG@line}{-1}}
\caption{\small
{\color{teal}Res}ilient {\color{teal}C}onsensus \& Cooperation Cooperation over {\color{teal}U}nreliable N{\color{teal}e}tworks
}\label{alg:rescue}
\begin{algorithmic}[1]
\Require $ \Phi_{\sigt}^{\, i} $ in \eqref{eq:2hop_info}, $\Sigma^{\s \Oc}_{\s \Vc^{i^{\s \dprime}}_{\sigma}}$ in \eqref{eq:2hop_obs}, and $ \epsilon^{i,j}_{\sigma}$ in  \eqref{eq:res_threshold}, $\forall \, i \in \Vc\setminus\Ac$

        \State { \color{gray}\footnotesize // Accept the null $\nullH$ in \eqref{eq:hypotheses_local_null}  and assume $\Ac^{\dprime} = \emptyset $ }
\Procedure{1: distributed Detection \& Isolation}{}
   \State { \color{gray} \hspace{-2em} \footnotesize // Use the most recent info $\Phi_{\sigt}^{\, i}$ to (re)initialize $\Sigma^{\s \Oc}_{\s \Vc^{\dprime}_{i}}$}
   \vspace{-1ex}
  	\State \hspace{-2em} Compute the residual ${\res}^{\, i}_{\sigma}(t)$ and the corresponding thresholds $\epsilon^{i,j}_{\sigma}$ 
   \For{$j \in \onehop{i}$} 
    \If{$ |{\res}^{i,j}_{\sigma}(t)|  > {\epsilon}^{i,j}_{\sigma} $}
        \State Reject the null hypothesis $\nullH$ in \eqref{eq:hypotheses_local_null} \Comment{{\color{gray} \footnotesize 
        Detection
        }}
        \State $\Ac^{\dprime} \leftarrow j \in \onehop{i}$
        \State Set $ a^{\sigma}_{ij} = 0,\; j \in \Ac^{\dprime} \cap \onehop{i}$ \Comment{{\color{gray}\footnotesize
        Stop comm. with $\!\Ac^{\dprime} $}}
        \State Update $\Phi_{\sigt}^{\, i}$
    \EndIf
    \EndFor
\EndProcedure
%
\Procedure{2: Resilient Cooperation defined in \eqref{eq:formation_consensus}}{}
    \State Run $\ub^{\rm{n}}_{i}(t) $ given in \eqref{eq:ctrl_proto} with the information from $\onehop{i} \setminus \Ac^{\dprime}$
\EndProcedure
\end{algorithmic}
\normalsize 
\end{algorithm}
\begin{figure}[t]
  \centering 
\includegraphics[width=.49\textwidth]{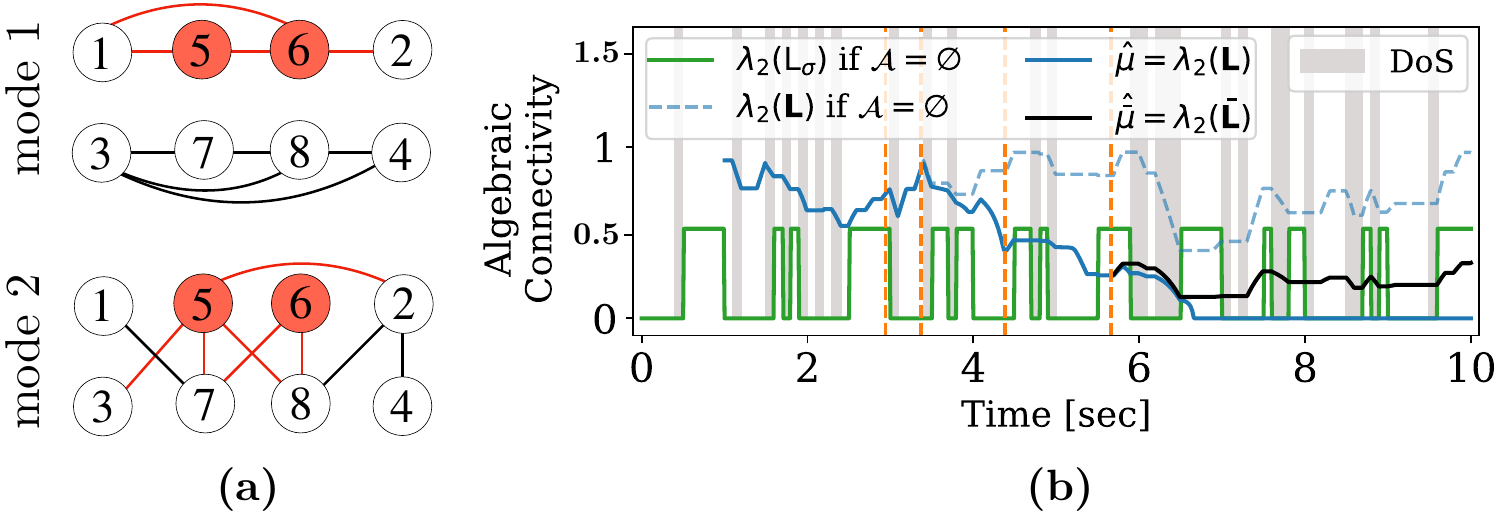} 
  \caption{\small Communication network $\Gc_{\sigt}$ in (a) and its algebraic connectivity in the integral sense of \eqref{eq:PE_cond} in (b) \rev{for Section \ref{Sec:sim_studies}-Example 1}. \rev{(a) The network switches between two modes every $0.5~\si{sec}$ whose union forms a static overlay network $\boldsymbol{\Gc}_{\s T}^{\mu}$ with $\lambda_{2}(\boldsymbol{\lap})=2.1049$ that is 3-robust \cite[Fig. 4]{leblanc2013resilient}}, ensuring $(3,1)$-robustness, and $(3,1)$-vertex-connectivity (see \rev{Section \ref{sec:net_res_and_stability}} and \eqref{eq:graph_bounds_all}). \rev{per Section \ref{Sec:adversary_model},}
  the network $\Gc_{\sigt}$ is subject to a $2$-total and $2$-local set of malicious agents $\Ac = \braces{5,6}$. It is also subject to a distributed DoS whose link dropouts follow a binomial distribution with $100$ trials and a success probability of $0.3$ during $10~\si{sec}$.
(b) The illustration of positive algebraic connectivity $\lambda_{2}(\cdot)$ in the integral sense \eqref{eq:PE_cond} for
    for the network $\Gc_{\sigt}$ and its induced network $\bar\Gc_{\sigt}$ in \eqref{eq:comm_net_res} despite their intermittent connections (\rev{See also remark} \ref{rmk:PE_connectivity}). 
  The results in (b) are from resilient consensus in Fig. \ref{fig:exp1_consensus}a through Algorithm \ref{alg:rescue}. 
  The decrements in $\!\lambda_{2}(\cdot)$ during $t \!\in\! [0, \, 5.66]$ are due to the permanent link disconnections that occurred in the attack detection and isolation procedure, \rev{see Fig. \ref{fig:exp1_consensus}a}.
  }\label{fig:ex1_comm_net} 
\end{figure}

\section{Resilient Cooperation}\label{Sec:resilient_cooperation} 
Building upon the results in the previous sections, we present an algorithmic framework, summarized in Algorithm \ref{alg:rescue}, as a solution to the resilient cooperation problem stated in Section \ref{sec:problem_statement}. 
Algorithm \ref{alg:rescue} comprises two simultaneous procedures addressing the distributed detection and isolation of malicious agents by using \eqref{eq:hypotheses_local} for decision-making, and resilient cooperation. In what follows, we present the technical discussions of Algorithm \ref{alg:rescue}.

\textbf{Isolation of the set of the malicious agents $\Ac \subset \Vc$}.
Upon detection of neighboring malicious agents by each cooperative agent $i\in \Vc\setminus\Ac$, there follows the isolation (removal) of the detected malicious agents from the network (Lines 7-10 in Algorithm \ref{alg:rescue}). 
Note that the results in Proposition  \ref{prop:observability_local}, \rev{Lemma \ref{lemma:net_level_obs},} and Theorem \ref{thm:detectability_local} allow each cooperative agent $i \in \Vc\setminus\Ac$ in a \rev{$(F+1,T)$-robust (resp. $(F+1,T)$-vertex-connected) network} to perform the distributed hypothesis testing in \eqref{eq:hypotheses_local} and detect a \emph{candidate} set of malicious agents within its $1$-hop neighbors, provided the \emph{actual} set of malicious agents, $\Ac$, is at most $F$-local (resp. $F$-total).
Here, the distinction between a \emph{candidate} set and the \emph{actual} set of malicious agents is due to the possibility of \emph{false} alarms in \eqref{eq:hypotheses_local}. 
(i.e., a \emph{candidate} set is almost always a superset of the \emph{actual} set for a sufficiently small threshold in \eqref{eq:hypotheses_local}).  
The foregoing sets coincide if $\bar{\mb{A}}^{\s \Ic}_{\sigma} = (\A^{\s \Ic}-\Hobs^{\s \Ic} \C^{\s \Ic})$ in \eqref{eq:2hop_obs} features distinct eigenvalues, guaranteeing that each residual's component $\res^{i,j}_{\sigma}$ in \eqref{eq:detectability_local_st2} is most sensitive to only one of the input directions associated with the malicious agents within the $1$-hop neighbors. 

%
%
\begin{figure*}[ht]
\centering
\includegraphics[width=.95\textwidth]{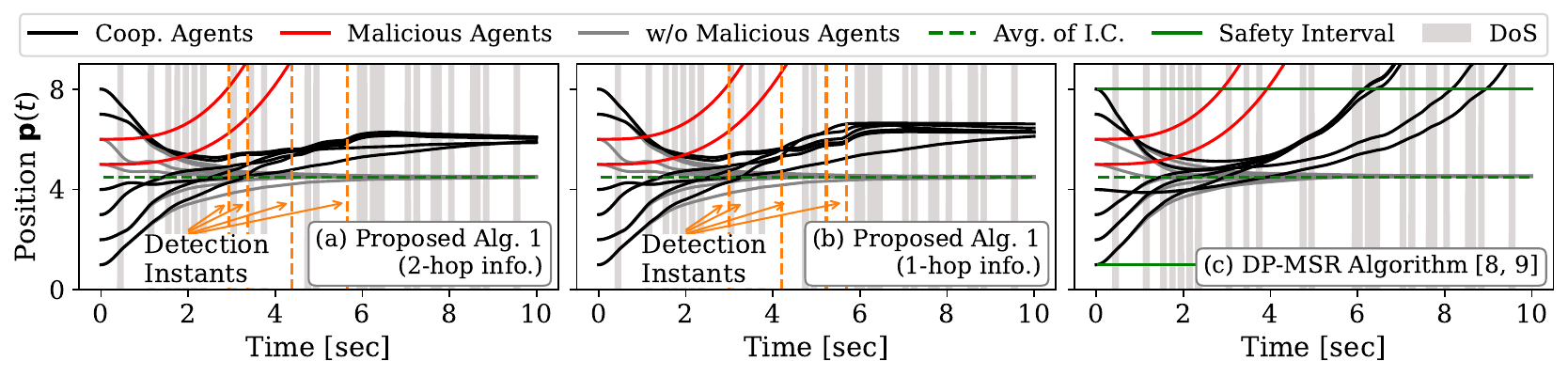} 
   \caption{\small \rev{\textbf{Example 1}: 
   Comparison of resilient consensus in an 8-agent network $\Gc_{\sigt}$ that is, as shown in Fig. \ref{fig:ex1_comm_net}, $(3,1)$-robust and subject to DoS attacks and a $2$-total and $2$-local set of malicious agents $\Ac=\{5,6\}$} with $\ub_5(t) = 0.3 t $ and $\ub_6(t) = 0.5 t $ in \eqref{eq:ctrl_proto}.
   \rev{(a)} Resilient consensus using Algorithm \ref{alg:rescue} \rev{whose resilient to the $2$-total/$2$-local set $\Ac$ in the $(3,1)$-robust network is guaranteed by Lemma \ref{lemma:net_level_obs} and Theorem \ref{thm:detectability_local}}. Also, the {vertical orange dashed lines} specify the time instants where cooperative agents detected and disconnected from their respective neighboring malicious agents (lines 7-10 of Algorithm \ref{alg:rescue} with $\epsilon^{i,j}_{\sigma}=0.95$ as shown in Fig. \ref{fig:_residuals}) using its local attack detector in \eqref{eq:2hop_obs}.
   \revv{(b) Resilient consensus using Algorithm \ref{alg:rescue} with only $1$-hop information, i.e. $\graphx{i}$ in \eqref{eq:2hop_info} and $\Ic_i = \Vc^{\, i^{\prime}}_{\sigma}$ in \eqref{eq:locally_measurements} and \eqref{eq:2hop_sys}, and with the same threshold $\epsilon^{i,j}_{\sigma}=0.95$ as in (a). (c)} Resilient consensus using the DP-MSR algorithm that for a $3$-robust network has provable resilient consensus only in the presence of up to $1$-local or $1$-total malicious agents \cite{dibaji2017resilient,dibaji2015consensus}, accounting for the failure of the approach in this case where $\Ac$ is $2$-local and $2$-total.
   We note that the analysis of resilient consensus via the DP-MSR algorithm was originally developed for a discretized version of \eqref{eq:cl_sys} in \cite{dibaji2017resilient,dibaji2015consensus} while our results are in the continuous-time domain. To have the results in a comparable time scale, we used the DP-MSR procedure with the small sample time $T_{\rm s}=0.001$ and the gains $\gamma =3$ and $\alpha=1$ in the zero-order-hold discretization of \eqref{eq:ctrl_proto}. This set of parameters does not completely satisfy the sufficient condition in \cite[eq. (9)]{dibaji2017resilient}, but does satisfy a relaxation thereof, similar to the discussion in a footnote in \cite{dibaji2017resilient}. This enables an asymptotic resilience consensus in the case $\Ac=\emptyset$ (shown with the gray-colored state trajectories) and also in the cases of $(F\!=\!1)$-local and $(F\!=\!1)$-total adversary sets (not shown herein) over any $3$-robust network. 
   }\label{fig:exp1_consensus}
   \subfloat[Results for the 2-hop information case in Fig. \ref{fig:exp1_consensus}a \label{fig:ex1-1-res1}]{\small \centering
    \includegraphics[width=.45\linewidth, valign=t]{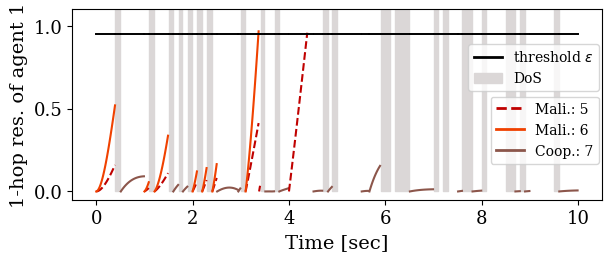}
  }
  \quad
\subfloat[Results for the 1-hop information case in Fig. \ref{fig:exp1_consensus}b \label{fig:ex1-1-res2}]{\small \centering
\includegraphics[width=.45\linewidth, valign=t]{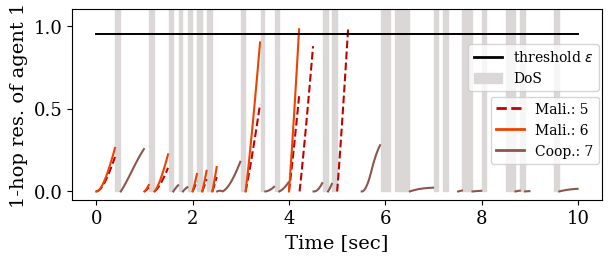}
  }
  \caption{\small The absolute value of the observer's residuals $\res^{i,j}_\sigma$ for the 1-hop neighbors of agent 1 for the results of Fig. \ref{fig:exp1_consensus}. These residuals are is used in line 6 of Algorithm \ref{alg:rescue}.} \label{fig:_residuals}
\end{figure*}

\textbf{State-dependent switching}.
We note the isolation process based on \eqref{eq:hypotheses_local} (lines 7-10 in Algorithm \ref{alg:rescue}) imposes
a finite number of state-dependent switches that are not explicitly incorporated in the condition \eqref{eq:PE_cond} for $\lap_{\sigt}$ with time-dependent switches $\sigt: \mathbb{R}_{\geq 0}\rightarrow  \Qc $. On the other hand, the results of Theorem \ref{thm:robustness-to-disconection} for the bound on the network connectivity in the integral sense of \eqref{eq:PE_cond} after node and edge removal holds independent of the type of switches. Therefore, upon a link removal between a cooperative and malicious agent(s), there exists a new Laplacian matrix $\lap_{\sigt}$ that holds the connectivity condition of the from \eqref{eq:PE_cond} for 
the system in \eqref{eq:cl_sys} starting from the new initial condition $\x(t_k) \in \real^{2|\Vc|}$ 
with $t_k$, $k\in \intgnonneg$, being the time instant of the newly active mode $\sigma(t_k)\in \Qc$. Having the integral connectivity as in \eqref{eq:PE_cond} independent of the states' initial condition, Proposition \ref{prop:convergence_PE_like} can be applied. It is worth mentioning that the independence from the states' initial conditions for the $(\mu,T)$-PE connectivity in \eqref{eq:PE_cond} is a special case of having \eqref{eq:PE_cond} parameterized of the form $ \frac{1}{T}\int_{t}^{t+T} Q\lap_{\sigma(\tau, \boldsymbol{\lambda})} Q^{\top} \, {\rm d}\tau \geq \mu I_{N-1},   \, \forall \, t\in \realnonneg, $ that holds for each $\boldsymbol{\lambda}:=(t_{\rm o},\x_{\rm o})\neq(t_{0},\x(t_0))$ with the switching signal $\sigma(t,\x(t_k)): \mathbb{R}_{\geq 0} \times \Xc \rightarrow  \Qc $, $\Xc \subset \real^{2|\Vc|}$ (see \cite{loria2002uniform}).
\begin{figure*}[ht]
  \centering 
\includegraphics[width=.98\textwidth]{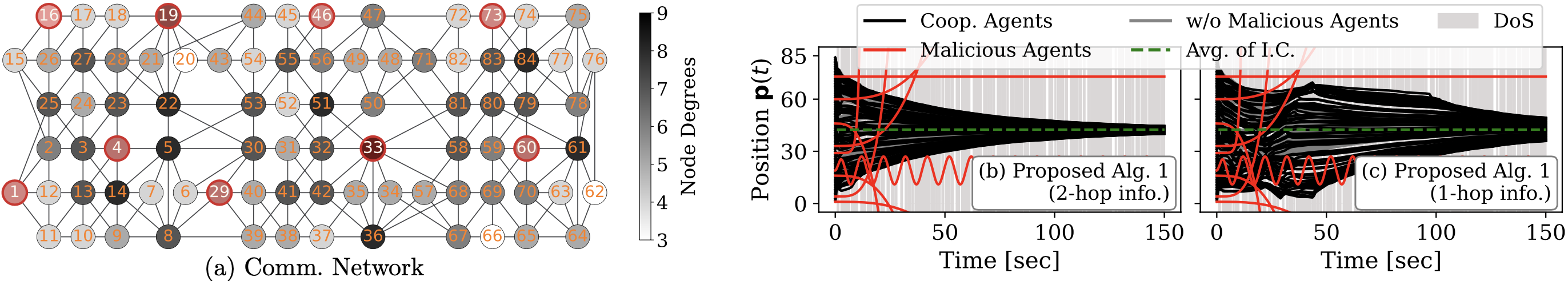}
  \caption{\small 
  \rev{\textbf{Example 2}: Resilient consensus in an 84-agent network $\Gc_{\sigt}$
  subject to deception and DoS attacks defined in Section \ref{Sec:adversary_model}.
  The deception attacks are introduced by a 1-local set of 9 malicious agents, $\Ac=\{1,4,16,19,29,33,46,60,73\}$, which are shown in red color. The distributed DoS attack \eqref{eq:DoS_state} imposes link dropouts following a binomial distribution with $600$ trials and a success probability of $0.4$. 
  (a) The static overlay network $\boldsymbol{\Gc}_{\s T}^{\mu}$ is $2$-robust, constructed using the preferential-attachment model in \cite[Thm. 5]{leblanc2013resilient} based on the topology in \cite[Fig. 6]{leblanc2013resilient}. Despite intermittent connections, the network $\Gc_{\sigt}\!$ is $(2,1)$-robust and $(3,1)$-vertex-connected (see Definitions \ref{def:r_T_robust} and \ref{def:k_T_conect}, and Lemma \ref{lemma:PE_equvalence}). $(2,1)$-robustness, then, ensures resilience to any 1-local set $\Ac$ as it follows from Lemma \ref{lemma:net_level_obs} and Theorem \ref{thm:detectability_local}. (b) Resilient consensus using Algorithm \ref{alg:rescue} (with the threshold $\epsilon^{i,j}_{\sigma}\!=\!10 e^{-t}+0.95$ as shown in Fig. \ref{fig:ex2-2-res1}) over the intermittent network $\Gc_{\sigt}\!$ in (a) and in the presence of the 1-local malicious set $\Ac$.
  \revv{(c) Resilient consensus using Algorithm \ref{alg:rescue} with only $1$-hop information, i.e. $\graphx{i}$ in \eqref{eq:2hop_info} and $\Ic_i = \Vc^{\, i^{\prime}}_{\sigma}$ in \eqref{eq:locally_measurements} and \eqref{eq:2hop_sys}, and with the threshold $\epsilon^{i,j}_{\sigma}\!=\!30 e^{-0.1t}+1.5$ as shown in Fig. \ref{fig:ex2-2-res2}. The results suggest that Algorithm \ref{alg:rescue}'s detection capability is maintained with minimal local information, albeit with some performance degradation in resilient consensus if a more conservative threshold is used for larger networks where the effect of the coupling term $\rho(\x_{\s \Ic }, \x_{\s \Rc})$ in \eqref{eq:2hop_sys} and \eqref{eq:2hop_obs_error} is more significant during transient periods.}
  }
 }\label{fig:exp2_consensus} 
 \subfloat[Results for the 2-hop information case in Fig. \ref{fig:exp2_consensus}b \label{fig:ex2-2-res1}]{\small \centering
    \includegraphics[width=.45\linewidth, valign=t]{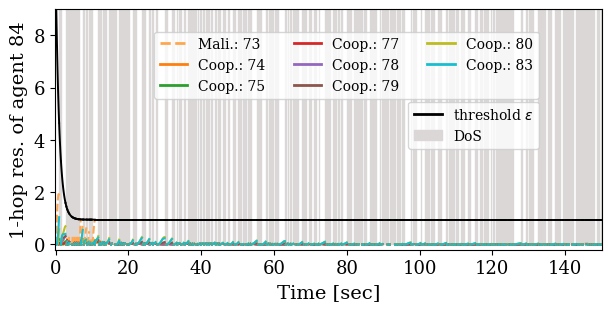}
  }
  \quad
\subfloat[Results for the 1-hop information case in Fig. \ref{fig:exp2_consensus}c \label{fig:ex2-2-res2}]{\small \centering
\includegraphics[width=.45\linewidth, valign=t]{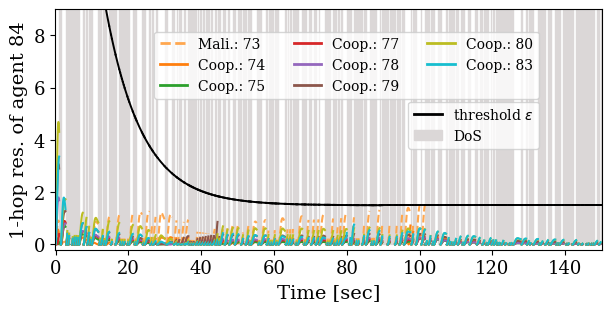}
  }
  \caption{\small The absolute value of the observer's residuals $\res^{i,j}_\sigma$ for the 1-hop neighbors of agent 84 for the results of Fig. \ref{fig:exp2_consensus}. This is used in line 6 of Algorithm \ref{alg:rescue}.
  The results suggest that Algorithm \ref{alg:rescue}'s detection capability is maintained with minimal local information, albeit with some performance degradation in resilient consensus if a more conservative threshold is used for larger networks where the effect of the coupling term $\rho(\x_{\s \Ic }, \x_{\s \Rc})$ in \eqref{eq:2hop_sys} and \eqref{eq:2hop_obs_error} is more significant during transient periods.}   
\end{figure*}

\section{Simulation Results}\label{Sec:sim_studies}
\rev{We conduct two simulation studies to illustrate the theoretical results and compare them with the state-of-the-art \cite{mao2020novel,dibaji2017resilient}. We also provide the open-source code at
{\href{https://github.com/SASLabStevens/rescue}{\color{teal}https://github.com/SASLabStevens/rescue}}.
}

\rev{\textbf{Example 1}.} 
\rev{We compare our proposed Algorithm \ref{alg:rescue} with the DP-MSR algorithm \cite{dibaji2015consensus,dibaji2017resilient}. Employing Algorithm \ref{alg:rescue}, we achieve resilient consensus in an 8-agent network that is subject to switching topology and both deception and DoS attacks. See Figs. \ref{fig:ex1_comm_net} and \ref{fig:exp1_consensus}(a).
In contrast, the DP-MSR algorithm fails to achieve resilient consensus for the same network (Fig. \ref{fig:exp1_consensus}(c)), despite the advantage of operating over a network that is the union of the two modes shown in Fig. \ref{fig:ex1_comm_net}(a). The outperformance of Algorithm \ref{alg:rescue} is because of its observer-based nature that leverages local information $\Phi_{\sigma}^{\, i}$ in \eqref{eq:2hop_info} to detect a larger set of malicious agents in a network with a specific degree of connectivity and $r$-robustness (see Lemma \ref{lemma:net_level_obs}), a capability not shared by the DP-MSR algorithm.
} \revv{Additionally, we ran Algorithm \ref{alg:rescue} using only $1$-hop information instead of the $2$-hop information in \eqref{eq:2hop_info}. The results, shown in Fig. \ref{fig:exp1_consensus}(b), demonstrate that Algorithm \ref{alg:rescue}'s detection capability is maintained even with minimal local information.
}

\rev{\textbf{Example 2}. We evaluate the scalability of our framework on an 84-agent network subject to a DoS attack and 9 malicious agents that form a 1-local set, see Fig. \ref{fig:exp2_consensus}. 
Notably, each agent has at most 9 neighbors (less than $11\%$ of total agents), resulting in a sparse graph. Despite sparsity, the graph has the required robustness properties. This observation underscores the significance of sparse graphs with strong robustness/connectivity properties (e.g., expander graphs), offering resilience without excessive communication overhead.}

\revv{Additionally, we note that if a malicious agent sends incorrect structural information about its $1$-hop neighbors, it causes model inconsistency between the receiving neighbors' local observers \eqref{eq:2hop_obs} and the true multi-agent dynamics \eqref{eq:2hop_sys}-\eqref{eq:rest_sys}. This model discrepancy, causing nontrivial residuals for \eqref{eq:2hop_obs}, can be leveraged for attack detection \cite{mao2020novel,bahrami2021privacy}.}  

\rev{
\textbf{Scalability and computational complexity}.
We remark that the proposed Algorithm \ref{alg:rescue} improves the scalability of the system-theoretic frameworks relying on observers for attack detection \cite{pasqualetti2011consensus,pasqualetti2013attack,mao2020novel}. 
Note that, for each agent, attack detection in Algorithm \ref{alg:rescue} (lines 4-10), requires only one observer with a 2-hop dynamics $\A^{\s \Ic}$ in \eqref{eq:2hop_obs} with worst-case complexity $\Oc(|\Vc^{\,i^{\dprime}}_{\sigma}|^2)$ rather than the complete model $\A$ in \eqref{eq:cl_sys} with $\Oc(|\Vc|^2)$, ($|\Vc^{\,i^{\dprime}}_{\sigma}|\leq|\Vc|$, see \eqref{eq:local_graphs}), which is the case in \cite{mao2020novel}. 
Moreover, the local information \eqref{eq:2hop_info} allows for detecting a greater number of malicious agents in a given network, compared to the prior work \cite{pasqualetti2011consensus} including the graph-theoretic MSR-like algorithms \cite{dibaji2015consensus,dibaji2017resilient,saldana2017resilient} (see Fig. \ref{fig:exp1_consensus}), whose worst-case complexity is quadratic in time $\Oc(|n|^2)$ and linear in space $\Oc(n)$, w.r.t. the size of inclusive $1$-hop neighbors \cite{leblanc2012low}, i.e. $n=|\Vc^{\, i^{\prime}}_{\sigma}|$, see \eqref{eq:local_graphs}.

Finally, given the switching nature of the local observer \eqref{eq:2hop_obs} with resetting initial conditions, an increased frequency of topology switching, potentially violating Assumption \ref{assum:switching}, would lead to significant performance degradation in attack detection as observer's residuals would persist in a transient convergence phase.
}

\section{Conclusions}\label{Sec:conclusion}
We considered the consensus and formation of multi-agent systems over a time-varying communication network subject to deception and DoS attacks.
We showed, for a given integer number $F$, the communication network requires to be at least $(F+1)$-vertex-connected (resp. $(F+1)$-robust) in an integral sense and uniformly in time over a period of time $T$ for resilience to an $F$-total (resp. $F$-local) adversary set. We presented theoretical guarantees and explicit bounds for exponentially fast convergence to the consensus/formation equilibrium. We also presented a distributed attack detection framework with theoretical guarantees, allowing for resilient cooperation. 
%
\appendices
\section{proof of Lemma \ref{lemma:PE_equvalence}}\label{app_PE_equvalence}
The proof of the equivalence of the statements is achieved by demonstrating that each ensures positive algebraic connectivity in an integral sense for the graph $\Gc_{\sigt}=(\Vc, \Ec_{\sigt})$, that is $\lambda_2 ( \boldsymbol{\lap}= \frac{1}{T} \int_{t}^{t+T} 
         \lap_{\sigma(\tau)}  \, {\rm d}\tau ) > \mu$ holds  $\forall \,  t \in \realnonneg  $ and $\exists \, \mu,T \in \realpos$ as defined in \eqref{eq:PE_cond}.
We provide a sketch of the proof herein and refer to \cite{RB_PHDthesis2024} for details of derivations.

Note that $\ones^{}_{N}\ones^{\top}_{N}/{N}$ and $ Q^{\top}Q^{}=I_N - \ones^{}_{N}\ones^{\top}_{N}/{N} $ are both orthogonal projection matrices with $\spec(\ones^{}_{N}\ones^{\top}_{N}/N )= 
\{\zeros^{\top}_{N-1},1\}$ and $\spec(Q^{\top}Q^{}) = 
\{0,\ones^{\top}_{\s N-1}\}$. Then $\spec(Q^{} \lap_{\sigma(\tau)}Q^{\top}  ) = \spec( \lap_{\sigma(\tau)}  )\setminus \{0\} \implies \nonumber 
\lambda_1(Q^{} \lap_{\sigma(\tau)}Q^{\top}) = \lambda_2( \lap_{\sigma(\tau)})$.

\ref{itm:PE_cond_alt} $\! \implies \!$ \ref{itm:PE_cond}: if \ref{itm:PE_cond_alt} holds, then $ (1/T)\int_{t}^{t+T}
        ( \lap_{\sigma(\tau)} + \ones^{}_{N}\ones^{\top}_{N}/{N} ) \, {\rm d}\tau  $ is positive definite, and it can be shown, similar to \cite[Thm. 1]{anderson2016convergence}, that $\lambda_2 \paren{ {{\frac{1}{T}\int_{t}^{t+T} 
         \lap_{\sigma(\tau)}  \, {\rm d}\tau}} } \geq  \mu_m$.
Also, \eqref{eq:PE_cond_alt} can be rewritten as 
\begin{align*}
\hspace{-1em}
    \mu_{m} I_{N} - \frac{\ones^{}_{N}\ones^{\top}_{N}}{{N}}   \leq 
    \frac{1}{T}\int_{t}^{t+T} 
     \lap_{\sigma(\tau)}  \, {\rm d}\tau \leq
     \mu_{M} I_{N} - \frac{\ones^{}_{N}\ones^{\top}_{N}}{{N}}, 
\end{align*}
which by pre- and post-multiplying, resp., by $Q$ and $Q^{\top}$, yields
\noindent
\begin{align*}
    \mu_{m} I_{N-1} \leq 
    \frac{1}{T}\int_{t}^{t+T} 
    Q^{} \lap_{\sigma(\tau)}Q^{\top}  \, {\rm d}\tau \leq
      \mu_{M} I_{N-1},  \ \ \forall \, t\in \realnonneg,
\end{align*}
that is equivalent to \eqref{eq:PE_cond} with $\mu_m=\mu$. The existence of the upper bound, $\mu_{M} I_{N-1}$, for \eqref{eq:PE_cond} is trivial because of the boundedness of $a^{\sigma(t)}_{ij}$'s in the adjacency matrix and the integration over a finite interval.

\ref{itm:PE_cond} $\! \implies \!$ \ref{itm:PE_cond_alt}: It is similar to the $\! \impliedby \!$ part and thus omitted.

\ref{itm:PE_cond_alt} $\! \iff \! $ \ref{itm:PE_cond_alt_edge}: See \cite[Thm. 1]{anderson2016convergence}. We conclude by restating that
\begin{align*}
\lambda_2 \paren{
    {
    \frac{1}{T}\int_{t}^{t+T} 
         \lap_{\sigma(\tau)}  \, {\rm d}\tau} 
         }
         \overset{\eqref{eq:PE_lap_adj}}{=}
        &
        \lambda_2(\boldsymbol{\lap})
         = 
        \nonumber \\ 
        &
    \lambda_1 \paren{  { \frac{1}{T}\int_{t}^{t+T} 
         Q\lap_{\sigma(\tau)}Q^{\top}  \, {\rm d}\tau } }    
         \geq  \mu.
\end{align*}

\section{Proof of Proposition \ref{prop:graph_bounds}}\label{app_graph_bounds}
It follows from Lemma \ref{lemma:PE_equvalence} that a $(\mu,T)$-PE connected $ \Gc_{\sigt}$ forms, uniformly in time, 
the connected static graph $\boldsymbol{\Gc}_{\s T}^{\mu} = (\Vc, \boldsymbol{\Ec}_{\s T}^{\mu})$ with the edge set $\boldsymbol{\Ec}_{\s T}^{\mu}$ in \eqref{eq:PE_cond_alt_edge} and algebraic connectivity $\lambda_2(\boldsymbol{\lap}) \geq \mu >0 $. Then, from
\cite{shahrivar2017spectral} and \cite[Thm. 2]{saulnier2017resilient} we have
$\ceil{\frac{\lambda_2(\boldsymbol{\lap})}{2}} \leq r(\boldsymbol{\Gc}_{\s T}^{\mu})$
for $0\leq r(\boldsymbol{\Gc}_{\s T}^{\mu}) \leq \ceil{|\Vc|/2}$, and one can conclude from \cite{shahrivar2017spectral} and \cite[Thm. 6]{leblanc2013resilient} that
$ r(\boldsymbol{\Gc}_{\s T}^{\mu}) \leq \boldsymbol{\kappa}(\boldsymbol{\Gc}_{\s T}^{\mu})$.
It also follows from \cite[Ch. 13.5]{godsil2001algebraic} for any simple non-complete graph $\boldsymbol{\Gc}_{\s T}^{\mu}$, that $\lambda_2(\boldsymbol{\lap}) \leq
         \boldsymbol{\kappa}(\boldsymbol{\Gc}_{\s T}^{\mu}) \leq |\Vc|-1$.
Finally, note that $0=\lambda_{1}(\boldsymbol{\lap}) < \lambda_{2}(\boldsymbol{\lap}) \leq \dots \leq \lambda_{N}(\boldsymbol{\lap}) \leq |\Vc|=N$ holds for $\boldsymbol{\Gc}_{\s T}^{\mu}$ and that the equality $\lambda_{2}(\boldsymbol{\lap})=N$ holds for complete graphs \cite[Corr. 13.1.4]{godsil2001algebraic}. Then \eqref{eq:graph_bounds_all} is concluded.

\section{Proof of Theorem \ref{thm:robustness-to-disconection}}\label{app_robustness-to-disconection}
Recall from Proposition \ref{prop:graph_bounds} that a $(\mu,T)$-PE connected $\Gc_{\sigt}$ forms, uniformly in time, a static network $\boldsymbol{\Gc}_{\s T}^{\mu} = (\Vc, \boldsymbol{\Ec}_{\s T}^{\mu})$ whose robustness $r(\boldsymbol{\Gc}_{\s T}^{\mu})$ and vertex-connectivity $\boldsymbol{\kappa}(\boldsymbol{\Gc}_{\s T}^{\mu})$ are lower bounded by $ \ceil{ \mu/2 }$. 
    It then follows from 
    \cite[Thm. 1]{zhang2015notion}
    for an $r$-robust network $\boldsymbol{\Gc}_{\s T}^{\mu}$, where $r\leq r(\boldsymbol{\Gc}_{\s T}^{\mu})$, that, by definition, no $F$-local adversary subset $\Ac \subset\Vc$, where $F \leq r-1$, make a vertex cutset for $\boldsymbol{\Gc}_{\s T}^{\mu}$.
    Therefore, the removal of up to $F\leq r-1$ malicious nodes (agents) and their incident edges from the neighbors of the remaining nodes (cooperative agents) does not render the induced subgraph $\bar{\boldsymbol{\Gc}} =( \Vc\setminus\Ac, \bar{\boldsymbol{\Ec}})$ disconnected (equiv., $\lambda_2(\bar{\boldsymbol{\lap}}) > 0$, where $\bar{\boldsymbol{\lap}}$ is the Laplacian matrix of $\bar{\boldsymbol{\Gc}}$).
    %
    
    Likewise, given the vertex-connectivity $\boldsymbol{\kappa}(\boldsymbol{\Gc}_{\s T}^{\mu}) $, it follows, by definition, that no $F$-total adversary subset $\Ac \subset\Vc$, with $F \leq  \boldsymbol{\kappa}-1 \leq \boldsymbol{\kappa}(\boldsymbol{\Gc}_{\s T}^{\mu})-1 $, make a vertex cutset for $\boldsymbol{\Gc}_{\s T}^{\mu}$.  
    Therefore, the removal of up to $F \leq \boldsymbol{\kappa}-1$ (malicious) nodes, in total, and their incident edges does not render the induced subgraph $\bar{\boldsymbol{\Gc}} =( \Vc\setminus\Ac, \bar{\boldsymbol{\Ec}})$ disconnected (equivalently, $\lambda_2(\bar{\boldsymbol{\lap}}) > 0$).
    
    Finally, note that the induced subgraph $ \bar\Gc_{\sigt}=(\bar\Vc, \bar\Ec_{\sigt}) $ associated with the connected graph $\bar{\boldsymbol{\Gc}}$ with $0<\lambda_2(\bar{\boldsymbol{\lap}})=:\bar\mu$ meets the conditions in Lemma \ref{lemma:PE_equvalence}-({\rm{iii}}) for some $\bar{T}\leq T$ (where the inequality holds because $|\bar\Vc|<|\Vc|$ and $|\bar{\boldsymbol{\Ec}}|<|\boldsymbol{\Ec}_{\s T}^{\mu}| $).  
    Moreover, it follows from \cite[Thm. 13.5.1]{godsil2001algebraic} for the graph $\boldsymbol{\Gc}_{\s T}^{\mu}$ and its induced subgraph $\bar{\boldsymbol{\Gc}}$, resp., with $\boldsymbol{\lap}$ and $\bar{\boldsymbol{\lap}}$ that $\lambda_2(\boldsymbol{\lap}) \leq \lambda_2(\bar{\boldsymbol{\lap}}) +|\Ac|$. Then, from $\lambda_2(\boldsymbol{\lap}) \geq \mu$ (see \eqref{eq:graph_bounds_all}) and  $\bar\mu=\lambda_2(\bar{\boldsymbol{\lap}})>0$, one can conclude 
    $ \mu \leq \bar\mu +  |\Ac|$.

\section{proof of Proposition \ref{prop:convergence_PE_like}}\label{app_convergence_PE_like}
The first part of the proof has two steps similar to that in \cite[Thm. 1]{cichella2015cooperative}. 
First, we consider a system of the form
\noindent
\begin{align}\label{eq:PE_aux_func}
    \dot{\chi} = - \frac{\alpha}{\gamma} \underline{\lap}_{\sigma(t)} \chi, 
    \qquad  \chi(t_0) \in \real^{N-1}, 
\end{align}
in which $\underline{\lap}_{\sigma(t)} = Q\lap_{\sigma(\tau)} Q^{\top}$ satisfies the $(\mu,T)$-PE condition in \eqref{eq:PE_cond}. 
It follows from \cite[lemma 1]{efimov2015design} that\footnote{We note \cite[lemma 1]{efimov2015design} has defined the function of the form $\underline{\lap}_{\sigma(t)}$ in $\eqref{eq:PE_aux_func}$ to be continuous. Yet, a unique solution to $\eqref{eq:PE_aux_func}$ exists (see \cite[Thm. 3.2]{khalil2002nonlinear}) in the case $\underline{\lap}_{\sigma(t)}$ is piecewise continuous and bounded with a finite set of point-wise discontinuities, and the results hold as stated herein.} \eqref{eq:PE_aux_func} is globally uniformly exponentially stable (GUES) with the convergence rate $\lambda_{ \chi} \in \realpos $ such that 
\noindent
\begin{align}\label{eq:PE_aux_func_conv_rate}
     \norm{ \chi(t) } \leq \kappa_{ \chi}  \norm{ \chi(0)} e^{-\lambda_{ \chi}t}, \ \ \forall \, t \in  \realnonneg,
\end{align}
in which $\kappa_{ \chi} = \sqrt{ \frac{\alpha N}{\gamma \lambda_{\chi}}}$ and $\lambda_{ \chi} = \eta e^{-2 \eta T}$ with $\eta = -\frac{1}{2T} \ln (1-\frac{(\alpha / \gamma) \mu T}{1+(\alpha / \gamma)^{2} N^{2} T^{2}})$, where $N$ is obtained from $ \norm{\lap_{\sigma(t)}} \leq N$ (see \cite[Corrollary 13.1.4]{godsil2001algebraic}).
Given the GUES of \eqref{eq:PE_aux_func} under condition \eqref{eq:PE_cond}, it follows from \cite[Lemma 1]{loria2002uniform} and \cite[Thm. 4.12]{khalil2002nonlinear} that there exists a Lyapunov function $v(t,\chi(t)) =\chi(t)^{\top}P(t)\chi(t) $ with $P(t)=P(t)^{\top} \in \real^{(N-1)\times (N-1)} > 0$ such that  $\forall \, t \in  \realnonneg$ the following inequalities hold:
\noindent
\begin{subequations}\label{eq:PE_aux_func_lyap}
\begin{align}\label{eq:PE_aux_func_lyap_1}
    0 <  \frac{\gamma}{2\alpha N} I_{N-1} \leq P(t)  \leq \frac{1}{2 \lambda_{ \chi}}  & I_{N-1} , 
    \\ \label{eq:PE_aux_func_lyap_2}
    \dot{P}(t) - \frac{\alpha}{\gamma} \underline{\lap}_{\sigma(t)} P(t) - \frac{\alpha}{\gamma} P(t) \underline{\lap}_{\sigma(t)} + & I_{N-1} = \zeros. 
\end{align}
\end{subequations}

In the second step, we use the stability properties of \eqref{eq:PE_aux_func} as given in \eqref{eq:PE_aux_func_conv_rate} and \eqref{eq:PE_aux_func_lyap} in the stability analysis of \eqref{eq:cl_sys}.
By defining an intermediary state  transformation as
\noindent
\begin{align}\label{eq:aux_var}
    \boldsymbol{\chi}
    =
    \begin{bmatrix}
       \xi \\  \vb
    \end{bmatrix}
    =
    \begin{bmatrix*}[l]
       \gamma I_{\s N-1} & Q \\ 
       \zeros_{\s N \times (N-1)} & I_{\s N} 
    \end{bmatrix*}
    \begin{bmatrix}
       \zeta \\  \vb
    \end{bmatrix}
    =C^{-1} {\Yc}, 
\end{align}
the system $\Sigma_{\sigt}$ in \eqref{eq:cl_sys} with \eqref{eq:output_coord_vec} can be rewritten as 
\noindent
\begin{subequations}\label{eq:cl_sys_new}
\begin{align}
\dot{\boldsymbol{\chi}}
&=
\begin{bsmallmatrix*}[l]
 -\frac{\alpha}{\gamma} Q{\lap}_{\sigma(t)}Q^{\top}  &  +\frac{\alpha}{\gamma} Q{\lap}_{\sigma(t)} \\
 -\frac{\alpha}{\gamma} {\lap}_{\sigma(t)}Q^{\top} & -(\gamma I_N - \frac{\alpha}{\gamma} {\lap}_{\sigma(t)})
\end{bsmallmatrix*}
\boldsymbol{{\chi}}
+
\begin{bmatrix*}[r]
  Q I_{\Ac}  \\ I_{\Ac} 
\end{bmatrix*}\ua, 
\\
{\Yc} &= C  \boldsymbol{{\chi}},
\end{align}
\end{subequations}
Associated with \eqref{eq:aux_var}-\eqref{eq:cl_sys_new}, a Lyapunov function is defined as $V(t,\boldsymbol{\chi}(t))=\xi^{\top} P(t)\xi + \frac{\beta}{2} \vb^{\top} \vb $,  
with $P(t)$ as in \eqref{eq:PE_aux_func_lyap} and $\beta \in \realpos$.
By taking the derivative of $V(t,\boldsymbol{\chi}(t))$ along the trajectories of \eqref{eq:cl_sys_new}
and
using \eqref{eq:PE_aux_func_lyap}, $\norm{\lap_{\sigma(t)}} \leq N $, and $\norm{Q} \leq 1$, 
we obtain
\noindent
\begin{align}\label{eq:PE_aux_func_lyap_diff}
      \dot{V}(t,\boldsymbol{\chi}(t)) \leq 
      -
    \begin{bsmallmatrix}
     \norm{\xi} \\ \norm{\vb}
    \end{bsmallmatrix}^{\top}
\overline{M}
    \begin{bsmallmatrix}
     \norm{\xi} \\ \norm{\vb}
    \end{bsmallmatrix}
    + 
   \max\{\lambda^{-1}_{\chi},\beta\} 
    \nonumber \\
   \begin{bsmallmatrix}
     \norm{\xi} \\ \norm{\vb}
    \end{bsmallmatrix}^{\top} \! \ones_{2} \norm{\ua},
    \;\; 
    \bar{M} =
    \begin{bsmallmatrix}
 1 
 &  
 -\frac{\alpha}{\gamma} ( \beta + \frac{1}{\lambda_{ \chi}} )\frac{ N}{2}
 \\
 -\frac{\alpha}{\gamma} ( \beta + \frac{1}{\lambda_{ \chi}} )\frac{ N}{2} & \beta(\gamma - \frac{\alpha}{\gamma} N )
    \end{bsmallmatrix}.   
\end{align}
%
Note that by selecting $ \lambda_{\x} < \lambda_{ \chi}$ and a sufficiently large $\gamma $ \rev{(e.g., $\gamma=\alpha N$, for $\alpha\geq1$)}, one can verify that
\eqref{eq:PD_aux} holds\footnote{{The determinant of \eqref{eq:PD_aux} yields a cubic function of $\gamma$, to which applying the Routh's stability criterion indicates the existence of one positive root.}}:
%
\noindent
\begin{align}\label{eq:PD_aux}
\hspace{-1ex}
   \overline{M}
    - 2 \lambda_{\x} 
   \begin{bsmallmatrix*}
   \frac{1}{2\lambda_{\s \chi}} & 0 \\ 0 & \frac{\beta}{2}
   \end{bsmallmatrix*}
 \!=\!
\begin{bsmallmatrix*}[l]
 (1-\frac{\lambda_{ \mb{x}}}{\lambda_{\s \chi}})
 &  
 -\frac{\alpha}{\gamma} ( \beta + \frac{1}{\lambda_{\s \chi}} )\frac{ N}{2}
 \\
 -\frac{\alpha}{\gamma} ( \beta + \frac{1}{\lambda_{\s \chi}} )\frac{ N}{2} & \beta(\gamma - \frac{\alpha}{\gamma} N - \lambda_{\x} )
    \end{bsmallmatrix*}
\! > \! 0 \, .
\end{align}
Then from \eqref{eq:PE_aux_func_lyap},  \eqref{eq:PE_aux_func_lyap_diff}, and \eqref{eq:PD_aux}, we obtain
\noindent
\begin{align*}
    \dot{V}(t,\boldsymbol{\chi}(t)) \! \leq \!
    - 2 \lambda_{\x} & {V}(t,\boldsymbol{\chi}(t)) +\!
    \sqrt{2}\max\{\lambda^{-1}_{\chi},\beta\} \norm{\boldsymbol{\chi}(t)} \! \norm{\ua},
\end{align*}
to which, applying the comparison lemma \cite[Lemma 3.4]{khalil2002nonlinear} and by considering \eqref{eq:PE_aux_func_lyap} and $V(t,\boldsymbol{\chi}(t))$ yields
\noindent
\begin{align}\label{eq:sol}
    \norm{\boldsymbol{\chi}(t)} \leq 
    {
    \sqrt{
    \frac{ \max \braces{\lambda^{-1}_{ \chi},\beta}}{ \min \braces{ \frac{\gamma}{\alpha N},\beta}}}
    }
    &\norm{\boldsymbol{\chi}(t_0)}
    e^{-\lambda_{\x }(t-t_0)}
    + \nonumber \\
    {
    \frac{ \max \braces{\lambda^{-1}_{ \chi},\beta}}{ \lambda^{}_{ \x} \min \braces{ \frac{\gamma}{2\alpha N},\frac{\beta}{2}}}
    }
    \sup_{ t_0 \leq t \leq T_d} 
    &\norm{\ua (t)} , \;
    \forall \, t \geq t_0 \in \realnonneg.
\end{align}
Now one can conclude from \eqref{eq:sol} that 
\eqref{eq:cl_sys_new} is input-to-state stable (ISS) in the case $\ua \neq \zeros$, provided $\sup_{ t_0 \leq t \leq T_d} \norm{\ua (t)} < \infty$ for every $T_d \in  [0,\, \infty)$. 
It then follows from 
\eqref{eq:sol} that 
\noindent
\begin{align}\label{eq:state_bound}
\hspace{-1ex}
    \norm{(\Yc)_{T_d}}_{\Lc_p}
    &\leq 
    \kappa_{\x}
    e^{-\lambda_{\x}(t-t_0)}
    \norm{{\x}(t_0)}
    +
    \kappa_{\mb{u}} 
    \norm{(\ua)_{T_d}}_{\Lc_p}, 
\end{align}
where we used \eqref{eq:output_coord_vec}, \eqref{eq:aux_var}, $\norm{{\Yc}(t_0)} \leq \norm{{\x}(t_0)}$, and $\kappa_{\x}$ and $\kappa_{\mb{u}}$ as given in \eqref{eq:state_bounds}. 
Given \eqref{eq:state_bound}, the finite-gain $\Lc_{p}$ stability of \eqref{eq:cl_sys_new} (equiv. $\Sigma_{\sigt}$ in \eqref{eq:cl_sys}) is concluded for every $\norm{\x(t_0)}\leq \infty$ and every $\Lc_{p e}$-bounded $\ua$,\cite[Thm. 5.1 and Corollary 5.1]{khalil2002nonlinear}.
Finally, to calculate the bounds in \eqref{eq:formation_consensus_expo}, note that 
\begin{align}\label{eq:avg_coord_var}
    Q^{\top} \zeta \overset{\eqref{eq:output_coord_vec}}{=} Q^{\top}Q \widetilde{\pb} \overset{\eqref{eq:PE_Q}}{=} 
\widetilde{\pb} - \ones^{}_{N} \boldsymbol{\rm p}_{\rm avg},
\ \
 \boldsymbol{\rm p}_{\rm avg} =(\frac{1}{N}\ones^{\top}_{N} \widetilde{\pb}),
\end{align}
and let $Q$ be partitioned as $Q=\begin{bmatrix}
    q_1\, |\, q_2 \,| \cdots |\, q_N
\end{bmatrix}$, where $q_i \in \real^{N-1}$. We also have $Q^{\top}Q=I_N-(1/N)\ones^{\top}_{N}\ones^{}_{N} \implies \norm{q_i}^{2}=1-1/N$ for every $i \in \{1, \dots, N\}$. Then, using \eqref{eq:avg_coord_var}, one can write for every $i,j \in \Vc$ that
\noindent
\begin{align*}
\left| \widetilde{\pb}_{i}(t) - \widetilde{\pb}_{j}(t) \right| 
& =
\left|
q^{\top}_{i} \zeta(t) - q^{\top}_{j} \zeta(t) \right| 
 \leq 
\norm{q^{\top}_{i} - q^{\top}_{j} } \norm{\zeta(t)} 
\nonumber \\
&\leq 
\sqrt{2} \norm{\Yc(t)} 
\leq  
\sqrt{2}
\kappa_{\x}
    e^{-\lambda_{\x}(t-t_0)}
    \norm{{\x}(t_0)},    
\end{align*}
where we used \eqref{eq:output_coord_vec}, $\norm{q^{\top}_{i} - q^{\top}_{j} } \leq 2 \norm{q_i}\leq 2 \sqrt{\paren{1-1/N}}\leq \sqrt{2}$, and \eqref{eq:state_bound} with $\ua=\zeros$. Similarly, one can obtain \eqref{eq:formation_consensus_vel_expo}. 
This concludes the proof.

\section{proof of Lemma \ref{lemma:net_level_obs}}\label{app_net_level_obs}
The equivalence of ($\rm i$) and ($\rm ii$) follows from the definition of \emph{state and input observability} (SIO) and the invariant zeros of the switched LTI systems: Recall that $\Sigma_{\sigt}$ in \eqref{eq:cl_sys_new} is an LTI system in each mode $\sigma \in \Qc$ under Assumptions \ref{assum:switching} and \ref{assum:comm_capability}. 
Without loss of generality, let $\sigma(t_k) = \mb{q} \in \Qc'$, for some $k \in \intgnonneg$. Then, it follows from \cite{boukhobza2007state} that if an LTI system, here $\Sigma_{\mb{q}}$, is SIO, then, any non-zero input $\ua(t)$ is observable at the output $\y^{\s \Vc\setminus\Ac}_{\mb{q}}$ in \eqref{eq:collective_measurements}.
It also follows from \cite[Ch. 3.11]{zhou1996robust} and \cite[Thm. 2]{boukhobza2007state} that the necessary and sufficient condition for $\Sigma_{\mb{q}}$ to be SIO (here, attack detectable) is that $\boldsymbol{P}(\lambda_{\rm o}, \sigma =\mb{q})$ in \eqref{eq:pencil_collective} is full column rank.
Thus, in a mode $ \mb{q} \in \Qc'$, for a $\ua(t)\neq\zeros$ to be unobservable at $\y^{\s \Vc\setminus\Ac}_{\mb{q}}$ (stealthy in the sense of \eqref{eq:undetectable} as characterized by \eqref{eq:undetect_chrz} in Lemma \ref{lemma:undetectable_chrz}), it is required that $\boldsymbol{P}(\lambda_{\rm o}, \sigma=\mb{q})$ is rank deficient, inducing an output-zeroing subspace 
such that 
$ \begin{bsmallmatrix}
        {\rm x}(t_k) \\ \ua(t_k)
\end{bsmallmatrix} \in \ker  \paren {\boldsymbol{P}(\lambda_{\rm o}, \sigma=\mb{q})} $ holds for some  $ \lambda_{\rm o} \in \cplx $ and some nontrivial initial conditions $  {\rm x}(t_k) \neq \zeros $. 
By construction, it follows for $\Sigma_{\sigt}$ in \eqref{eq:cl_sys_new} with finitely many switches over any given interval $[t_0, \, t_0+T)$,
that a non-zero input $\ua(t)$ stealthy in the sense of \eqref{eq:undetectable}, requires  $\cap_{\sigma \in  \Qc'} \ker \paren { \boldsymbol{P}(\lambda_{\rm o}, \sigma) } \neq \emptyset $ for some $\lambda_{\rm o} \in \cplx$.

The proof of statement ($\rm ii$) follows from a contradiction argument:
Assume for a nontrivial output-zeroing direction $\col({\rm x_0}, {\rm u_0})$, where $ {\rm x_0} \!\neq \! \zeros$, $ {\rm u_0} \! \neq \! \zeros $,
%
$ \exists \, \lambda_{\rm o} \! \in \! \cplx, 
\; \textrm{s.t.} \; 
\forall \, \mb{q} \in \Qc',
    \boldsymbol{P}(\lambda_{\rm o},\mb{q}) 
    \begin{bmatrix}
        {\rm x_0} \\ {\rm u_0}
    \end{bmatrix} \!=\! \zeros.$
Then, from \eqref{eq:cl_sys_new}, \eqref{eq:collective_measurements}, \eqref{eq:pencil_collective}, we have 
%
%
\noindent
\begin{subequations}\label{eq:joint_inpt_expnd}
\begin{align}\label{eq:joint_inpt_expnd_1st}
& \hspace{5em}
\exists \; \lambda_{\rm o} \in \cplx, 
\qquad \textrm{s.t.}  \qquad \forall \, \mb{q} \in \Qc',
\nonumber \\
& \hspace{12em}
\lambda_{\mathrm{o}}      
    \begin{bmatrix*}[l]
     {\mathrm{p}}^{\s \Vc \setminus \Ac}_0 \\ {\mathrm{p}}^{\s \Ac}_0 
    \end{bmatrix*} 
    = 
    \begin{bmatrix*}[l]
     {\mathrm{v}}^{\s \Vc \setminus \Ac}_0 \\ {\mathrm{v}}^{\s \Ac}_0 
    \end{bmatrix*},  
\\ \label{eq:joint_inpt_expnd_1st_lower}
& \hspace{-1em}
\begin{bsmallmatrix*}[r]
\lambda^2_{\mathrm{o}} I + \lambda_{\mathrm{o}} \gamma I + \alpha \lap^{\s \Vc \setminus \Ac}_{\mb{q}} 
& \alpha \lap^{\s \Vc \setminus \Ac,  \Ac}_{\mb{q}} \\ 
\alpha \lap^{\s \Ac,  \Vc \setminus \Ac }_{\mb{q}} & \lambda^2_{\mathrm{o}} I + \lambda_{\mathrm{o}} \gamma I + \alpha \lap^{\s \Ac}_{\mb{q}} 
\end{bsmallmatrix*}
\begin{bmatrix*}[l]
     {\mathrm{p}}^{\s \Vc \setminus \Ac}_0 \\ {\mathrm{p}}^{\s \Ac}_0 
\end{bmatrix*}
=
\begin{bmatrix*}
     \zeros^{\phantom{\s \Vc}} \\ I^{\phantom{\s \Ac}}_{\s |\Ac|} 
    \end{bmatrix*} \mathrm{u_0},
\\ 
& \hspace{8ex} \;\, \label{eq:joint_inpt_expnd_2nd}
\mb{C}^{\s \Vc \setminus \Ac}_{\mb{q}}
\col
\begin{pmatrix*}[l]
     {\mathrm{p}}^{\s \Vc \setminus \Ac}_0, & {\mathrm{p}}^{\s \Ac}_0, &
     {\mathrm{v}}^{\s \Vc \setminus \Ac}_0, &{\mathrm{v}}^{\s \Ac}_0 
\end{pmatrix*}
    = \zeros.
\end{align}
\end{subequations}
where we used ${\rm x_0} = \col({\rm p_0}, {\rm v_0})$, and the Laplacian matrix $\lap_{\mb{q}}$ partitioned such that the set of cooperative agents $\Vc \setminus \Ac$ comes first and the set of malicious agents $\Ac$ comes second.
Under Assumptions \ref{assum:switching} and \ref{assum:comm_capability} and for an $F$-total/$F$-total with the given bounds, it follows from Lemma \ref{lemma:kernel_C} that
\eqref{eq:joint_inpt_expnd_2nd} results in $ \forall \, \mb{q} \in \Qc' $, ${\rm{p}}^{\s \Vc \setminus \Ac}_0 = {\rm{v}}^{\s \Vc \setminus \Ac}_0 = \zeros, \; \ind(\adj^{}_{\mb{q}}) {\rm{p}}^{\s \Ac}_0 = \zeros$, where $\ind(\adj^{}_{\sigma})$, is given in \eqref{eq:ind_matrix}, with $\ker_{\sigma \in \Qc'} \ind(\adj^{}_{\sigma}) = \emptyset$ over $[t_0, \, t_0+T), \, \forall \,t_0 \in \realnonneg$, and ${\rm{v}}^{\s  \Ac}_0$ to be an arbitrary state vector. 
Therefore, $ {\rm{p}}^{\s \Ac}_0 = {\rm{v}}^{\s \Ac}_0 = \zeros$ is the only solution to \eqref{eq:joint_inpt_expnd}. 
Noting that
for any $ \lambda_{\rm 0} \in \cplx$ with a positive real part $\paren{ 
    \lambda^2_{\rm o} I + \lambda_{\rm o} \gamma I + \alpha \lap^{\s \Ac}_{\mb{q}} }$ is positive definite. It then follows from \eqref{eq:joint_inpt_expnd_1st_lower} that $\zeros = I^{\phantom{\s \Ac}}_{\s |\Ac|} {\rm u_0} \! \implies \! {\rm u_0}=\zeros $. This contradicts the assumption made at the beginning of this section, thereby concluding that statement (\rm{ii}) holds.

\section{proof of Proposition \ref{prop:couplingTerm_convergence}}\label{app_couplingTerm_convergence}
First, note that similar to the last part of the proof of Proposition \ref{prop:convergence_PE_like}, from \eqref{eq:state_bound} and \eqref{eq:avg_coord_var}, one can obtain
\begin{align}\label{eq:temp_inq}
    \norm {\widetilde{\pb}- \ones_N \boldsymbol{\rm p}_{\rm avg}} 
     \leq
     \kappa_{\x}  e^{-\lambda_{\x}(t-t_0)} 
     \norm{\x(t_0)}
    + \nonumber \\
    \kappa_{\mb{u}}
    \sup_{ t_0 \leq t \leq T_d} \norm{\ua(t)}, \;
     \forall \, t \geq t_0 \in \realnonneg.
\end{align}
Now, consider $\underline{\rho} = \alpha \begin{bmatrix}
    \dtilde{\lap}^{}_{\sigma} & \lap^{(23)}_{\sigma}
    \end{bmatrix}
\begin{bmatrix*}[l]
    \widetilde{\pb}_{\s \twohop{i}} \\ \widetilde{\pb}_{\s \Rc}
\end{bmatrix*} $ in $\rho(\x_{\s \Ic} , \x_{\s \Rc} )$ as given in \eqref{eq:sys_matrices_twohop}.
Also, note that from the definition of Laplacian matrix and the matrix decomposition in \eqref{eq:laplacian_decomposition}, we have $\begin{bmatrix}
\dtilde{\lap}^{}_{\sigma} & \lap^{(23)}_{\sigma}
\end{bmatrix} \ones = \zeros$ (i.e. the matrix is zero row-sum), where $\dtilde{\lap}^{}_{\sigma}$ is positive semi-definite, all the elements of $\lap^{(23)}_{\sigma}$ are either $0$ or $-1$, and the all-ones vector $\ones$ is of the mode-dependent dimension $\real^{\s |\twohop{i}|+|\Rc_i|}$. Then, one can write for $\underline{\rho}$ in \eqref{eq:sys_matrices_twohop} that
\noindent
\begin{align}\label{eq:temp_rho}
    \underline{\rho}
    &= -\alpha 
    \begin{bmatrix}
    \dtilde{\lap}^{}_{\sigma} & \lap^{(23)}_{\sigma}
    \end{bmatrix}
    \begin{bmatrix}
    \widetilde{\pb}_{\s \twohop{i}} \\ \widetilde{\pb}_{\s \Rc}
    \end{bmatrix}
    \nonumber 
    \\ &=
    -\alpha
    \begin{bmatrix}
    \dtilde{\lap}^{}_{\sigma} & \lap^{(23)}_{\sigma}
    \end{bmatrix}
    \paren{
    \begin{bmatrix}
    \widetilde{\pb}_{\s \twohop{i}} \\ \widetilde{\pb}_{\s \Rc}
    \end{bmatrix}
    - \ones \boldsymbol{\rm p}_{\rm avg}
    }.
\end{align}
Using \eqref{eq:temp_inq} and \eqref{eq:temp_rho}, we have

\noindent
\begin{align*}
    \norm{\rho (\x_{\s \Ic} , \x_{\s \Rc} ) }_{ }
    &\leq 
    \norm{\underline{\rho}}
    \leq
    \alpha
    \norm{ \begin{bmatrix}
    \dtilde{\lap}^{}_{\sigma} & \lap^{(23)}_{\sigma}
    \end{bmatrix} }_{ }
    \norm{ 
    \begin{bmatrix}
    \widetilde{\pb}_{\s \twohop{i}} \\ \widetilde{\pb}_{\s \Rc}
    \end{bmatrix}
    - \ones \boldsymbol{\rm p}_{\rm avg} }_{ }
    \\ & \leq 
    \alpha
    \norm{ \begin{bmatrix}
    \dtilde{\lap}^{}_{\sigma} & \lap^{(23)}_{\sigma}
    \end{bmatrix} }_{ }
    \norm{ 
    \widetilde{\pb}
    - \ones_{N} \boldsymbol{\rm p}_{\rm avg} }_{ }
    \\ & 
    \overset{\eqref{eq:temp_inq}}
    \leq
    \alpha
     \kappa_{\x}  e^{-\lambda_{\x}(t-t_0)} 
     \norm{\x(t_0)}
    +
    \nonumber \\ 
    & \qquad
    \alpha \kappa_{\mb{u}}
    \sup_{ t_0 \leq t \leq T_d} \norm{\ua(t)}, \;
     \forall \, t \geq t_0 \in \realnonneg.
\end{align*}
where we used $  \norm{ \begin{bmatrix}
    \dtilde{\lap}^{}_{\sigma} & \lap^{(23)}_{\sigma}
    \end{bmatrix} }_{ } \geq 1 $  that holds when the matrix is not all zero (i.e. when $\underline{\rho}$ exists). 
    Additionally, if $\ua(t) = \zeros$, then we have $\norm{\rho (\x_{\s \Ic} , \x_{\s \Rc} ) }_{}
    \leq \alpha
    {\kappa}_{\x} e^{-\lambda_{\x}(t-t_0)} \norm{\x(t_0)},
    \;   \forall \, t \geq t_0 \in \realnonneg$. This concludes the proof.

\section{proof of Proposition \ref{prop:observability_local}}\label{app_observablity_local}
Proof of {\ref{prop:observability_local_st1}}:
Given $\Phi_{\sigt}^{\, i}$ in \eqref{eq:2hop_info} and \eqref{eq:2hop_sys}, the observability of $\Sigma_{\s \Vc^{\dprime}_{i}}, \; \forall \, i \in \Vc$ directly follows from a PBH test for each active mode of $\sigma \in \Qc$.  
W.l.o.g., let $ \sigma(t) = \mb{q} \in \Qc$ denote an active mode. It then follows after some algebraic manipulation 
that $\rank \begin{bsmallmatrix}
      \lambda_{\rm o} I - {\mb{A}}^{\s \Ic}_{\mb{q}}  
      \\
      \mb{C}^{\s \Ic}_{\mb{q}} 
    \end{bsmallmatrix} = 2|\Vc^{\,i^{\dprime}}_{\mb{q}}|=2|\Ic_i|,  \; \forall \, \lambda_{\rm o} \in \cplx$. We refer to \cite{RB_PHDthesis2024} for details.

Proof of \ref{prop:observability_local_st2}: Recall that the dynamics in \eqref{eq:2hop_sys}-\eqref{eq:rest_sys} are a representation of \eqref{eq:cl_sys} from the $i$-th agent perspective, and the measurements $\y^{\, i}_{\sigma}$'s in \eqref{eq:2hop_sys} and \eqref{eq:locally_measurements} are the same set of measurements (see \eqref{eq:sys_matrices_twohop}). Therefore, this statement directly follows from Lemma \ref{lemma:net_level_obs} under Assumptions \ref{assum:switching} and \ref{assum:comm_capability} and for any $F$-total (resp. $F$-local) set $\Ac$ of malicious agents with $ 0 \leq F \leq \boldsymbol{\kappa}(\boldsymbol{\Gc}_{\s T}^{\mu})-1$ (resp. $ 0 \leq F \leq r(\boldsymbol{\Gc}_{\s T}^{\mu})-1$).
\section{Proof of Theorem \ref{thm:detectability_local}}\label{app_detectability_local}
Note that each agents' local attack detector $\Sigma^{\s \Oc}_{\s \Vc^{i^{\s \dprime}}_{\sigma}}$'s in \eqref{eq:2hop_obs_error}, have decoupled dynamics that are reinitialized based on \eqref{eq:2hop_obs_IC}. Therefore, without loss of generality, we consider the stability of one $\Sigma^{\s \Oc}_{\s \Vc^{i^{\s \dprime}}_{\sigma}}$ and start off with the proof of its input-to-state stability (ISS), in each mode $\sigma \in \Qc$.
From \eqref{eq:2hop_sys}, \eqref{eq:rest_sys}, \eqref{eq:sys_matrices_twohop}, and \eqref{eq:2hop_obs_error}, we have
\noindent
\begin{subequations}
    \begin{align*}
 \Sigma_{\sigma}:    \begin{bmatrix}
        \xdot_{\s \Ic} \\ {\xdot_{\s \Rc} } 
    \end{bmatrix}
    \!&=\!  
    \begin{bmatrix*}[l]   
{{\mb{A}}^{\s \Ic}_{\sigma}}+ 
\widetilde{\widetilde{{\mb{A}}^{\s \Ic}_{\sigma}}} 
& {{\mb{A}}^{\s \Ic, \Rc}_{\sigma}}   
       \\
  {{\mb{A}}^{\s \Rc, \Ic}_{\sigma}}   &  {{\mb{A}}^{\s \Rc}_{\sigma}} 
 \end{bmatrix*}
 \!
    \begin{bmatrix}
        \x_{\s \Ic} \\ {\x_{\s \Rc} }  
    \end{bmatrix} 
    \!+\!
    \begin{bmatrix*}[l]
        \Bapp   & \zeros \\ \zeros &  \Barest 
    \end{bmatrix*} 
    \!
   \begin{bmatrix}
         \uapp  \\  \uarest
    \end{bmatrix},
\nonumber \\ 
\rho(\x_{\s \Ic}, \x_{\s \Rc}) &=\! \begin{bmatrix}
      \widetilde{\widetilde{{\mb{A}}^{\s \Ic}_{\sigma}}} 
       & \phantom{,{{\mb{A}}^{\s \Ic}_{\sigma}}} 
       &  {{\mb{A}}^{\s \Ic, \Rc}_{\sigma}} 
\end{bmatrix}
\!
\begin{bmatrix}
        \x_{\s \Ic} \\ {\x_{\s \Rc} }  
    \end{bmatrix},
\\
\Sigma^{\s \Oc}_{\s \Vc^{i^{\s \dprime}}_{\sigma}} : \edot_{\s \Ic} &= {\bar{\mb{A}}^{\s \Ic}_{\sigma}} \e_{\s \Ic} + \rho(\x_{\s \Ic}, \x_{\s \Rc}) + \Bapp \uapp.
    \end{align*}
\end{subequations}
Let each mode $\sigma(t)=\mb{q} \in \Qc,\; \forall \, t \in [t_k, \; t_{k+1}),\; k \in \intgnonneg$. Then, we have from \eqref{eq:2hop_obs_error}, which is also appeared in last equation above, that
\noindent
\begin{align}\label{eq:error_solution}
        \e_{\s \Ic}(t) 
        &= 
         e^{\bar{\mb{A}}^{\s \Ic}_{\mb{q}}(t-t_k)} \e_{\s \Ic}(t_k) 
        + 
        \int_{t_k}^{t} e^{\bar{\mb{A}}^{\s \Ic}_{\mb{q}}(t-\tau)} 
        \rho (\x_{\s \Ic}, \x_{\s \Rc}) \,{\rm d} \tau \, +
        \nonumber  \\
        &  \hspace{7em} 
        \int_{t_k}^{t} e^{\bar{\mb{A}}^{\s \Ic}_{\mb{q}}(t-\tau)} \Bapp \uapp(\tau) \,{\rm d} \tau.
    \end{align}
%
Recall that $\bar{\mb{A}}^{\s \Ic}_{\mb{q}} $ in \eqref{eq:error_solution} is Hurwitz stable, as defined in \eqref{eq:2hop_obs}, 
ensuring the inequality $ \norm{ e^{\bar{\mb{A}}^{\s \Ic}_{\mb{q}}(t-t_k)} } \leq {\kappa}^{\s \Ic}_{\e} e^{- {\lambda}^{\s \Ic}_{\e}  (t-t_k)} $ holds for some constants $ {\kappa}^{\s \Ic}_{\e}, \, {\lambda}^{\s \Ic}_{\e} \in \realpos $ in each mode $\mb{q} \in \Qc$.
Moreover, in each mode, $ \e_{\s \Ic} \!=\! \zeros$ is the exponentially stable equilibrium point of the unforced system $\Sigma^{\s \Oc}_{\s \Vc^{i^{\s \dprime}}_{\sigma}}$ (i.e. no attack or coupling term perturbation).
We also have from Propositions \ref{prop:convergence_PE_like} and \ref{prop:couplingTerm_convergence} that the unknown input $\rho(\x_{\s \Ic}, \x_{\s \Rc})$ in \eqref{eq:2hop_obs_error} and \eqref{eq:error_solution} is $\Lc_{pe}$-bounded for any $\Lc_{pe}$-bounded $\norm{\x(t_0)}$ and any bounded input $ \begin{bmatrix}
         \uapp  \\  \uarest
    \end{bmatrix} \!=\! \ua \! \in \! \Lc_{p e}$ that are injected by an $F$-local$/F$-total set $\Ac$ with the given upper bounds. 
Then, from  \eqref{eq:error_solution}, we have
\noindent
\begin{align}\label{eq:error_solution_bound}
\hspace{-2pt}
      \norm{ \e_{\s \Ic}(t) }
\leq & \; 
    {\kappa}^{\s \Ic}_{\e}
    \norm{\e_{\s \Ic}(t_k)} e^{- {\lambda}^{\s \Ic}_{\e}  (t-t_k)}
       + 
       {\scriptstyle
    \frac{{\kappa}^{\s \Ic}_{\res}}{\lambda^{\s \Ic}_{\e}}} \norm{\x(t_0)} e^{-\lambda_{\x}(t_k-t_0)}
 \nonumber \\
    & \ \   
    (1 - e^{ - \lambda^{\s \Ic}_{\e} ( t -t_k ) }) 
    +
(
{\scriptstyle
\frac{1+ {\kappa}^{\s \Ic}_{\res}}{\lambda^{\s \Ic}_{\e}}})
\sup_{ t_0 \leq t \leq T_d} \norm{\ua(t)}
\nonumber\\
& \ \
   (1 - e^{ - \lambda^{\s \Ic}_{\e} ( t -t_k ) }), \ \
        \forall \,  t \geq t_k \geq t_0 \in \realnonneg,
    \end{align}
where we used $\norm{ \rho (\x_{\s \Ic}, \x_{\s \Rc}) } \leq  \alpha {\kappa}_{\x} e^{-\lambda_{\x}(t_k-t_0)} \norm{\x(t_0)} +  
    \alpha \kappa_{\mb{u}}
    \sup_{ t_0 \leq t \leq T_d} \norm{\ua(t)},
    \,  \forall \, t \geq t_k \geq t_0 \in \realnonneg$, with $T_d \in [0,\, \infty)$ from Proposition \ref{prop:couplingTerm_convergence}, $ \norm{\uapp} \leq  \norm{\ua}$, and ${\kappa}^{\s \Ic}_{\res} = \alpha  {\kappa}_{\x} {\kappa}^{\s \Ic}_{\e}$.
Noting that the first two terms in the right-hand side of \eqref{eq:error_solution_bound} are exponentially decreasing and that $\norm{\e_{\s \Ic}(t_k)}\leq \norm{\x(t_k)}$ when $\Vc^{i^{\s \dprime}}_{\sigma(t_{k})} \neq \Vc^{i^{\s \dprime}}_{\sigma(t_{k-1})}$ or $k=0$ (see \eqref{eq:2hop_obs_error}), it can be verified, along the same lines as in \cite[Lemma 4.6]{khalil2002nonlinear}, that each $\Sigma^{\s \Oc}_{\s \Vc^{i^{\s \dprime}}_{\sigma}}$ is ISS, and that an arbitrarily large $w_{\s \Ic} \in \realpos$ exist such that $ \norm{ \e_{\s \Ic}(t_{k}) }  < w _{\s \Ic} < \infty$ holds\footnote{Any $ \norm{ \e_{\s \Ic}(t_{k}) }  \leq \paren{1/{\kappa}^{\s \Ic}_{\e}} w $, with $  0 < w  < w_{\s \Ic} - ({\kappa}^{\s \Ic}_{\res}/\lambda^{\s \Ic}_{\e}) \norm{\x(t_0)}$ $ e^{-\lambda_{\x}(t_k-t_0)} $, and $\sup_{ t_0 \leq t \leq T_d} \norm{\ua(t)}
    \leq  
    {
    \frac{\lambda^{\s \Ic}_{\e} w }{1+ {\kappa}^{\s \Ic}_{\res}}} $ ensures $ \norm{ \e_{\s \Ic}(t_{k+1}) } \leq  w + ({\kappa}^{\s \Ic}_{\res}/\lambda^{\s \Ic}_{\e}) \norm{\x(t_0)} e^{-\lambda_{\x}(t_k-t_0)} < w_{\s \Ic} $ for \eqref{eq:error_solution_bound}.
    } for all $t_k$'s with $ \Vc^{i^{\s \dprime}}_{\sigma(t_{k+1})} = \Vc^{i^{\s \dprime}}_{\sigma(t_{k})} $.
Also, note that $\norm{\C^{\s \Ic}} = 1$ with its $j$-th row being $({\efrak^{\; j}_{\s 2|\Ic|}})^{\s \top}$, where $j \in \{1, \dots, |\Ic_i|+1\}$ (see \eqref{eq:sys_matrices_twohop_C}). Then, along the same lines as in \cite[Cor. 5.1, Thm. 5.3]{khalil2002nonlinear}, the finite-gain $\Lc_p$ stability of \eqref{eq:2hop_obs_error} with $ \res^{\, i}_{ \mb{q}}(t) = {\mb{C}}^{\s \Ic}_{\mb{q}} \e_{\s \Ic}(t) $ 
can be concluded from \eqref{eq:error_solution_bound} with the bound \eqref{eq:detectability_local_st2} for the $j$-th component of $\res^{\, i}_{ \mb{q}}(t)$.
Finally, if $\ua(t)=\zeros, \, \forall \, t \in \realnonneg$, we obtain from \eqref{eq:detectability_local_st2}, the bound in \eqref{eq:res_threshold}. 
\section{}
\label{app_aux_results}
%
The following result quantifies the inaccessible state measurements for the system in \eqref{eq:cl_sys} with the collective measurements in \eqref{eq:collective_measurements} obtained under Assumptions \ref{assum:switching} and \ref{assum:comm_capability}. It shows that the only states not accessible at any time for any cooperative agent are the velocity states of the set of malicious agents.
\begin{lemma}\label{lemma:kernel_C}
    Consider the system in \eqref{eq:cl_sys}-\eqref{eq:locally_measurements} over a time interval $ [t_0, \, t_{0}+T)$ under Assumptions \ref{assum:switching} and \ref{assum:comm_capability}. Let $\Ic_i= \Vc^{\, i^{\dprime}}_{\sigma}$ in \eqref{eq:locally_measurements} and let \eqref{eq:cl_sys} be subject to an $F$-total (resp. $F$-local) adversary set with $ 0 \leq F \leq \boldsymbol{\kappa}(\boldsymbol{\Gc}_{\s T}^{\mu})-1$ (resp. $ 0 \leq F \leq r(\boldsymbol{\Gc}_{\s T}^{\mu})-1$).
    Then, the nullspace of the matrix $\mb{C}^{\s \Vc \setminus \Ac}_{\sigma}$ in \eqref{eq:collective_measurements}, defined uniformly over $[t_0, \, t_{0}+T), \, \forall \,t_0 \in \realnonneg$ is given by
\begin{align}\label{eq:kernel_C}
\hspace{-1ex}
{\Xc}_{[t_0,\, t_0+T)} 
\!=\!
\cap_{\sigma \in \Qc'} \ker \mb{C}^{\s \Vc \setminus \Ac}_{\sigma} 
\!=
\spann \braces{ \begin{bsmallmatrix*}[l]
    \zeros_{N}  \\  \efrak^{\,i}_{\s N}
\end{bsmallmatrix*},\, \forall \, i \in \Ac }\!.
\end{align}
\end{lemma}

\begin{proof}
w.l.o.g., let $\x$ in \eqref{eq:cl_sys} be partitioned as 
$\x = \col (\widetilde{\pb}_{\s \Vc \setminus \Ac}, \widetilde{\pb}_{\s \Ac}, {\vb}_{\s \Vc \setminus \Ac}, {\vb}_{\s \Ac} ) \in \real^{2N} $. 
Also, from \eqref{eq:kernel_C} we have
\begin{align}\label{eq:temp_C_kernel}
\hspace{-1em}
\cap_{\sigma \in \Qc'} \ker \mb{C}^{\s \Vc \setminus \Ac}_{\sigma} 
=
\braces{ \x \in \real^{2N} \mid \mb{C}^{\s \Vc \setminus \Ac}_{\sigma}  \x \!=\! \zeros,  \;
    \forall \, \sigma \in \Qc' },\!
\end{align}
which by using \eqref{eq:locally_measurements}, \eqref{eq:collective_measurements}, and $\Ic_i= \Vc^{\, i^{\dprime}}_{\sigma}$ can be rewritten as 
\noindent
\begin{align}\label{eq:joint_inpt_expnd_2st_aux_more}
\forall \, \sigma \in \Qc', \;
    \begin{bsmallmatrix*}[l]
     I_{\s |\Vc \setminus \Ac|} & \zeros \\ \zeros & \ind(\adj^{}_{\sigma})\\
     \boldsymbol{\star} & \zeros
\end{bsmallmatrix*}\!
\begin{bmatrix*}[l]
     {\widetilde{\pb}}_{\s \Vc \setminus \Ac} \\ {\widetilde{\pb}}_{\s \Ac}
\end{bmatrix*}
    &\!=\! \zeros ,
    \;
    \begin{bsmallmatrix*}[l]
        I_{\s |\Vc \setminus \Ac|} & \zeros
\end{bsmallmatrix*} \!
\begin{bmatrix*}[l]
     {\vb}_{\s \Vc \setminus \Ac} \\ {\vb}_{\s \Ac}
\end{bmatrix*}
    \!=\! \zeros, 
\end{align}
where $\boldsymbol{\star}$ is a binary matrix
whose structure is immaterial.
$\ind(\adj^{}_{\sigma})$ is a binary matrix-valued function
defined as follows
\noindent
\begin{align}\label{eq:ind_matrix}
\hspace{-1.5ex}
\ind(\adj^{}_{\sigma}) \!=\!\!  \begin{bsmallmatrix}
   \! \ind_1(\adj^{\s \Vc \setminus \Ac, \Ac}_{\sigma})
    \\ \vdots \\
   \! \! \ind_{\s |\Vc \setminus \Ac|}(\adj^{\s \Vc \setminus \Ac, \Ac}_{\sigma}) \! \!
\end{bsmallmatrix}
\!\!=\!\!
\begin{bsmallmatrix}
          {\rm diag} (a^{\sigma(t)}_{i_1 j_1}, & \dots\,, & a^{\sigma(t)}_{i_1 j_{\s |\Ac|}}) 
        \\
    &  \vdots &
        \\
  \!\!  {\rm diag} (a^{\sigma(t)}_{i_{\s |\Vc\setminus\Ac|} j_1}, & \dots\,, & a^{\sigma(t)}_{i_{\s |\Vc\setminus\Ac|} j_{\s |\Ac|}}) \!\!
\end{bsmallmatrix}\!,
\end{align}
where $\adj^{\s \Vc \setminus \Ac, \Ac}_{\sigma} $, taken from $ \adj_{\sigma(t)} = \begin{bsmallmatrix*}[l]
    \adj^{\s \Vc \setminus \Ac}_{\sigma(t)} & \adj^{\s \Vc \setminus \Ac, \Ac}_{\sigma(t)} \\
    \adj^{\s \Ac, \Vc \setminus \Ac}_{\sigma(t)} &
    \adj^{\s \Ac}_{\sigma(t)}
\end{bsmallmatrix*}$, indicates the mode-dependent communication links between the sets $ \Vc\setminus\Ac = \braces{i_1, \dots, i_{\s |\Vc\setminus\Ac|}}  $ and $ \Ac = \braces{j_1, \dots, j_{\s |\Ac|}} $.
Under Assumptions \ref{assum:switching} and \ref{assum:comm_capability}, it follows from Proposition \ref{prop:graph_bounds} that 
in the case of an $F$-local (resp. $F$-total) adversary set with $F \leq r(\boldsymbol{\Gc}^{\mu}_{\s T})-1  $ (resp. $ F \leq \boldsymbol{\kappa}(\boldsymbol{\Gc}^{\mu}_{\s T})-1$), each malicious agent $j \in \Ac \subset \Vc$ should have at least one neighbor outside of its set, i.e. the set $\Vc \setminus \Ac$, over some period of time. Formally,
\noindent
\begin{align*}
&\forall \, j \in \Ac, \; \;  \exists \, i \in \Vc \setminus \Ac, \; \exists \, \sigma \in \Qc' \ \  \text{s.t.} \ \
j \in  \onehop{i} \subseteq \Vc^{\, i^{\dprime}}_{\sigma} 
\iff
\nonumber \\
& (i,j) \! \in \Ec_{\sigma}
\overset{\text{Lemma \ref{lemma:PE_equvalence}}}{\iff}
\frac{1}{T} \int_{t_0}^{t_0+T} \!\! a^{\sigma(\tau)}_{i j}  {\rm d}\tau \geq \delta,  \, \forall \, t_0 \in \realnonneg.
\end{align*}
by which, one can readily show for $\ind(\adj^{}_{\sigma})$ in \eqref{eq:ind_matrix}  that
\begin{align}\label{eq:some_bounds_ker}
  (1/T) \int_{t_0}^{t_0+T} \ind^{\top} (\adj^{}_{\sigma}) \ind (\adj^{}_{\sigma})  \, {\rm d}\tau 
  &\geq  \underline{\delta} I_{|\Ac|} , 
  \nonumber   \\
  \forall \, t_0 \in \realnonneg, \, \exists \, \underline{\delta} \in \realpos  
    \implies \cap_{\sigma \in \Qc'}
  &\ker \ind (\adj_{\sigma}) = \emptyset.
\end{align}
Then, from \eqref{eq:some_bounds_ker} and \eqref{eq:joint_inpt_expnd_2st_aux_more} we can conclude for \eqref{eq:temp_C_kernel} that $\forall \, \sigma \in \Qc'$, $\mb{C}^{\s \Vc \setminus \Ac}_{\sigma}  \x = 
\mb{C}^{\s \Vc \setminus \Ac}_{\sigma} \begin{bsmallmatrix*}[l]
     \zeros_{\s |\Vc \setminus \Ac|} \\ 
     \zeros_{\s |\Ac|} \\
     \zeros_{\s |\Vc \setminus \Ac|} \\ {\vb}_{\s \Ac}
\end{bsmallmatrix*}
    = \zeros,  \ \  \forall \, {\vb}_{\s \Ac} \in \real^{|\Ac|}$,
%
which in turn implies $ \cap_{\sigma \in \Qc'} \ker \mb{C}^{\s \Vc \setminus \Ac}_{\sigma} = \spann \braces{ \begin{bsmallmatrix*}[l]
    \zeros_{|\Vc|}  \\ \zeros_{|\Vc \setminus \Ac|} \\ \ones_{|\Ac|}
\end{bsmallmatrix*} }$. This concludes the proof.
\end{proof}
\bibliographystyle{IEEEtran}
\bibliography{IEEEabrv,references}
%
\end{document}